\documentclass[authorcolumns,numberwithinsect]{no-lipics-v2022}

\pdfoutput=1
\usepackage{enumitem}

\makeatletter

\newcommand{\homs}[2]{\mathsf{Hom}(#1 \to #2)}
\newcommand{\embs}[2]{\mathsf{Emb}(#1 \to #2)}
\newcommand{\auts}{\mathsf{Aut}}

\newcommand{\surhoms}[2]{\mathsf{SurHom}(#1 \to #2)}
\newcommand{\lovasz}{Lov{\'{a}}sz}

\newcommand{\subs}[2]{\mathsf{Sub}(#1 \to #2)}

\newcommand{\W}{\mathrm{W}}

\newcommand{\homsupp}{\chi}
\newcommand{\ppart}{\ensuremath{\mathsf{Part}}}

\newcommand{\zsig}[2]{\ensuremath{#1\lvert_{0\mapsto#2}}}

\newcommand{\parthomsX}[3]{\mathsf{Hom}^{\phi}_{#3}(#1 \to #2)}

\newcommand{\cphomsprob}{\text{\sc{cp-Hom}}}
\newcommand{\colhomsprob}{\text{\sc{col-Hom}}}
\newcommand{\listhomsprob}
{\text{\sc{list-Hom}}}

\newcommand{\fptred}{\leq^{\mathsf{FPT}}_{\mathsf{T}}}
\newcommand{\fptinterred}{\equiv^{\mathsf{FPT}}_{\mathsf{T}}}
\newcommand{\fptlinred}{\leq^{\mathsf{FPT-lin}}_{\mathsf{T}}}
\newcommand{\fptinterlinred}{\equiv^{\mathsf{FPT-lin}}_{\mathsf{T}}}

\def\fracture#1#2{\ensuremath{#1\raisebox{.2ex}{\rotatebox[origin=c]{-15}{$\sharp$}}#2}}

\newcommand{\holantprob}{\text{\sc{p-Holant}}}
\newcommand{\holantprobstar}{\text{\sc{p-Holant}}^{\mathsf{Hcol}}}

\newcommand{\homscp}{\mathsf{Hom}_{\mathsf{cp}}}
\newcommand{\embscp}{\mathsf{Emb}_{\mathsf{cp}}}

\newcommand{\cphoms}{\mathsf{cp}\text{-}\mathsf{Hom}}

\newcommand{\holant}{\mathsf{Holant}}

\def\fracture#1#2{\ensuremath{#1\raisebox{.2ex}{\rotatebox[origin=c]{-15}{$\sharp$}}#2}}

\usepackage{tikz, tikz-cd}
\usetikzlibrary{matrix}
\usetikzlibrary{shapes}
\usetikzlibrary{arrows,decorations.markings, tikzmark}
\usetikzlibrary{decorations.pathreplacing,calligraphy,backgrounds}
\usetikzlibrary{arrows.meta,arrows}
\usetikzlibrary{conference}
\usetikzlibrary{cd}
\usepackage{todonotes}

\newcommand*{\boldone}{\text{\usefont{U}{bbold}{m}{n}1}}

\newcommand{\he}{\mathsf{he}}  % half-edges

 \nolinenumbers

\title{Parameterised Holant Problems}
\author{Panagiotis Aivasiliotis}{Hasso Plattner Institute, University of Potsdam}{panos.aivasiliotis@hpi.de}{}{}

\author{Andreas Göbel}{Hasso Plattner Institute, University of Potsdam}{andreas.goebel@hpi.de}{}{}

\author{Marc Roth}{School of Electronic Engineering and Computer Science, Queen Mary University of London}{m.roth@qmul.ac.uk}{}{}

\author{Johannes Schmitt}{Department of Mathematics, ETH Zürich}{johannes.schmitt@math.ethz.ch}{}{}
\authorrunning{P. Aivasiliotis, A. Göbel, M. Roth, J. Schmitt}


\acknowledgements{The authors would like to thank Miriam Backens for their insightful feedback on this work. Andreas Göbel was funded by the project PAGES (project No. 467516565) of the German Research Foundation (DFG).
Johannes Schmitt was supported by the Swiss National Science Foundation (project No. 219369) and SwissMAP.}

\begin{document}
\maketitle

\begin{abstract}
    We investigate the complexity of parameterised holant problems $\textsc{p-Holant}(\mathcal{S})$ for families of symmetric signatures~$\mathcal{S}$. The parameterised holant framework has been introduced by Curticapean in 2015 as a counter-part to the classical and well-established theory of holographic reductions and algorithms, and it constitutes an extensive family of coloured and weighted counting constraint satisfaction problems on graph-like structures, encoding as special cases various well-studied counting problems in parameterised and fine-grained complexity theory such as counting edge-colourful $k$-matchings, graph-factors, Eulerian orientations or, more generally, subgraphs with weighted degree constraints. We establish an exhaustive complexity trichotomy along the set of signatures $\mathcal{S}$: Depending on the signatures, $\textsc{p-Holant}(\mathcal{S})$ is either
    \begin{itemize}
        \item[(1)] solvable in FPT-near-linear time, i.e., in time $f(k)\cdot \tilde{\mathcal{O}}(|x|)$, or
        \item[(2)] solvable in ``FPT-matrix-multiplication time'', i.e., in time $f(k)\cdot {\mathcal{O}}(n^{\omega})$, where $n$ is the number of vertices of the underlying graph, but not solvable in FPT-near-linear time, unless the Triangle Conjecture fails, or
        \item[(3)] $\#\W[1]$-complete and no significant improvement over the naive brute force algorithm is possible unless the Exponential Time Hypothesis fails. 
    \end{itemize}
    This classification reveals a significant and surprising gap in the complexity landscape of parameterised Holants: Not only is every instance either fixed-parameter tractable or $\#\W[1]$-complete, but additionally, every FPT instance is solvable in time (at most) $f(k)\cdot {\mathcal{O}}(n^{\omega})$. 
    We show that there are infinitely many instances of each of the types; for example, all constant signatures yield holant problems of type (1), and the problem of counting edge-colourful $k$-matchings modulo $p$ is of type ($p$) for $p\in\{2,3\}$. 

    Finally, we also establish a complete classification for a natural uncoloured version of parameterised holant problem $\textsc{p-UnColHolant}(\mathcal{S})$, which encodes as special cases the non-coloured analogues of the aforementioned examples. We show that the complexity of $\textsc{p-UnColHolant}(\mathcal{S})$ is different: Depending on $\mathcal{S}$ all instances are either solvable in FPT-near-linear time, or $\#\W[1]$-complete, that is, there are no instances of type (2). 
\end{abstract}

\newpage

\section{Introduction}
Inspired by Valiant's work on holographic algorithms~\cite{Valiant08}, the so-called \emph{holant framework}, first introduced in the conference version~\cite{CaiLX09} of~\cite{CaiLX14}, constitutes one of the most powerful and ubiquitous tools for the analysis of computational counting problems. Holants, defined momentarily, strictly generalise counting constraint satisfaction problems (``$\#\textsc{CSP}$s'') and are able to model various (in)famous counting problems such as counting perfect matchings, graph factors, Eulerian orientations, and proper edge-colourings~\cite{CaiG21} (see also~\cite{CaiL11}). Moreover, the holant framework has been used for the analysis of the complexity of computing partition functions from statistical physics (see e.g.\ \cite{CaiFX18,CaiF23}), and it has shown to allow for the application of tools from quantum information theory, particular of entanglement, to the analysis of counting problems~\cite{Backens18}.

Formally, an instance of a holant problem\footnote{We present here the case of signatures with Boolean domain, but we point out that more general versions have been studied (see e.g.\ \cite{CaiG21}).} is a pair of a graph~$G$ and an assignment from vertices~$v$ of~$G$ to \emph{signatures} $s_v$, where each $s_v$ is a function with values in algebraic complex numbers and with domain $\{0,1\}^{\mathsf{deg}(v)}$. The value of the holant on input $(G,\{s_v\}_{v\in V(G)})$ is then defined as
\begin{equation}\label{eq:intro_classical_holant}
    \sum_{\alpha: E(G) \to \{0,1\}}~ \prod_{v\in V(G)} s_v(\alpha|_{E(v)}) \,,
\end{equation}
where $\alpha|_{E(v)}$ is the restriction of $\alpha$ on the edges incident to $v$. For example, let $G$ be a $d$-regular graph, and set
for each $v\in V(G)$ of degree $d$ the signature $s_v$ as the $d$-ary function $\mathsf{hw}^d_{=1}$ that outputs~$1$ if precisely one edge incident to $v$ is mapped by~$\alpha$ to~$1$, and~$0$ otherwise. Then~\eqref{eq:intro_classical_holant} is equal to the number of perfect matchings of $G$, that is, for the signature $\mathsf{hw}^d_{=1}$, the holant problem is equivalent to counting perfect matchings in $d$-regular graphs.

Since their inception in 2009, the holant framework has seen immense success in the quest of charting the complexity landscape of computational counting problems~\cite{CaiL11,CaiLX14,CaiGW16,LinW17,Backens18}. In a majority of the previous works, the central question was to determine the complexity of evaluating the Holant~\eqref{eq:intro_classical_holant} depending on the allowed signatures. For example, if each $s_v$ is the signature of having an even number of $1$s, then \eqref{eq:intro_classical_holant} can be computed by counting the number of solutions to a system of linear equations over $\mathbb{Z}/2\mathbb{Z}$, which can be done in polynomial time; this approach generalises to families of \emph{affine signatures}~\cite{ParityHolant13}. On the other hand, if we allow the signatures $\mathsf{hw}^i_{=1}$ for $i\in \mathbb{N}$, then the holant problem becomes $\#\mathrm{P}$-hard\footnote{$\#\mathrm{P}$ is the class of all (counting) problems polynomial-time reducible to $\#\textsc{SAT}$, the problem of counting satisfying assignments of a Boolean formula. By a result of Toda~\cite{Toda91} $\#\mathrm{P}$-hard problems are at least as hard as all problems in the polynomial-time hierarchy $\mathrm{PH}$.} since it is at least as hard as the $\#\mathrm{P}$-complete problem of counting perfect matchings~\cite{Valiant79,Valiant79b}.

It has turned out that, in numerous settings, the complexity of evaluating holants is either solvable in polynomial time, or $\#\mathrm{P}$-hard, and the dichotomy criterion only depends on the set of allowed signatures, that is, there are no instances of intermediate complexity.\footnote{In contrast, by (the counting version of) a result of Ladner~\cite{Ladner75}, assuming that $\mathrm{FP}\neq \#\mathrm{P}$, there are counting problems in $\#\mathrm{P}$ that are neither solvable in polynomial time, nor $\#\mathrm{P}$-hard.} Among others, key results include, but are not limited to, the complete complexity dichotomy for Boolean symmetric holants~\cite{CaiGW16}, for non-negative real-valued holants~\cite{LinW17}, and for holants with constant unary signatures~\cite{Backens18}.

\subsection{Parameterised Complexity Meets Holant Problems}
Introduced by Downey and Fellows in the early 90s~\cite{DowneyF92}, the field of parameterised complexity theory, also called multivariate algorithmics, investigates the complexity of computational problems not only along the size $|x|$ of an input $x$, but also along a parameter $k=\kappa(x)$ taking into account one or more (structural) properties of $x$. The notion of efficient algorithms is then relaxed from polynomial-time algorithms to fixed-parameter tractable (FPT) algorithms, which are required to run in time $f(k)\cdot |x|^{O(1)}$, where $f$ can be any computable function. The notion of fixed-parameter tractability formalises the existence of efficient algorithms if the parameter of the problem input is significantly smaller than the input size. For example, in the database query evaluation problem, the input is a pair $x=(\varphi,D)$ of a query $\varphi$ and a database $D$. Choosing the size of the query as a parameter, i.e., setting $k=|\varphi|$, an FPT algorithm for this problem can then be thought as an efficient algorithm for instances in which the size of the query is significantly smaller than the size of the database, which is true for realistic instances. 

Independently introduced by McCartin~\cite{McCartin06} and by Flum and Grohe~\cite{FlumG04}, the field of parameterised \emph{counting} complexity theory aims to apply the tools and methods from parameterised algorithmics to computational counting problems. In the context of counting problems, the notion of tractability is naturally
still given by FPT algorithms, and the notion of intractability is given by $\#\W[1]$-hardness. In a nutshell, the class $\#\W[1]$ can be thought as a parameterised equivalent of $\#\mathrm{P}$ and its canonical complete problem is the problem of, given a positive integer $k$ and a graph $G$, counting the number of $k$-cliques in $G$, parameterised by $k$. Under standard assumptions such as the Exponential Time Hypothesis, $\#\W[1]$-hard problems are not fixed-parameter tractable (see e.g.\ \cite{Chenetal05,Chenetal06,CyganFKLMPPS15}).

In addition to early key results such as the classification of the parameterised homomorphism counting problem by Dalmau and Jonson~\cite{DalmauJ04}, and the resolution of the complexity of the problem of counting $k$-matchings by Curticapean~\cite{Curticapean13}, the field of parameterised counting recently witnessed significant advancements with the introduction of the framework of graph motif parameters~\cite{CurticapeanDM17} which has subsequently been used to fully resolve the parameterised and fine-grained complexity of numerous network pattern counting problems~\cite{Roth17,DellRW19,RothSW20,BressanR21,BeraGLSS22,PenaS22,BLR2023stoc,GishbolinerLSY23,DoringMW24,Curticapean24,CurticapeanN24,DoringMW25}.

\paragraph*{Parameterised Holant Problems}
In the present work we investigate holant problems under the lens of parameterised (counting) complexity theory. In fact, parameterised holant problems have already been introduced and used almost a decade ago~\cite{Curticapean15}, but no attempt has been made to establish exhaustive classification results comparable to the classical holant dichotomies. 

Let us now introduce the parameterised holant framework following the approach of~\cite{Curticapean15}. To this end, let $\mathcal{S}$ be a finite set of symmetric signatures --- think for now of a signature in $\mathcal{S}$ as a function from $\{0,1\}^\ast$ to algebraic complex numbers; we will see later that we have to be very careful about which functions we can allow as signatures. A $k$-\emph{edge-coloured signature grid} over $\mathcal{S}$ then consists of a graph $G$, a (not necessarily proper) edge-colouring $\xi:E(G)\to [k]$, and an assignment from vertices $v\in V(G)$ to signatures in $\mathcal{S}$. We write $s_v\in \mathcal{S}$ for the signature assigned to vertex $v$. We say that an assignment $\alpha:E(G) \to \{0,1\}$ is \emph{edge-colourful} if $\alpha$ maps exactly $k$ edges to $1$, and each edge colour (w.r.t.\ $\xi$) is hit exactly once. The holant value of the $k$-edge-colourful signature grid $\Omega=(G,\xi,\{s_v\}_{v\in V(G)})$ is then defined as
\[ \holant(\Omega)= \sum_{\substack{\alpha: E(G)\to \{0,1\}\\\text{edge-colourful}}} ~\prod_{v\in V(G)} s_v(\alpha|_{E(v)}) \,.\]

For a set of signatures $\mathcal{S}$, the problem $\holantprob(\mathcal{S})$ gets as input a $k$-edge-coloured signature grid $\Omega$ over $\mathcal{S}$ and outputs $\holant(\Omega)$. The problem parameter is $k$. Similarly as in the case of classical Holants, our goal is to identify precisely the signature sets $\mathcal{S}$ for which $\holantprob(\mathcal{S})$ becomes tractable. However, we ask for fixed-parameter tractability, rather than for polynomial-time tractability, that is, our goal is the construction of algorithms running in time $f(k)\cdot |\Omega|^{O(1)}$. Analogously, we define the (arguably more natural)\footnote{While the uncoloured parameterised Holant problem is most likely more natural to readers outside of the field of parameterised counting, we decided to follow the notation of Curticapean~\cite{Curticapean15} and denote the coloured Holant problem by $\holantprob(\mathcal{S})$. To avoid any confusion we thus denote the uncoloured Holant problem by $\text{\sc{p-UnColHolant}}(\mathcal{S})$.} \emph{uncoloured} version $\text{\sc{p-UnColHolant}}(\mathcal{S})$, in which the signature grid does not come with a $k$-edge-colouring, and we instead sum over all $\alpha:E(G)\to\{0,1\}$ that map exactly $k$ edges to $1$. We will see in the applications of our main results that both $\holantprob(\mathcal{S})$ and $\text{\sc{p-UnColHolant}}(\mathcal{S})$ encode various well-studied parameterised counting problems, such as counting (coloured or uncoloured) $k$-matchings, counting (coloured or uncoloured) $k$-factors, and, to some extent, counting the number of weight-$k$ solutions of systems of linear equations.

As indicated previously, before stating our main results, we have to discuss some subtleties about the definition of signatures. In classical holant theory, each signature has a fixed arity $d$ and only allows for inputs in $\{0,1\}^d$. This implies that only vertices of degree $d$ can be equipped with such a signature. As a consequence, if we would also enforce a fixed arity for each signature, then for any \emph{finite} set of signatures $\mathcal{S}$, the only valid inputs to $\holantprob(\mathcal{S})$ and $\text{\sc{p-UnColHolant}}(\mathcal{S})$ would be signature grids, the underlying graphs of which have maximum degree $d$, where $d$ equals the maximum arity of the signatures in $\mathcal{S}$. Using standard tools from parameterised algorithmics, such as the bounded search-tree paradigm, the problems $\holantprob(\mathcal{S})$ and $\text{\sc{p-UnColHolant}}(\mathcal{S})$ would then always be fixed-parameter tractable. At the same time, modelling key parameterised counting problems such as counting $k$-matchings would require an infinite set of signatures. Therefore, we leverage the setting by allowing signatures to be functions with domain $\{0,1\}^\ast$ --- for example, setting $\mathsf{hw}_{\leq 1}$ to be the signature that maps an input tuple to $1$ if and only if at most one of its elements is $1$, the problems $\holantprob(\{\mathsf{hw}_{\leq 1}\})$ and $\text{\sc{p-UnColHolant}}(\{\mathsf{hw}_{\leq 1}\})$ become the problems of counting edge-colourful and uncoloured $k$-matchings, respectively. Moreover, in this work, we will restrict ourselves to \emph{symmetric} signatures, that is, signatures $s$ satisfying $s(x)=s(y)$ whenever $x$ can be obtained from $y$ by permuting its entries. Since the domain of signatures is $\{0,1\}^\ast$, the value of $s(x)$ only depends on the Hamming weight, i.e., the number of $1$s, in $x$. For notational convenience, we thus define signatures as functions from $\mathbb{N}$ to algebraic complex numbers.

\begin{definition}[Signatures]\label{def:signatures_intro}
    A \emph{signature} is a computable function $s$ from non-negative integers to algebraic complex numbers. Moreover, we require $s(0)\neq 0$.
\end{definition}
We note that the requirement $s(0)\neq 0$ makes sure that a vertex $v$ equipped with $s$ is allowed to not ``participate'' in the $k$-edge-subset chosen by $\alpha$; 
% if all vertices would be required to participate in the chosen edge-subsets, then the number of vertices of the signature grid must either be bounded by $2k$, or the holant value is $0$. 
if more than $2k$ vertices were equipped with signatures that allow for $s(0)=0$, then the holant value would be trivially $0$.
As the core of our technical work applies to signatures that fulfil the $s(0)\neq0$ requirement, we focus our discussion on such signatures.
% We exclude, for now, those degenerate instances, since, in the spirit of parameterised complexity, we wish to assume that the signature grid is significantly larger than $k$. 
However, for completeness, we show in Section~\ref{sec:sig0} how all our results can be extended to the case in which signatures $s$ with $s(0)=0$ are allowed.

Now, with this definition of signatures, we can simplify the notation in the definition of parameterised holant problems as follows; for the remainder of the paper, we will use the subsequent notation:

\begin{definition}[The Parameterised Holant Problem]\label{def:intro_param_holant}
    Let $\mathcal{S}$ be a finite set of signatures. The problem $\holantprob(\mathcal{S})$ gets as input a
    $k$-edge-coloured signature grid $\Omega=(G,\xi,\{s_v\}_{v\in V(G)})$ with $s_v\in \mathcal{S}$ for all $v\in V(G)$. The output is
    \[ \holant(\Omega):= \sum_{\substack{A \subseteq E(G)\\|A|=k,~\xi(A)=[k]}}~\prod_{v\in V(G)} s_v(|A \cap E(v)|)\,, \]
    where $E(v)$ denotes the set of edges incident to $v$. The problem parameter is $k$.
\end{definition}
\begin{definition}[The Parameterised Uncoloured Holant Problem]\label{def:intro_param_uncol_holant}
    Let $\mathcal{S}$ be a finite set of signatures. The problem $\text{\sc{p-UnColHolant}}(\mathcal{S})$ gets as input a positive integer $k$, and a signature grid $\Omega=(G,\{s_v\}_{v\in V(G)})$ with $s_v\in \mathcal{S}$ for all $v\in V(G)$. The output is
    \[ \holant(\Omega,k):= \sum_{\substack{A \subseteq E(G)\\|A|=k}}~\prod_{v\in V(G)} s_v(|A \cap E(v)|)\,. \]
    The problem parameter is $k$.
\end{definition}

\subsection{Our Contributions}
We provide a complexity ``trichotomy'' for $\holantprob(\mathcal{S})$, and a complexity dichotomy for $\text{\sc{p-UnColHolant}}(\mathcal{S})$, both along the permitted signatures $\mathcal{S}$. 

For the statement of our results, we first introduce signature fingerprints and types of signature sets.

\begin{definition}[Signature Fingerprints]\label{def:fingerprint_intro}
Let $d$ be a positive integer and let $s$ be a signature. The \emph{fingerprint} of $d$ and $s$ is defined as
    \[\chi(d,s) := \sum_{\sigma} (-1)^{|\sigma|-1} (|\sigma|-1)! \cdot \prod_{B \in \sigma} \frac{s(|B|)}{s(0)} \,,\]
    where the sum is over all partitions of $[d]$.
\end{definition}

We point out that $\chi(d,s)$ is equal to a weighted sum over all evaluations of the M\"obius function of the lattice of partitions of the set $[d]$; the details necessary for this work are provided in Section~\ref{sec:parts_quots_mobius} and we refer the reader to e.g.\ \cite[Chapter 3]{Stanley11} for further reading.

\begin{definition}[Types of Signature Sets]\label{def:sigtype_intro}
    Let $\mathcal{S}$ be a finite set of signatures. We say that $\mathcal{S}$ is of type
    \begin{itemize}
        \item[(I)] $\mathbb{T}[\mathsf{Lin}]$ if $\chi(d,s)=0$ for all $s\in \mathcal{S}, d\geq 2$, 
        \item[(II)] $\mathbb{T}[\omega]$ if $\chi(d,s)=0$ for all $s\in \mathcal{S}, d\geq 3$, but there exists $s\in \mathcal{S}$ with $\chi(2,s)\neq 0$, and
        \item[(III)] $\mathbb{T}[\infty]$ otherwise, i.e., there exists $s\in \mathcal{S}$ and $d\geq 3$ such that $\chi(d,s)\neq 0$.
    \end{itemize}
\end{definition}

We point out that there are infinitely many signature sets of each type --- we provide an easy construction in the appendix. For now, let us discuss some natural examples which also help us gain some initial understanding of the properties of the signature fingerprints.
\begin{itemize}
    \item[(I)] Unsurprisingly, any signature set containing only a constant function $s(x)=c\neq 0$ for all $x\in \mathbb{N}$ makes the problem easy.  In that case, we have, for each $d$, that
    $\chi(d,s) := \sum_{\sigma} (-1)^{|\sigma|-1} (|\sigma|-1)!$, since $s(|B|)/s(0)$ will always be $1$. However, as indicated previously, this alternating sum is equal to $\sum_{\rho}\mu(\rho,\top)$, where $\mu$ is the M\"obius function of the partition lattice, and $\top$ is the coarsest partition of the $d$-element set. This sum is well-known to be $0$ for every $d\geq 2$ (see e.g.\ \cite[Section~3.7]{Stanley11}), implying that $\mathcal{S}=\{s_c\}$ is indeed of type $\mathbb{T}[\mathsf{Lin}]$. Analogously, any finite signature set $\mathcal{S}$ containing only constant functions is of type $\mathbb{T}[\mathsf{Lin}]$ as well.

    As a slightly more interesting example, a similar argument applies to the function $s(x):= 2^{ax+b}$ for any pair of rational numbers $a$ and $b$, since $s(|B|)/s(0)= 2^{ax}$, implying that $\prod_{B\in \sigma} s(|B|)/s(0) = 2^{a\sum_{B\in \sigma}|B|} = 2^{ad}$, and thus $\chi(d,s) := 2^{ad} \cdot \sum_{\sigma} (-1)^{|\sigma|-1} (|\sigma|-1)! =0$, for $d \geq 2$.
    \item[(II)] For type $\mathbb{T}[\mathsf{\omega}]$, we consider the evaluation of $\holantprob(\mathcal{S})$ modulo $2$ (in Section~\ref{sec:modular} we show how to lift our results on the edge-coloured holant problem to modular counting). Then we set $\mathcal{S}=\{\mathsf{hw}_{\leq 1}\}$, and we recall that $\mathsf{hw}_{\leq 1}(x)$ evaluates to $1$ if $x\in \{0,1\}$ and to $0$ otherwise. Then the holant problem is precisely the problem of counting edge-colourful $k$-matchings modulo $2$. Now note that we have $(|\sigma|-1)!=0$ modulo $2$ whenever $|\sigma| \geq 3$, and, for $s=\mathsf{hw}_{\leq 1}$, $s(|B|)=0$ whenever $|B|\geq 2$. Therefore, for any $d\geq 3$, we have that $\chi(d,s)=0$ modulo $2$.

    However, note that $\chi(2,s)=1$ modulo $2$: The set $[2]$ only has two partitions $\bot_2=\{\{1\},\{2\}\}$ and $\top_2=\{\{1,2\}\}$, and observe that $\top_2$ contains a block $B$ of size $2$, hence $s(|B|)$ and the contribution of $\top_2$ to $\chi(2,s)$ vanishes. Therefore 
    \[ \chi(2,s) = (-1)^{|\bot_2|-1}(|\bot_2|-1)! \prod_{B\in \bot_2} \frac{s(|B|)}{s(0)} = 1 \mod 2 \,. \]
    Thus $\mathcal{S}=\{\mathsf{hw}_{\leq 1}\}$ is indeed of type $\mathbb{T}[\mathsf{\omega}]$ when computation is done modulo $2$.
    \item[(III)] A natural example for type $\mathbb{T}[\mathsf{\infty}]$ is the signature $\mathsf{even}(x)$ which maps $x$ to $1$ if $x$ is even, and to $0$ otherwise. Let us fix $d=4$ and set $s=\mathsf{even}$. Observe that \[\prod_{B\in \sigma}\frac{s(|B|)}{s(0)} = \begin{cases}
        1 & \forall B\in \sigma: |B|=0 \mod 2\\
        0 & \text{otherwise}
    \end{cases}\]
    There are precisely four partitions of $[4]$ that only contain even sized blocks: $\{\{1,2\},\{3,4\}\}$, $\{\{1,3\},\{2,4\}\}$, $\{\{1,4\},\{2,3\}\}$, and $\{\{1,2,3,4\}\}$. The former three each contribute $(-1)^{2-1} (2-1)! = -1$ to $\chi(4,s)$, and the latter contributes $(-1)^{1-1}(1-1)!=1$ to $\chi(4,s)$. Thus $\chi(4,s)=-3+1=-2\neq 0$ and $\mathcal{S}=\{\mathsf{even}\}$ is, as promised, of type $\mathbb{T}[\mathsf{\infty}]$. 
\end{itemize}

In some cases, e.g.\ for signatures with co-domain $\{0,1\}$, it will be possible to simplify the definition of the types. Moreover, we will show that $\mathbb{T}[\mathsf{Lin}]$ can always be simplified if $s(0)=1$ for each signature (and we will later see that we can always make this assumption without loss of generality).
\begin{lemma}\label{lem:simplify_type_1_intro}
Let $\mathcal{S}$ be a finite set of signatures such that $s(0)=1$ for all $s\in \mathcal{S}$. Then
$\mathcal{S}$ is of type $\mathbb{T}[\mathsf{Lin}]$ if and only if, for each $s\in \mathcal{S}$ and $n\in \mathbb{N}$, we have $s(n)=s(1)^n$.\qed
\end{lemma}

We are now able to state our main results. Some of the lower bounds rely on two well-known hardness assumptions from fine-grained complexity theory: The Exponential Time Hypothesis (ETH), and the Triangle Conjecture. We introduce both in Section~\ref{sec:prelims_fgct}; in a nutshell, ETH asserts that $3$-$\textsc{SAT}$ cannot be solved in sub-exponential time in the number of variables, and the Triangle conjecture asserts that there is no linear-time algorithm for finding a triangle in a graph.  
\begin{mtheorem}[Complexity Trichotomy for $\textsc{p-Holant}(\mathcal{S})$]\label{main_thm}
    Let $\mathcal{S}$ be a finite set of signatures.
    \begin{itemize}
        \item[(I)] If $\mathcal{S}$ is of type $\mathbb{T}[\mathsf{Lin}]$, then $\textsc{p-Holant}(\mathcal{S})$ can be solved in FPT-near-linear time, that is, there is a computable function $f$ such that $\textsc{p-Holant}(\mathcal{S})$ can be solved in time $f(k)\cdot \tilde{\mathcal{O}}(|V(\Omega)|+|E(\Omega)|)$.
        \item[(II)] If $\mathcal{S}$ is of type $\mathbb{T}[\omega]$, then $\textsc{p-Holant}(\mathcal{S})$ can be solved in FPT-matrix-multiplication time, that is, there is a computable function $f$ such that $\textsc{p-Holant}(\mathcal{S})$ can be solved in time $f(k)\cdot \mathcal{O}(|V(\Omega)|^{\omega})$. Moreover, $\textsc{p-Holant}(\mathcal{S})$ cannot be solved in time $f(k)\cdot \tilde{\mathcal{O}}(|V(\Omega)|+|E(\Omega)|)$ for any function $f$, unless the Triangle Conjecture fails.
        \item[(III)] Otherwise, that is, if $\mathcal{S}$ is of type $\mathbb{T}[\infty]$, $\textsc{p-Holant}(\mathcal{S})$ is $\#\W[1]$-complete. Moreover, $\textsc{p-Holant}(\mathcal{S})$ cannot be solved in time $f(k)\cdot |V(\Omega)|^{o(k/\log k)}$ for any function $f$, unless the Exponential Time Hypothesis fails.  \qed
    \end{itemize}
\end{mtheorem}
At the time of writing this paper the matrix multiplication exponent $\omega$ is known to be bounded by $2\leq \omega\leq 2.371552$~\cite{WilliamsXXZ24}.
Note that the lower bound in (III) matches, up to a factor of $1/\log k$ in the exponent, the running time of the brute force algorithm --- which runs in time $|\Omega|^{\mathcal{O}(k)}$ ---- making this bound (almost) tight. Moreover, the factor $1/\log k$ is not an artifact of our proofs, but a consequence of the notoriously open problem of whether ``you can beat treewidth''~\cite{Marx10}.

We emphasise that our complexity trichotomy above is \emph{much} stronger than an FPT vs $\#\W[1]$ classification result: Under ETH and the Triangle Conjecture, the best possible exponent of $|\Omega|$ in the running time is either $1$, or between $1$ and $\omega$, or lower bounded by $o(k/\log k)$. In particular, there are no FPT instances requiring an exponent larger than $\omega$.

Perhaps surprisingly, we discover that the complexity changes for the uncoloured holant problem; we will see in Section~\ref{sec:uncoloured} that $\text{\sc{p-UnColHolant}}(\mathcal{S})$ appears to be \emph{strictly} harder than $\textsc{p-Holant}(\mathcal{S})$.

\begin{mtheorem}[Complexity Dichotomy for $\text{\sc{p-UnColHolant}}(\mathcal{S})$]\label{thm:main_uncol}
    Let $\mathcal{S}$ be a finite set of signatures. 
    \begin{itemize}
        \item[(I)] If $\mathcal{S}$ is of type $\mathbb{T}[\mathsf{Lin}]$, then $\text{\sc{p-UnColHolant}}(\mathcal{S})$ can be solved in FPT-near-linear time.
        \item[(II)] Otherwise $\text{\sc{p-UnColHolant}}(\mathcal{S})$ is $\#\W[1]$-complete. If, additionally, $\mathcal{S}$ is of type $\mathbb{T}[\infty]$, then $\text{\sc{p-UnColHolant}}(\mathcal{S})$ cannot be solved in time $f(k)\cdot |V(\Omega)|^{o(k/\log k)}$, unless ETH fails. \qed
    \end{itemize}
\end{mtheorem}

\paragraph*{Explicit Tractability Criteria}
For both of our main classifications, the tractability criteria appear to be implicit and hard to verify in the sense that it might be non-trivial to decide for a concrete set of signatures whether it yields a tractable instance of a parameterised Holant problem.

We address this issue by using Lemma~\ref{lem:simplify_type_1_intro}, and we obtain an explicit classification for signatures satisfying $s(0)=1$; intuitively, in those instances the vertices not incident to any of the chosen edges do not contribute to the holant value. We provide below the corresponding dichotomy for the uncoloured Holant problem; in combination with the assumption $s(0)=1$ this variant accounts for the arguably most natural instantiation of parameterised Holant problems, and in this case, the tractability criterion is directly and easily verifyable.
\begin{mtheorem}
    Let $\mathcal{S}$ be a finite set of signatures such that $s(0)=1$ for all $s\in \mathcal{S}$. 
    \item[(I)] If  $s(n)=s(1)^n$ for all $s\in \mathcal{S}$ and $n\geq 1$, then $\text{\sc{p-UnColHolant}}(\mathcal{S})$ can be solved in FPT-near-linear time.
        \item[(II)] Otherwise $\text{\sc{p-UnColHolant}}(\mathcal{S})$ is $\#\W[1]$-complete. \qed
\end{mtheorem}

\subsubsection{Applications of our Main Results}

We now discuss new complexity classifications of specific problems that transpire from our main theorems. The complexities of these problems where previously known only for special cases.

\paragraph*{Counting Factors of size $k$}
Given a graph $G$ and a function $f:V(G)\to \mathcal{P}(\mathbb{N})$, an $f$-\emph{factor} of $G$ is a subset of edges $A$ such that $|E(v)\cap A| \in f(v)$ for all $v\in V(G)$. Moreover, given a $k$-edge-colouring $\xi$ of $G$ a factor is called colourful (w.r.t.\ $\xi$) if it contains precisely one edge per colour.

Graph factors have been studied since at least the 70s (c.f.\ \cite{Plummer07} for a survey) --- see also~\cite{MarxSS21,MarxSS22} for more recent results under the lens of parameterised complexity. Let us consider the following two problems:
\begin{definition}
    Let $\mathcal{B}$ be a finite non-empty subset of $\mathcal{P}(\mathbb{N})$. 
    \begin{itemize}
    \item $\textsc{ColFactor}(\mathcal{B})$ expects as input a graph $G$, a $k$-edge-colouring $\xi$ of $G$ and a mapping $f:V(G) \to \mathcal{B}$, and outputs the number of colourful $f$-factors of $G$. The parameter is $k$.
    \item $\textsc{Factor}(\mathcal{B})$ expects as input a graph $G$, a positive integer $k$, and a mapping $f:V(G) \to \mathcal{B}$, and outputs the number of $f$-factors of size $k$ in $G$. The parameter is $k$.
\end{itemize}
\end{definition}

In other words, $\textsc{Factor}(\mathcal{B})$ and $\textsc{ColFactor}(\mathcal{B})$ can be seen as coloured and uncoloured parameterised holant problems on signatures with co-domain $\{0,1\}$. Those allow us to express counting (colourful) $k$-edge-subgraphs with pre-specified degree constraints, subsuming among others, the problems of counting (colourful) $k$-matchings, $k$-partial cycle covers~\cite{BlaserC12}, or, more generally, $d$-regular $k$-edge subgraphs. We prove the following classification in Corollary~\ref{cor:factor_classification_col} for the coloured setting, and in Corollary~\ref{cor:factor_classification_uncol} for the uncoloured setting. 

\begin{theorem}\label{thm:factor_classification}
    If $\mathcal{B}$ contains a set $\{0\} \subsetneq S \subsetneq \mathbb{N}$ then the problems $\textsc{ColFactor}(\mathcal{B})$ and $\textsc{Factor}(\mathcal{B})$ are $\#\W[1]$-complete, and cannot be solved in time $f(k)\cdot n^{o(k/\log k)}$ for any function~$f$, unless the Exponential Time Hypothesis fails. Otherwise both problems are solvable in FPT-near-linear time.\qed
\end{theorem}

As a consequence of Theorem~\ref{thm:factor_classification} we do not only immediately infer the known parameterised hardness results for counting $k$-partial cycle covers~\cite{BlaserC12} and $k$-matchings~\cite{Curticapean13}, but we also obtain the following, significantly more general, lower bounds.

\begin{corollary}
    Let $d\geq 1$ be any fixed integer. The problem of counting $d$-regular $k$-edge subgraphs in an $n$-vertex graph $G$ is $\#\W[1]$-complete when parameterised by $k$, and cannot be solved in time $f(k)\cdot n^{o(k/\log k)}$ for any function~$f$, unless the Exponential Time Hypothesis fails. The same holds true if $G$ comes with a $k$-edge-colouring and the goal is to count edge-colourful $d$-regular subgraphs.\qed
\end{corollary}

\paragraph*{Modular Counting of (Colourful) Matchings}
The parameterised counting problems $\oplus\textsc{ColMatch}$ and $\oplus\textsc{Match}$ ask, respectively, to compute the number of colourful $k$-matchings in a $k$-edge-coloured graph and the number of $k$-matchings in an (uncoloured) graph; here a $k$-matching is colourful if it contains precisely one edge per colour. Both problems are parameterised by $k$.

Recently, Curticapean, Dell, and Husfeldt analysed the parameterised complexity of counting small subgraphs, modulo a fixed prime $p$~\cite{CurticapeanDH21}. One of their results is an FPT algorithm for the problem $\oplus\textsc{Match}$. Using a standard trick based on inclusion-exclusion, it can be shown that the edge-colourful variant $\oplus\textsc{ColMatch}$ reduces to $\oplus\textsc{Match}$ via parameterised reductions. Therefore, it is not surprising that $\oplus\textsc{ColMatch}$ is fixed-parameter tractable as well. However, we can prove the stronger fact that $\oplus\textsc{ColMatch}$ can be solved in FPT-matrix-multiplication time. Moreover, we also show that neither $\oplus\textsc{ColMatch}$ nor $\oplus\textsc{Match}$ can be solved in FPT-near-linear time; the proof of the subsequent result can be found in Section~\ref{sec:modular}:
\begin{theorem}\label{thm:colmatch_mod_2_intro}
    $\oplus\textsc{ColMatch}$ can be solved in FPT-matrix-multiplication time. Moreover, neither $\oplus\textsc{ColMatch}$, nor $\oplus\textsc{Match}$ can be solved in FPT-near-linear time, unless the Triangle Conjecture fails.\qed
\end{theorem}
Interestingly, to the best of our knowledge, it is (and remains) unknown whether $\oplus\textsc{Match}$ can also be solved in FPT-matrix-multiplication time, since our main result for the uncoloured holant problem does \emph{not} extend to modular counting.

\paragraph*{Counting Weight-$k$ Solutions to Systems of Linear Equations}
We provide an intractability result for the following coding problem (see e.g.\ \cite{BerlekampMT78,DowneyFVW99}), also known as $\#\textsc{Weighted-XOR-SAT}$, or the counting version of $\textsc{Exact-Even-Set}$.\footnote{Note that the problem of \textbf{counting} solutions of Hamming weight $k$ is equivalent to counting solutions of weight \emph{at most} $k$: For one direction, the number of solutions of Hamming weight at most $k$ can be obtained by adding the solutions of Hamming weight $\ell$ for $\ell=0,\dots,k$. For the other direction, the number of solutions of Hamming weight $k$ is equal to the difference of number of solutions of Hamming weight at most $k$ and at most $k-1$.} 

\begin{corollary}\label{cor:coding_hard}
    The problem of counting the Hamming weight $k$ solutions of a system of linear equations $A\vec{x}=0$ over $\mathbb{Z}/2\mathbb{Z}$ is $\#\W[1]$-hard when parameterised by $k$, and cannot be solved in time $f(k)\cdot n^{o(k/\log k)}$ for any function~$f$, unless the Exponential Time Hypothesis fails. This holds true even if the matrix $A$ is promised to contain at most two $1$s per column.
\end{corollary}
\begin{proof}
    Follows immediately by observing that $\textsc{p-UnColHolant}(\{\mathsf{even}\})$ is a special case of this problem: Each edge $e$ of the signature grid becomes a variable $x_e$, and each vertex $v$ of the signature grid yields the equation $\sum_{e\in N(v)}x_e = 0 \mod 2$.  We have already seen that $\mathsf{even}$ is of type $\mathbb{T}[\infty]$; thus the claim holds by Main Theorem~\ref{thm:main_uncol}.
\end{proof}
We point out that the absence of an FPT algorithm for the problem in the previous result is not surprising, given hardness results of related parameterised decision problems~\cite{DowneyFVW99}, as well as the breakthrough (hardness) result on the $\textsc{EvenSet}$ problem due to Bhattacharyya et al.\ \cite{BhattacharyyaGS18}. However, the latter results do not come with tight lower bounds under ETH, and~\cite{BhattacharyyaGS18} even requires as lower bound assumption a stronger, approximate version of ETH, called $\mathsf{GAP}$-ETH. Moreover, neither result applies to the restriction to input matrices with at most two $1s$ per column.

We further remark that Corollary~\ref{cor:coding_hard} does not contradict the tractability results of Creignou and Vollmer~\cite{CreignouV15}, and of Marx~\cite{Marx05} on a seemingly similar weighted satisfiability problem, since both~\cite{CreignouV15} and~\cite{Marx05} enforce a constant upper bound on the number of variables in each equation.

Finally, note that, Corollary~\ref{cor:coding_hard} also shows that the tractability criteria for parameterised Holants significantly differ from classical holant problems, as $\mathsf{even}$ is an affine signature, and affine signatures yield polynomial time tractability if given an instance of a classical Holant~\cite{ParityHolant13}.

\subsection{Technical Contributions}
It turns out that we do not need to rely on any of the well-established tools for analysing holant problems, such as matchgates~\cite{CaiC07}, combined signatures~\cite{CurticapeanX15}, or holographic reductions~\cite{Valiant08,CaiL11}, despite the fact that some of those tools have already been adapted to the realm of parameterised counting by Curticapean~\cite{Curticapean15}.
Instead, similarly to most works on exact parameterised counting throughout the last seven years~\cite{Roth17,DellRW19,RothSW20,BressanR21,BLR2023stoc,DoringMW24,CurticapeanN24,DoringMW25}, we crucially rely on the framework of complexity monotonicity of graph motif parameters as introduced by Curticapean, Dell, and Marx~\cite{CurticapeanDM17}. In a nutshell, we show that both the coloured and the uncoloured holant problems can be cast as a finite linear combination of homomorphism counts. Let us make this explicit for the uncoloured version --- in fact, the coloured version is easier to analyse, but requires a more extensive set-up on edge-coloured graphs, which we omit from the introduction for the sake of conceptual clarity.

Fix a finite set $\mathcal{S}=\{s_1,\dots,s_\ell\}$ of signatures. We write $\mathcal{G}(\mathcal{S})$ for the set of all signature grids over $\mathcal{S}$. Given two signature grids $\Omega_H=(H,\{s^H_v\}_{v\in V(H)})$ and $\Omega_G=(G,\{s^G_v\}_{v\in V(G)})$ in $\mathcal{G}(\mathcal{S})$, a \emph{homomorphism} from $\Omega_H$ to $\Omega_G$ is a graph homomorphism $\varphi$ from $H$ to $G$ such that $s^G_{\varphi(v)}=s^H_v$ for all $v\in V(H)$, that is $\varphi$ must preserve signatures. We write $\#\homs{\Omega_H}{\Omega_G}$ for the number of homomorpisms from $\Omega_H$ to $\Omega_G$.

Using M\"obius inversion, it is not hard to show that, for each $k$, there is a finitely supported function $\zeta_{\mathcal{S},k}$ from $\mathcal{G}(\mathcal{S})$ to algebraic complex numbers such that, for each signature grid $\Omega$ over $\mathcal{S}$, we have
\begin{equation}\label{eq:to_hombasis_intro}
    \mathsf{UnColHolant}(\Omega,k) = \prod_{i\in \ell}s_i(0)^{n_i} \cdot \sum_{\Omega_H\in \mathcal{G}(\mathcal{S})} \zeta_{\mathcal{S},k}(\Omega_H) \cdot \#\homs{\Omega_H}{\Omega}\,.
\end{equation}
As mentioned before, a similar transformation exists for the colourful holant problem.
Now, an adaptation of the principle of ``complexity monotonicity''~\cite{CurticapeanDM17} will imply that evaluating the linear combination~\eqref{eq:to_hombasis_intro} is \emph{precisely as hard as} evaluating its hardest term $\#\homs{\Omega_H}{\Omega}$ with a non-zero coefficient $\zeta_{\mathcal{S},k}(\Omega_H)\neq 0$. Fortunately, the complexity of evaluating the individual terms $\#\homs{\Omega_H}{\Omega}$ is very well understood: under standard assumptions from fine-grained and parameterised complexity theory, the evaluation is hard if the treewidth (see Section~\ref{sec:prelims}) of $\Omega_H$ is large, and the evaluation is easy if the treewidth of $\Omega_H$ is small. For this reason, the goal of understanding the complexity of parameterised holant problems can be reduced to the purely combinatorial problem of understanding the coefficient function $\zeta_{\mathcal{S},k}$, and its analogue in the coloured setting.

In previous results, this approach has mostly been used for establishing \emph{lower bounds} on pattern counting problems on graphs~\cite{FockeR22,DoringMW24,Curticapean24,CurticapeanN24}, in which case it suffices to find a high-treewidth term that survives with a non-zero coefficient, and it has turned out that even finding \emph{one} such a term can constitute solving highly challenging combinatorial problems~\cite{RothSW20,PeyerimhoffRSSVW23,DoringMW24,DoringMW25}. Using the framework for upper bounds is usually even more difficult since it makes it necessary to find an upper bound on the treewidth of \emph{all} graphs surviving with a non-zero coefficient (see e.g.\ \cite{PeyerimhoffRSV21,BeraGLSS22,BressanR21,BLR23}). In the current work, we solve this task as precisely as possible; intuitively, our main combinatorial result reads as follows:

\begin{center}
    \textit{The maximum treewidth of the graphs surviving in the ``homomorphism basis'' of \emph{any} instance of a parameterised Holant problem (coloured or uncoloured) is either $1$, $2$, or unbounded.}
\end{center}
Concretely, in the coloured setting, we show that 
\begin{enumerate}
    \item[(1)] for $\mathcal{S}$ of type $\mathbb{T}[\mathsf{Lin}]$ all terms of treewidth $\geq 2$ vanish,
    \item[(2)] for $\mathcal{S}$ of type $\mathbb{T}[\omega]$ all terms of treewidth $\geq 3$ vanish, but at least one term of treewidth $2$ survives, and
    \item[(3)] for $\mathcal{S}$ of type $\mathbb{T}[\infty]$, terms of arbitrary high treewidth survive.
\end{enumerate}
In the uncoloured setting (specifically in Equation~\eqref{eq:to_hombasis_intro}), we show that 
\begin{enumerate}
    \item[(1)] for $\mathcal{S}$ of type $\mathbb{T}[\mathsf{Lin}]$ all terms of treewidth $\geq 2$ vanish,
    \item[(2)] for $\mathcal{S}$ of type $\mathbb{T}[\omega]$ or $\mathbb{T}[\infty]$, terms of arbitrary high treewidth survive.
\end{enumerate}

We consider the combinatorial analysis of the coefficients $\zeta$ to be the central technical contribution of this work. Moreover, we emphasize that, by understanding the coefficients in this detail, our work is, to the best of our knowledge, the only application of the framework of complexity monotonicity achieving a classification of a general family of counting problems into FPT and $\#\mathrm{W}[1]$-complete cases in which the exponents of all FPT cases are almost precisely determined under assumptions from fine-grained complexity theory.

\subsection{Conclusion and Future Directions}
In this work, we focused on \emph{symmetric} signatures for parameterised Holant problems. We provided complete classifications for both the coloured and the uncoloured variant of the problem, not only identifying precisely those instances that allow for FPT algorithms, but also proving almost tight bounds for the best possible run-time exponents under assumptions from fine-grained complexity theory.

We identify the following questions as starting points for future research on parameterised Holants:
\begin{itemize}
    \item Asymmetric signatures. In classical Holant theory, after the case of symmetric (boolean) signatures had been solved~\cite{CaiGW16}, a significant amount of effort has been put in understanding asymmetric signatures, which involve more intricate cases than the classifications for symmetric signatures \cite{CaiFX18,CaiF23,meng2025}. We expect that the tools we set up for the study of parameterised Holants can be generalised to apply to the asymmetric setting via encoding the orderings of activated incident edges to a vertex by imposing  constraints on the set of tuples of edge-colours.
    \item Polynomial Time vs.\ $\#\mathrm{P}$-hardness. The FPT cases of our classifications are all ``real'' FPT cases in the sense that the running time of our algorithms yields superpolynomial overhead in the parameter $k$. For a future direction, we propose to study for which cases the superpolynomial overhead is necessary, under standard assumptions from (counting) complexity theory such as $\mathrm{FP}\neq \#\mathrm{P}$.
\end{itemize}

\subsection{Organisation of the Paper}

We start with introducing the required preliminary material in Section~\ref{sec:prelims}. In particular, we revisit the framework of \emph{fractured graphs} as introduced in~\cite{PeyerimhoffRSSVW23}, which will be a central ingredient for the classification of $\holantprob(\mathcal{S})$.
Afterwards, in Section~\ref{sec:col_holants_equivalences}, we consider an intermediate problem $\holantprobstar(\mathcal{S})$ and show it to be interreducible with $\holantprob(\mathcal{S})$ under parameterised linear-time reductions. This allows us then to prove Main Theorem~\ref{main_thm} by classifying first the intermediate problem $\holantprobstar(\mathcal{S})$, which, for technical reasons, is much more convenient to work with than $\holantprob(\mathcal{S})$. In Section~\ref{sec:modular}, we present the extension of Main Theorem~\ref{main_thm} to counting modulo a fixed prime $p$.
Next, in Section~\ref{sec:sig0}, we show how our classification for $\holantprob(\mathcal{S})$ in Main Theorem~\ref{main_thm} can be extended to the case in which we allow signatures $s$ with $s(0)=0$.
Finally, in Section~\ref{sec:uncoloured}, we prove our classification for the uncoloured holant problem (Main Theorem~\ref{thm:main_uncol}), and we encapsulate the analysis of the coefficient function $\zeta$ for the uncoloured problem in Section~\ref{sec:analysis_of_zeta}. For reasons of accessibility, we present our analysis of $\zeta$ in Section~\ref{sec:analysis_of_zeta} in multiple steps, starting with the special cases that encapsulate the central ideas, and then iteratively and carefully generalising the arguments and results to obtain a thorough understanding of $\zeta$ in the full and unrestricted case.
We also include an appendix with some proofs that are omitted in the main paper, and a construction of infinitely many signature families of each of the three types in our classifications.

\newpage
\tableofcontents

\pagebreak

\section{Preliminaries}\label{sec:prelims}

For a finite set $S$, we denote the cardinality of~$S$ by~$|S|$ or $\#S$. Graphs in this work are undirected, simple and without loops, unless stated otherwise. 
Given a vertex $v$ of a graph $G$, we write $E(v)$ for the set of all edges of $G$ that are incident to $v$. Let $\mathcal{N}_{G}(v) = \{u \in V(G) \mid \{u, v\} \in E(G)\}$ denote the neighbourhood of a vertex $v\in G$ and let $d_G(v)=|\mathcal{N}_G(v)|$ denote the degree of $v$. 
For a subset $J\subseteq V(G)$, we denote by $G[J]$ the induced subgraph graph of $G$ with vertices $J$ and edges $\{u,v\}\in E(G)$ with $u,v\in J$. Throughout this paper, given a graph $G$, a set $S$, and a mapping $m: V(G) \to S$, we will always freely allow ourselves to extend $m$ to the edges of $G$ by just setting $m(\{u,v\}):=\{m(u),m(v)\}$. For a function $f:A\to B$ and a subset $C\subseteq A$, we write $f|_A: C \to B$ for the restriction of $f$ to $C$.

A \emph{tree-decomposition} of a graph $G$ is a pair of a tree $T$ and a mapping $\beta: V(T) \to 2^{V(G)}$ such that $\bigcup_{t\in V(T)} \beta(t)=V(G)$, for each edge $e$ of $G$ there is $t\in V(T)$ such that $e \subseteq \beta(t)$, and for each vertex $v$ of $G$ the graph $T[\{t\in V(T)\mid v\in \beta(t)\}]$ is connected.
The width of $(T,\beta)$ is $\max\{|\beta(t)|-1 \mid t\in V(t)\}$, and the \emph{treewidth} of $G$, denoted by $\mathsf{tw}(G)$ is the minimum possible width of a tree-decomposition of $G$. 
It will sometimes be convenient to consider tree-decompositions $(T,\beta)$ where $T$ is a \emph{rooted} tree with some fixed root $r$. Then,
for $t \in V(T)\setminus\{r\}$, let $\sigma(t)$ denote \textit{the separator at $t$}, that is $\sigma(t) = \beta(t)\,\cap\,\beta(t')$, where $t'$ is the parent of $t$ in $T$. For the root $r$, we set $\sigma(r) = \emptyset$. Let also $\gamma(t)\subseteq V(H)$ denote \textit{the cone at $t$}, that is, $\gamma(t)$ is the union of the $\beta(\hat{t})$ for all descendants $\hat{t}$ of $t$.

\subsection{Homomorphisms, Embeddings, and Automorphisms}
A \emph{homomorphism} from a graph $H$ to a graph $G$ is a mapping $\varphi: V(H) \to V(G)$ such that $\varphi(e)\in E(G)$ for each edge $e\in E(H)$. We write $\homs{H}{G}$ for the set of all homomorphisms from $H$ to $G$. An \emph{embedding} from $H$ to $G$ is a homomorphism from $H$ to $G$ that is injective (on the vertices of $H$), and we write $\embs{H}{G}$ for the set containing all embeddings from $H$ to $G$. Finally, an embedding from $H$ to itself is called an automorphism of $H$, and we write $\auts(H)$ for the set of all automorphisms of $H$. Note that, for $\pi \in \auts(H)$, we have $\{u,v\}\in E(H)\Leftrightarrow \{\pi(u),\pi(v)\}\in E(H)$ for each pair of vertices $u,v \in V(H)$.

\subsection{Partitions, Quotient Graphs, and M\"obius Functions}\label{sec:parts_quots_mobius}
A partition $\rho$ of a finite set $S$ is a set of pairwise disjoint and non-empty blocks $\rho=\{B_1,\dots,B_t\}$ such that $\dot\cup_{i=1}^t B_i = S$; we emphasise that the order of the blocks does not matter. We write $|\rho|$ for the number of blocks of $\rho$. Given two partitions $\rho$ and $\sigma$ of $S$, we say that $\sigma$ \emph{refines} $\rho$, denoted by $\sigma \leq \rho$, if for each block $B^\rho\in \rho$ there are blocks $B_1^\sigma,\dots,B_\ell^\sigma$ of $\sigma$ such that $\dot\cup_{i=1}^\ell B^\sigma_i = B^\rho$. Intuitively, this means that $\sigma$ can be decomposed into ``subpartitions'' $\sigma_B$ of $B$ for each block $B$ in $\rho$. We write $\bot_S$ for the \emph{finest partition}, that is, for the partition that contains blocks $\{s\}$ for each $s\in S$, and we write $\top_S$ for the \emph{coarsest partition}, that is, for the partition containing only one block $B=S$. We might drop the subscript and just write $\bot$ and $\top$ if~$S$ is clear from the context.
Given two partitions $\sigma \leq \rho$, the \emph{M\"obius function} of $\sigma$ and $\rho$ is defined as follows.\footnote{In fact, the M\"obius function is \emph{defined} using the incidence algebra of a partially ordered set (see e.g.\ \cite[Chapter 3.6]{Stanley11}), and Definition~\ref{def:mobius} is a lemma on the M\"obius function over the poset of partition refinement. However, since we will not rely on any further properties of the incidence algebra of that poset, we allow ourselves to introduce the M\"obius function via the explicit formula in Definition~\ref{def:mobius} for reasons of self-containment.}

\begin{definition}[The M\"obius function of partitions (cf.\ \cite{Stanley11})]\label{def:mobius}
    Let $S$ be a finite set and let $\rho=\{B_1,\dots,B_t\}$ be a partition of $S$. Let furthermore $\sigma$ be a partition of $S$ with $\sigma \leq \rho$. For each $i\in [t]$, let $\sigma^i$ be the subpartition of $\sigma$ refining $\{B_i\}$, that is, $\sigma = \dot{\cup}_{i=1}^t \sigma^i$ and $\sigma^i$ is a partition of $B_i$ for all $i\in [t]$. Then the \emph{M\"obius function} of $\sigma$ and $\rho$ is defined as follows:
    \[ \mu_S(\sigma,\rho):=  \prod_{i=1}^t (-1)^{|\sigma^i|-1}(|\sigma^i|-1)! \,,\]
    where $0!=1$ as usual. We will drop the subscript $S$ of $\mu$ if it is clear from the context.
\end{definition}

\paragraph*{Quotient Graphs}
Given a graph $H$, we write $\ppart(H)$ for the set of all partitions of $V(H)$. Given a partition $\rho$ of $V(H)$, the \emph{quotient graph} $H/\rho$ has as vertices the blocks of $\rho$ and two blocks $B_1$ and $B_2$ are made adjacent if (and only if) there are vertices $u\in B_1$ and $v\in B_2$ such that $\{u,v\} \in E(H)$. Note that $H/\rho$ \emph{might contain self-loops}. We write $h_\rho: V(H) \to V(H/\rho)$ for the mapping that assigns each vertex $v\in V(H)$ the block $B$ containing $v$, and we observe that $h_\rho$ is a homomorphism. We write $\mu_H$ for the M\"obius function $\mu_{V(H)}$ of partitions of $V(H)$, and we might again drop the subscript if it is clear from the context. 

\subsection{$H$-coloured graphs and fractured graphs}\label{sec:col_graphs_fractures}

Given a graph $H$, an $H$\emph{-coloured graph} is a pair $(G,h)$ of a graph $G$ and a homomorphism, called $H$-\emph{colouring}, $h\in \homs{G}{H}$.
Given a graph $H$ and an $H$-coloured graph $(G,h)$, a homomorphism $\varphi \in \homs{H}{G}$ is called \emph{colour-prescribed}, w.r.t.\ $h$, if $h(\varphi(v))=v$ for all $v\in V(H)$. We write $\cphoms(H \to (G,h))$ for the set of all colour-prescribed homomorphisms (w.r.t.\ $h$) from $H$ to $G$.

\begin{definition}[Fracture]\label{def:fracture}
    Let $H$ be a graph without isolated vertices. A \emph{fracture} of $H$ is a tuple $\vec\rho=(\rho_v)_{v\in V(H)}$ where $\rho_v$ is a partition of $E(v)$ for each $v\in V(H)$.
\end{definition}

Let $H$ be a graph without isolated vertices, we denote by $\mathcal{F}(H)$ the set of all fractures of $H$. Given $\vec{\rho}\in\mathcal{F}(H)$ and a vertex $v$ of $H$, we will write $\vec{\rho}(v)$ for the partition in $\vec{\rho}$ corresponding to vertex $v$. Given two fractures $\vec{\sigma}$ and $\vec{\rho}$ of a graph $H$, we slightly overload notation and denote point-wise partition refinement by $\leq$, that is, we write $\vec{\sigma}\leq \vec{\rho}$ if $\vec{\sigma}(v) \leq \vec{\rho}(v)$ for all vertices $v$ of $H$. The M\"oebius function extends from partitions to fractures as follows:\footnote{Again, the M\"obius function of fractures (of a graph $H$) is normally defined via the incidence algebra of the poset of fractures of $H$ and pointwise partition refinement. As observed in~\cite{PeyerimhoffRSSVW23}, the resulting poset is isomorphic to the product of $|V(H)|$ partition lattices and the M\"obius function factorises along the individual partitions.}

\begin{definition}[The M\"obius function of fractures]
    Let $\vec{\sigma} \leq \vec{\rho}$ be fractures of a graph $H$. The \emph{M\"obius} function of $\vec{\sigma}$ and $\vec{\rho}$ is defined as follows: 
    \[\vec{\mu}_H := \prod_{v\in V(H)} \mu_{E(v)}(\vec{\sigma}(v),\vec{\rho}(v))\,.\]
    We will drop the subscript $H$ if the graph is clear from the context.
\end{definition}

For notational convenience, we will always use $\vec{\mu}$ to denote the M\"obius function of fractures, and $\mu$ to denote the M\"obius function of partitions.

Informally, a fracture $\vec{\rho}$ of a graph $H$ is an instruction on how to split, or ``fracture'' the vertices of~$H$. Concretely, for each $v\in V(H)$, we replace $v$ by vertices $v^B$ for each $B \in \vec{\rho}(v)$. Then we add an edge between two vertices $u^B$ and $v^{B'}$ if and only if $\{u,v\}\in B \cap B'$. Formally, we state the definition as in~\cite{PeyerimhoffRSSVW23}; consult Figure~\ref{fig:fracture} for an illustration.

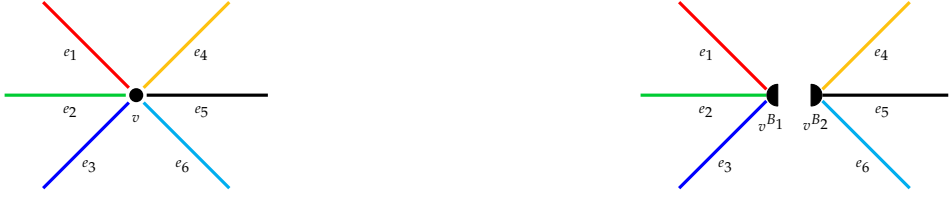
\begin{figure}[t!]
    \centering
    \begin{tikzpicture}[scale=1.75]

        \node[vertex,inner sep=.4ex,label={[label distance=.03]below:\(v\)}] (m) at (0, 0) {};

        \draw[very thick,red] (m) -- ++(135:1);
        \node[label={[label distance=.02]below:\(e_1\)}] at (135:0.7) {};
        \draw[very thick,green!80!blue] (m) -- ++(180:1);
        \node[label={below:\(e_2\)}] at (180:0.5) {};
        \draw[very thick,blue] (m) -- ++(-135:1);
        \node[label={[label distance=.02]below:\(e_3\)}] at (-135:0.5) {};

        \draw[very thick,yellow!50!orange] (m) -- ++(45:1);
        \node[label={[label distance=.03]below:\(e_4\)}] at (45:0.7) {};
        \draw[very thick] (m) -- ++(0:1);
        \node[label={below:\(e_5\)}] at (0:0.5) {};
        \draw[very thick,cyan] (m) -- ++(-45:1);
          \node[label={[label distance=.03]below:\(e_6\)}] at (-45:0.5) {};

        \begin{scope}[shift={(5,0)}]
            \begin{scope}[scale=2.4]
                \kowaen{0,0}{-90/90/white,90/270/white}{1};
            \end{scope}

            \node[label={[label distance=.03]below:\(v^{B_1}\)}]  at (1-2) {};

            \draw[very thick,red] (1-2) -- ++(135:1);
            
            \draw[very thick,green!80!blue] (1-2) -- ++(180:1);
            
            \draw[very thick,blue] (1-2) -- ++(-135:1);
            \begin{scope}[shift=(1-2)]
            \node[label={below:\(e_2\)}] at (180:0.5) {};
            \node[label={[label distance=.02]below:\(e_1\)}] at (135:0.7) {};
            \node[label={[label distance=.02]below:\(e_3\)}] at (-135:0.5) {};
            \end{scope}

            \node[label={[label distance=.03]below:\(v^{B_2}\)}]  at
                (1-1) {};

            \draw[very thick,yellow!50!orange] (1-1) -- ++(45:1);
            \draw[very thick] (1-1) -- ++(0:1);
            \draw[very thick,cyan] (1-1) -- ++(-45:1);
            \begin{scope}[shift=(1-1)]
            \node[label={[label distance=.03]below:\(e_6\)}] at (-45:0.5) {};
             \node[label={below:\(e_5\)}] at (0:0.5) {};
             \node[label={[label distance=.03]below:\(e_4\)}] at (45:0.7) {};
            \end{scope}

            \begin{scope}[scale=1.8]
                \kowaen{0,0}{-90/90/black,90/270/black}{1};
            \end{scope}

        \end{scope}
    \end{tikzpicture}
    \caption{\label{fig:fracture} Illustration of the construction of a
        fractured graph taken from~\cite{PeyerimhoffRSSVW23}. The left picture shows a vertex $v$ of a graph~$H$ with incident
        edges $E_H(v)=\{e_1,\dots,e_6\}$.
% $E_H(v)=\{\redc{2pt},\greenc{2pt},\bluec{2pt},\yellowc{2pt},\blackc{2pt},\brownc{2pt}\}$. 
The right
        picture shows the splitting of $v$ in the construction of the fractured
        graph~$\fracture{H}{\vec{\rho}}$ for a fracture $\vec{\rho}$ satisfying that the partition
        $\vec{\rho}(v)$ contains two blocks $B_1 =\{e_1,e_2,e_3\}$ and $B_2=\{e_4,e_5,e_6\}$.
        % $B_1 =\{ \redc{2pt},\greenc{2pt},\bluec{2pt}\}$, and
    % $B_2=\{\yellowc{2pt},\blackc{2pt},\brownc{2pt}\}$.
    }
\end{figure}

\begin{definition}[Fractured graphs $\fracture{H}{\vec{\rho}}$]\label{def:fract_graph}
 Let $H$ be a graph without isolated vertices, and let $\vec{\rho}\in\mathcal{F}(H)$. Write $M_H$ for the matching containing one copy of each edge of $H$, that is, $V(M_H)=\bigcup_{e \in E(H)}\{u_e,v_e\}$, and $E(M_H)= \{\{u_e,v_e\} \mid e \in E(H)\}$.

 Let $\tau$ be the partition on $V(M_H)$ that places two vertices $u_e,v_f$ into the same block if and only if $u=v$ and there exists $B\in \vec{\rho}(u)$ with $e,f \in B$. Then the \emph{fractured graph} $\fracture{H}{\vec{\rho}}$ is defined to be the quotient $M_H/\tau$. We write $v^B$ for the vertex of $\fracture{H}{\vec{\rho}}$ corresponding, in this construction, to vertex $v\in V(H)$ and block $B\in \vec{\rho}(v)$. 
\end{definition}

Given a graph $H$ and $\vec{\rho}\in\mathcal{F}(H)$, the fractured graph $\fracture{H}{\vec{\rho}}$ admits an $H$-colouring $h_{\vec{\rho}}$ which maps $v_B$ to $v$ for each $v\in V(H)$ and $B \in \vec{\rho}(v)$. We refer to $h_{\vec{\rho}}$ as the \emph{canonical $H$-colouring} of $\fracture{H}{\vec{\rho}}$. We next introduce a version of colour-preserving homomorphisms from fractured graphs.

\begin{definition}[$\homscp,\embscp$]
    Let $H$ be a graph and let $\vec{\rho}\in\mathcal{F}(H)$. Furthermore, let $(G,h)$ be an $H$-coloured graph. We set
    \[ \homscp(\fracture{H}{\vec{\rho}} \to (G,h))= \{\varphi \in \homs{\fracture{H}{\vec{\rho}}}{G} \mid \forall x \in V(\fracture{H}{\vec{\rho}}): h(\varphi(x)) = h_{\vec{\rho}}(x) \} \,, \]
    where $h_{\vec{\rho}}$ is the canonical $H$-colouring of $\fracture{H}{\vec{\rho}}$. The set $\embscp(\fracture{H}{\vec{\rho}} \to (G,h))$ is defined similarly for embeddings.
\end{definition}
In other words, $\homscp(\fracture{H}{\vec{\rho}} \to (G,h))$ and $\embscp(\fracture{H}{\vec{\rho}} \to (G,h))$ contain, respectively, the homomorphisms and embeddings from $\fracture{H}{\vec{\rho}} $ to $G$ that maps vertices $v_B$ of $\fracture{H}{\vec{\rho}}$ (i.e., $v\in V(H)$ and $B$ is a block of $\vec{\rho}(v)$) to vertices in $G$ coloured by $h$ with $v$.

\subsection{Parameterised and Fine-Grained Complexity Theory}
We provide a concise introduction to parameterised counting complexity in what follows, and we refer the reader to the standard textbook~\cite{CyganFKLMPPS15} for a comprehensive treatment of parameterised algorithms and to~\cite[Chapter 14]{FlumG06} and~\cite{Curticapean15} for a detailed overview over parameterised counting problems.

A \emph{parameterised counting problem} is a pair of a function $P:\{0,1\}^\ast \to \mathbb{Q}$ and a computable parameterisation\footnote{We note that some authors require the parameterisation to be computable in polynomial time, which is the case for all parameterisations encountered in this work.} $\kappa: \{0,1\}^\ast \to \mathbb{N}$. For example, $\#\textsc{Clique}$ denotes the parameterised counting problem of computing the number of $k$-cliques in a graph $G$, and it is parameterised by $k$, that is, $\kappa(G,k)=k$.
A parameterised counting problem $(P,\kappa)$ is called \emph{fixed-parameter tractable} (FPT) if there is a computable function $f$ and an algorithm $\mathbb{A}$ such that, on input $x$, computes $P(x)$ in time $f(\kappa(x))\cdot |x|^{O(1)}$. We call $\mathbb{A}$ an \emph{FPT} algorithm w.r.t.\ parameterisation $\kappa$.

For the purpose of this work, we will be interested in the exact constant in the exponent of $|x|$ in the running time of an FPT algorithm; in the definition below we use $\tilde{O}(n)$ to hide poly-logarithmic factors, i.e., $t\in \tilde{O}(n)$ if there is a constant $d$ such that $t\in O(n\log^dn)$.

\begin{definition}[FPT-near-linear time and FPT-matrix-multiplication time]
    An algorithm $\mathbb{A}$ is called an \emph{FPT-near-linear time algorithm} w.r.t.\ parameterisation $\kappa$ if there is a computable function $f$ such that $\mathbb{A}$ runs in time $f(\kappa(x))\cdot \tilde{O}(|x|)$. Moreover, we call $\mathbb{A}$ an \emph{FPT-matrix-multiplication time algorithm} if there is a computable function $f$ such that $\mathbb{A}$ runs in time $f(\kappa(x))\cdot O(|x|^{\omega})$, where $\omega$ is the matrix multiplication exponent.\footnote{At the time of writing of this work, the best known bound for $\omega$ is $2\leq \omega\leq 2.371552$~\cite{WilliamsXXZ24}.}

    We say that a parameterised counting problem is solvable in FPT-near-linear time (resp.\ FPT-matrix-multiplication time) if it can be solved by an FPT-near-linear time (resp.\ FPT-matrix-multiplication time) algorithm.
\end{definition}

We proceed by introducing two notions for reductions between parameterised counting problems. We note that FPT Turing-reductions are standard~\cite[Chapter 14]{FlumG06}, but \emph{linear} FPT Turing-reductions are less so, since we have to be careful with the number of oracle queries.

\begin{definition}[Parameterised Reductions]
    Let $(P,\kappa)$ and $(P',\kappa')$ be parameterised counting problems. An \emph{FPT Turing-reduction} from $(P,\kappa)$ to $(P',\kappa')$ is an algorithm $\mathbb{A}$ with the following properties: There exists a computable function $f$ such that
    \begin{enumerate}
        \item on input $x$, $\mathbb{A}$ computes $P(x)$ in time $f(\kappa(x))\cdot |x|^{O(1)}$, and
        \item $\mathbb{A}$ has oracle access to $P'$, and, on input $x$, each oracle query $y$ posed by $\mathbb{A}$ satisfies $\kappa'(y)\leq f(\kappa(x))$.
    \end{enumerate}
    We write $(P,\kappa)\fptred (P',\kappa')$ if an FPT Turing-reduction exists. Moreover, we write $(P,\kappa)\fptinterred (P',\kappa')$ if $(P,\kappa)\fptred (P',\kappa')$ and $(P',\kappa')\fptred (P,\kappa)$.

    A \emph{linear FPT Turing-reduction} from $(P,\kappa)$ to $(P',\kappa')$ is an algorithm $\mathbb{A}$ with the following properties: There exists a computable function $f$ such that
    \begin{enumerate}
        \item on input $x$, $\mathbb{A}$ computes $P(x)$ in time $f(\kappa(x))\cdot O(|x|)$.
        \item $\mathbb{A}$ has oracle access to $P'$, and, on input $x$, each oracle query $y$ posed by $\mathbb{A}$ satisfies $\kappa'(y)\leq f(\kappa(x))$, and the number of oracle queries must be bounded by $f(\kappa(x))$.
    \end{enumerate}
    We write $(P,\kappa)\fptlinred (P',\kappa')$ if a linear FPT Turing-reduction exists. Moreover, we write $(P,\kappa)\fptinterlinred (P',\kappa')$ if $(P,\kappa)\fptlinred (P',\kappa')$ and $(P',\kappa')\fptlinred (P,\kappa)$.
\end{definition}

It is well-known that $(P,\kappa)$ is FPT if $(P,\kappa)\fptred (P',\kappa')$ and $(P',\kappa')$ is FPT. For the purpose of this work, we establish a more fine-grained version of this fact via linear FPT Turing-reductions.

\begin{lemma}
    Let $d\geq 0$ and $c\geq 1$ be reals, and let $(P,\kappa)$ and $(P',\kappa')$ be parameterised counting problems such that $(P,\kappa)\fptlinred (P',\kappa')$. Assume that $(P',\kappa')$ can be solved in time $f(\kappa'(x))\cdot O(\log^d(|x|)\cdot |x|^c)$ for some computable function $f$. Then there is a computable function $g$ such that $(P,\kappa)$ can be solved in time $g(\kappa(x))\cdot O(\log^d(|x|)\cdot |x|^c)$. 
\end{lemma}
\begin{proof}
    Let $\mathbb{A}$ be the linear FPT Turing-reduction, that is, there is a computable function $f'$ such that $\mathbb{A}$, on input $x$, computes $P(x)$ in time $f'(\kappa(x))\cdot O(|x|)$ by querying at most $f'(\kappa(x))$ oracle queries to $P'$, and each oracle query $y$ satisfies $\kappa'(y)\leq f'(\kappa(x))$.

    We obtain an algorithm $\mathbb{A}'$ from $\mathbb{A}$ by simulating each oracle call via the given algorithm for $(P',\kappa')$ that runs, on input $y$, in time  $f(\kappa'(y))\cdot O(\log^d(|y|)\cdot |y|^c)$. 
    
    Since, on input $x$ the number of oracle calls posed by $\mathbb{A}$ is bounded by $f'(\kappa(x))$, the total running time is bounded by
    \begin{equation}\label{eq:lin_time_transitive}
    O(f'(\kappa(x))\cdot |x| + f'(\kappa(x))\cdot f(\kappa'(y_1))\cdot \log^d(|y_2|)\cdot |y_2|^c) \,,
    \end{equation}
    where $y_1$ and $y_2$ are the oracle queries maximizing $f(\kappa'(y))$ and $|y|$, respectively. Since each oracle query $y$ satisfies $\kappa'(y)\leq f'(\kappa(x))$, we have $f(\kappa'(y_1)) \leq f(f'(\kappa(x)))$.\footnote{Note that we assume wlog that $f$ is monotonically increasing.} Moreover, $|y_2|\in O(f'(\kappa(x))\cdot |x|)$. Inserting into Equation~\ref{eq:lin_time_transitive} and by elementary calculations, it is easy to conclude that there is a computable function $g$ such that the total running time is bounded by $g(\kappa(x))\cdot O(\log^d(|x|)\cdot |x|^c)$. 
\end{proof}

\subsubsection{Lower Bounds and Fine-Grained Complexity Theory}\label{sec:prelims_fgct}
For this work, we rely on the following two hardness assumptions from fine-grained complexity theory.

\begin{hypothesis}[ETH~\cite{ImpagliazzoP01,ImpagliazzoPZ01}]
    The \emph{Exponential Time Hypothesis} (ETH) asserts that $3\textsc{-SAT}$ cannot be solved in $\exp(o(n))$, where $n$ is the number of variables of the input formula. 
\end{hypothesis}

\begin{hypothesis}[Triangle Detection Conjecture~\cite{Abboud14-TriangleConj}]
    The \emph{Triangle Detection Conjecture} asserts that there is a positive real $\gamma >0$ such that, in the word RAM model of $O(\log n)$ bits, there is no (deterministic or randomised) algorithm that decides whether a graph with $m$ edges contains a triangle in (expected) time $O(m^{1+\gamma})$.
\end{hypothesis}

In addition to running time lower bounds based on the previous assumptions, we will also establish hardness results for the parameterised complexity class $\#\W[1]$, which can be thought of as the parameterised counting equivalent of $\mathrm{NP}$~\cite[Chapter 14]{FlumG06}. A parameterised counting problem $(P,\kappa)$ is $\#\W[1]$-\emph{hard} if $\#\textsc{Clique}\fptred (P,\kappa)$, and it is $\#\W[1]$-\emph{complete} if $\#\textsc{Clique}\fptinterred (P,\kappa)$.\footnote{For readers familiar with structural parameterised complexity theory, we note that, originally, containment in $\#\W[1]$ is defined via parameterised parsimonious reductions~\cite[Chapter 14]{FlumG06}. However, it has since become standard to use parameterised Turing-reductions for the definition of $\#\W[1]$-completeness instead; see \cite[Chapter 2.3.1]{Roth19} for a more comprehensive discussion.} It is well-known that $\#\W[1]$-hard problems are not FPT unless ETH fails~\cite{Chenetal05,Chenetal06,CyganFKLMPPS15}. 

For our analysis of the complexity of $\holantprob$, the following coloured homomorphism counting problem will be a key ingredient.

\begin{definition}[$\#\cphomsprob(\mathcal{H})$]
    Let $\mathcal{H}$ be a class of graphs. The problem $\#\cphomsprob(\mathcal{H})$ expects as input a graph $H\in \mathcal{H}$ and an $H$-coloured graph $(G,h)$, and outputs $\#\cphoms(H \to (G,h))$. The parameter is $|H|$
\end{definition}

We rely on the following (conditional) lower bounds on $\#\cphomsprob(\mathcal{H})$.

\begin{lemma}\label{lem:cphom_lower_bounds}
    Let $\mathcal{H}$ be a recursively enumerable\footnote{This is a standard technical condition, required to avoid discussing non-uniform FPT algorithms and parameterised reductions. All classes considered in this work will be recursively enumerable.} class of graphs.
    \begin{enumerate}
        \item If $\mathcal{H}$ contains a triangle then $\#\cphomsprob(\mathcal{H})$ cannot be solved in FPT-near-linear time, unless the Triangle Conjecture Fails.
        \item If $\mathcal{H}$ has unbounded treewidth, then $\#\cphomsprob(\mathcal{H})$ is $\#\W[1]$-hard and, assuming ETH, cannot be solved in time \[f(|H|)\cdot |V(G)|^{o(\mathsf{tw}(H)/\log(\mathsf{tw}(H)))}\]
        for any function $f$.
    \end{enumerate}
\end{lemma}
\begin{proof}
    We prove both parts separately.
\begin{enumerate}
    \item We provide an easy reduction from detecting a triangle in a graph. Given $G$, construct the (uncoloured) graph Tensor product $G' = G \otimes K_3$, where $K_3$ is the triangle; recall that the vertices of $G'$ are $V(G) \times V(K_3)$ and two pairs $(u,v)$ and $(u',v')$ are adjacent in $G'$ if and only if $u$ and $u'$ are adjacent in $G$ and $v$ and $v'$ are adjacent in $K_3$ (which is equivalent to just $v\neq v'$). Observe that $|V(G')|+|E(G')|\in O(|V(G)|+|E(G)|)$ and that the Tensor product can be computed in linear time.

    Consider the $K_3$-colouring $h$ of $G'$ that just maps $(u,v)$ to $v$. Then, clearly, $G$ contains a triangle if and only if $\#\cphoms(K_3 \to (G',h))>0$. Since $K_3 \in \mathcal{H}$, the proof is concluded.
    \item $\#\W[1]$-hardness follows by the classification of parameterised counting constraint satisfaction problems due to Dalmau and Jonson~\cite{DalmauJ04} (see also~\cite[Lemma 2.45]{Roth19} for an explicit reduction from counting $k$-cliques). The conditional lower bound under ETH holds by a result of Marx~\cite{Marx10} --- note that the latter refers to $\#\cphomsprob(\mathcal{H})$ as the partitioned subgraph problem.\qedhere
\end{enumerate}
\end{proof}

\subsubsection{Algorithms for counting coloured homomorphisms}\label{sec:algo_count_col_homs}
For what follows, given two graphs $H$ and $G$ with (not necessarily proper) vertex colourings $\nu_H$ and $\nu_G$, respectively, we write $\homs{(H,\nu_H)}{(G,\nu_G)}$ for the set of homomorphisms $\varphi$ from $H$ to $G$ that agree on the vertex colourings, i.e., $\nu_H(v)=\nu_G(\varphi(v))$ for all $v\in V(H)$. 

\begin{definition}[$\#\colhomsprob(\mathcal{H})$]
    Let $\mathcal{H}$ be a class of graphs. The problem $\#\colhomsprob(\mathcal{H})$ expects as input a graph $H\in \mathcal{H}$, a graph $G$, and (not necessarily proper) vertex colourings $\nu_H$ and $\nu_G$ of $H$ and $G$. and outputs $\#\homs{(H,\nu_H)}{(G,\nu_G)}$. The parameter is $|H|$.
\end{definition}

It is well-known, in fact, folklore, that counting homomorphisms from a graph $H$ to a graph $G$ can be counted in near-linear time (in $|V(G)|+ |E(G)|$) if $H$ is acyclic (i.e., if it has treewidth $1$) (see e.g.\ \cite[Theorem 7]{BeraGLSS22} for a formal statement and proof). Moreover, the same holds true for the more general problem of counting answers to acyclic conjunctive without quantified variables (see e.g.\ \cite[Theorem 12]{BraultBaron13}). Interpreting vertex-colours as unary predicates, we obtain as an immediate corollary:

\begin{fact}\label{fact:colhoms_lintime}
    $\#\colhomsprob(\mathcal{H})$ can be solved in FPT-near-linear time if $\mathcal{H}$ only contains acyclic graphs.
\end{fact}

Next, we adapt a result due to Curticapean, Dell and Marx~\cite[Theorem 1.7]{CurticapeanDM17} from uncoloured homomorphisms to coloured homomorphisms to obtain, for $\mathcal{H}$ containing graphs of treewidth at most~$2$, an algorithm for $\#\colhomsprob(\mathcal{H})$ via fast matrix multiplication; the proof of (a generalisation of) the following lemma can be found in the Appendix~\ref{sec:appendix_fastMM}.

\begin{lemma}\label{lem:colhom_matrix_multi}
    Let $\mathcal{H}$ be a class of graphs of treewidth at most $2$. Then $\#\colhomsprob(\mathcal{H})$ can be solved in time $f(|H|)\cdot \mathcal{O}(|V(G)|^{\omega})$ for some computable function $f$. \qed
\end{lemma}

\subsection{Parameterised Holant Problems}
For a smoother presentation we will first consider the following types of signatures.
\begin{definition}
    A \emph{signature} is a computable function $s:\mathbb{N} \to \mathbb{Q}$ with $s(0)\neq 0$.
\end{definition}
In Section~\ref{sec:sig0} we show how to deal with signatures $s$ allowing $s(0)=0$. We further point out that, in classical holant literature, the symbol $f$ is used for signature functions. However, since $f$ is co-notated with another role in the world of parameterised algorithms, we decided to use the symbol $s$ instead. 

\begin{definition}[Edge-Coloured Signature Grids]
    Let $\mathcal{S}$ be a finite set of signatures and let $k$ be a positive integer.
    A $k$-\emph{edge-coloured signature grid} over $\mathcal{S}$ is a triple $\Omega=(G,\xi,\{s_v\}_{v\in V(G)})$ of a graph $G$, a mapping $\xi:E(G) \to [k]$ called the $k$-\emph{edge-colouring}, and a collection of signatures $\{s_v\}_{v\in V(G)}$ with $s_v \in \mathcal{S}$ for all $v\in V(G)$.

    A subset of edges $A\subseteq E(G)$ of $G$ is called \emph{colourful} if $|A|=k$ and $\xi(A)=[k]$, that is, $A$ contains precisely one edge per colour.
\end{definition}

\begin{definition}[Edge-Colourful Holants]
    Let $\Omega=(G,\xi,\{s_v\}_{v\in V(G)})$ be a $k$-edge-colourful signature grid. Define
    \[ \holant(\Omega) = \sum_{\substack{A \subseteq E(G)\\ A \text{ colourful}}} \prod_{v\in V(G)} s_v(|A \cap E(v)|) \]
\end{definition}

We are now able to define the parameterised holant problem.

\begin{definition}[$\textsc{p-Holant}(\mathcal{S})$]
    Let $\mathcal{S}$ be a finite set of signatures. The problem $\textsc{p-Holant}(\mathcal{S})$ expects as input a positive integer $k$ and a $k$-edge-coloured signature grid $\Omega=(G,\xi,\{s_v\}_{v\in V(G)})$ over $\mathcal{S}$. The output is $\holant(\Omega)$, and the problem is parameterised by $k$. 
\end{definition}

For technical reasons, we will also consider the following restricted version of $\textsc{p-Holant}(\mathcal{S})$:
\begin{definition}[$\holantprobstar(\mathcal{S})$]
    Let $\mathcal{S}$ be a finite set of signatures. The problem $\holantprobstar(\mathcal{S})$ expects as input a graph $H$, an $H$-coloured graph $(G,h)$, and a collection $\{s_v\}_{v\in V(G)}$ of signatures in $\mathcal{S}$, such that for any pair of vertices $u,v$ of $G$ we have $h(u)=h(v)$ implies $s_u=s_v$. The output is $\holant(G,h,\{s_v\}_{v\in V(G)})$, and the problem is parameterised by $|H|$.
\end{definition}

Note that we slightly abuse notation in the previous definition by using the $H$-colouring $h$ of $G$ as the edge-colouring of the signature grid, that is, we assign an edge $e=\{u,v\}$ of $G$ the colour $h(e):=\{h(u),h(v)\} \in E(H)$. Formally, we can fix any bijection $b:E(H) \to [|E(H)|]$ and define the edge-colouring of the signature grid by setting $\xi(e):=b(h(e))$. For the sake of avoiding notational clutter, we will omit making $b$ explicit and just refer to $h$ as the edge-colouring of the signature grid in the remainder of the paper. 

In this way, observe that $\holantprobstar(\mathcal{S})$ is a restriction of $\holantprob(\mathcal{S})$ since we allow as edge-colourings only those that correspond to an underlying $H$-colouring. Thus, clearly:

\begin{fact}\label{fact:easy_direction_equivalence}
    Let $\mathcal{S}$ be a finite set of signatures. We have
    \[ \holantprobstar(\mathcal{S}) \fptlinred \holantprob(\mathcal{S}) \,.\]
\end{fact}

Finally, we define the \emph{uncoloured} parameterised Holant problem; in what follows an \emph{uncoloured signature grid} over $\mathcal{S}$ is just a pair $\Omega=(G,\{s_v\}_{v\in V(G)})$ with $s_v\in \mathcal{S}$ for all $v\in V(G)$.
\begin{definition}[$\text{\sc{p-UnColHolant}}(\mathcal{S})$]
    Let $\mathcal{S}$ be a finite set of signatures. The problem $\text{\sc{p-UnColHolant}}(\mathcal{S})$ gets as input a positive integer $k$, and a signature grid $\Omega=(G,\{s_v\}_{v\in V(G)})$ over $\mathcal{S}$. The output is
    \[ \holant(\Omega,k):= \sum_{\substack{A \subseteq E(G)\\|A|=k}}~\prod_{v\in V(G)} s_v(|A \cap E(v)|)\,. \]
    The problem parameter is $k$.
\end{definition}

\begin{remark}
    For finite sets of signatures $\mathcal{S}$, the problems $\textsc{p-Holant}(\mathcal{S})$ and $\textsc{p-UnColHolant}(\mathcal{S})$ can always be reduced to $\#\textsc{Clique}$ w.r.t.\ parameterised Turing-reductions, since we will see that both problems can easily be cast as a linear combination of homomorphism counts, which always reduces to $\#\textsc{Clique}$~\cite{DalmauJ04}. For this reason, we will only prove $\#\W[1]$-hardness when establishing $\#\W[1]$-completeness of our Holant problems.
\end{remark}

\section{Equivalence of $\holantprob$ and $\holantprobstar$}\label{sec:col_holants_equivalences}

In this section, we will prove that, for each finite set of signatures $\mathcal{S}$, the problems $\holantprob(\mathcal{S})$ and $\holantprobstar(\mathcal{S})$ are equivalent w.r.t.\ FPT near-linear-time reductions. 

For the proof, we will first introduce a class of coloured graphs that will be extremely useful in the proof of the aforementioned equivalence. We will only rely on this particular family of coloured graphs in the current section.

\subsection{$(\ell_1,\ell_2)$-Coloured Graphs}

An $(\ell_1,\ell_2)$\emph{-coloured graph} is a triple $(G,\nu,\xi)$ of a graph $G$, a mapping $\nu: V(G) \to S_1$ for a set $S_1$ of size $\ell_1$ (called the $\ell_1$-vertex-colouring), and a mapping $\xi: E(G) \to S_2$ for a set $S_2$ of size $\ell_2$ (called the $\ell_2$-edge-colouring). 

\begin{remark}
   An $H$-coloured graph $(G,h)$ naturally induces an $(\ell_1,\ell_2)$-coloured graph $(G,\nu_h,\xi_h)$ where $\nu_h=h$ and $\xi_h(e)=h(e)$ for all $e\in E(H)$.
\end{remark}

A homomorphism from $(G_1,\nu_1,\xi_1)$ to $(G_2,\nu_2,\xi_2)$ is a homomorphism $h\in \homs{G_1}{G_2}$ such that $\nu_2(h(v))=\nu_1(v)$ for all $v\in V(G_1)$ and $\xi_2(h(e))= \xi_1(e)$ for all $e\in E(G)$. We write $\homs{(G_1,\nu_1,\xi_1)}{(G_2,\nu_2,\xi_2)}$ for the set of all such homomorphisms. Embeddings and isomorphisms between $(\ell_1,\ell_2)$-coloured graphs are defined likewise. We write $(G_1,\nu_1,\xi_1)\cong (G_2,\nu_2,\xi_2)$ if $(G_1,\nu_1,\xi_1)$ and $(G_2,\nu_2,\xi_2)$ are isomorphic.

A subgraph of an an $(\ell_1,\ell_2)$-coloured graph $(G,\nu,\xi)$ is an $(\ell_1,\ell_2)$-coloured graph $(H,\nu|_{V(H)},\xi|_{E(H)})$ where $H$ is a subgraph of~$G$. 
We write $\subs{(H,\nu_H,\xi_H)}{(G,\nu_G,\xi_G)}$ for the set of all subgraphs of $(G,\nu_G,\xi_G)$ that are isomorphic to $(H,\nu_H,\xi_H)$.

We write $\auts(G,\nu,\xi)$ for the set of all automorphisms~$a$ of~$G$ such that $\nu(a(v))=\nu(v)$ and $\xi(a(e))=\xi(e)$ for all $v \in V(G)$ and $e\in E(G)$. The following identity is well-known; we include a proof only for reasons of self-containment:

\begin{proposition}\label{prop:col_subs_to_embs}
    $\#\embs{(H,\nu_H,\xi_H)}{(G,\nu_G,\xi_G))} = \#\auts(H,\nu_H,\xi_H)\cdot \#\subs{(H,\nu_H,\xi_H)}{(G,\nu_G,\xi_G)}$.
\end{proposition}
\begin{proof}
    We say that two embeddings $\varphi_1,\varphi_2\in\embs{(H,\nu_H,\xi_H)}{(G,\nu_G,\xi_G))}$ are equivalent if there is an automorphism $a\in \auts(H,\nu_H,\xi_H)$ such that $\varphi_1(v)=\varphi_2(a(v))$ for all $v\in V(H)$. The size of an equivalence class is then $\#\auts(H,\nu_H,\xi_H)$, and each equivalence class corresponds (uniquely) to an element in $\subs{(H,\nu_H,\xi_H)}{(G,\nu_G,\xi_G)}$.

    More formally, $\auts(H,\nu_H,\xi_H)$ acts as a group (with composition) on $\embs{(H,\nu_H,\xi_H)}{(G,\nu_G,\xi_G))}$. This action is a free group action.\footnote{If $\varphi = a \circ \varphi$ then $a$ must be the identity.} Moreover, the orbits clearly correspond one-to-one to the elements of $\subs{(H,\nu_H,\xi_H)}{(G,\nu_G,\xi_G)}$.
\end{proof}

Next we extend the notion of quotient graphs to $(\ell_1,\ell_2)$-coloured graphs; this requires us to restrict to partitions that do not identify vertices or edges with distinct colours.

\begin{definition}[Colour-Consistent Partitions]
    Let $(H,\nu,\xi)$ be an $(\ell_1,\ell_2)$-coloured graph. A partition $\rho \in \ppart(H)$ is called \emph{colour-consistent} (w.r.t.\ $\nu$ and $\xi$) if the following two constraints are satisfied:
    \begin{itemize}
        \item[(I)] If two vertices $u$ and $v$ are in the same block of $\rho$, then $\nu(u)=\nu(v)$.
        \item[(II)] If two edges $e_1$ and $e_2$ of $H$ are mapped to the same edge of $H/\rho$ by $h_\rho$, then $\xi(e_1)=\xi(e_2)$. 
    \end{itemize}
    We write $\ppart(H,\nu,\xi)$ for the set of all colour-consistent partitions of $(H,\nu,\xi)$.
\end{definition}

\begin{definition}[Quotient Graphs of $(\ell_1,\ell_2)$-Coloured Graphs]
    Given an $(\ell_1,\ell_2)$-coloured graph $(H,\nu,\xi)$ and a colour-consistent partition $\rho \in \ppart(H,\nu,\xi)$ we define $(H,\nu,\xi)/\rho := (H/\rho,\nu/\rho,\xi/\rho)$, where $\nu/\rho$ assigns a block the colour of its members, and $\xi/\rho$ assigns an edge $e\in E(H/\rho)$ the colour of the edges of $H$ that are mapped to $e$ by $h_\rho$.
\end{definition}
Observe that the previous construction of quotients for $(\ell_1,\ell_2)$-coloured graphs is well-defined as colour-consistent partitions can only lead to identifications of vertices and edges with the same colours.

Given an $(\ell_1,\ell_2)$-coloured graph $(H,\nu,\xi)$, we consider the poset of colour-consistent partitions $\ppart(H,\nu,\xi)$ with partition refinement. We write $\mu_{(H,\nu,\xi)}$ for the M\"obius function of this poset,\footnote{See e.g.\ \cite[Chapter 3.7]{Stanley11} for the definition of the M\"obius function of a poset. Since we do not need any additional properties of $\mu_{(H,\nu,\xi)}$, we avoid stating the definition in this paper.} and we simplify notation by setting $\mu_{(H,\nu,\xi)}(\rho) := \mu_{(H,\nu,\xi)}(\bot,\rho)$. The proof of the subsequent transformation follows by M\"obius inversion over the poset of colour-consistent partitions and reads almost verbatim as the proof for the uncoloured setting~\cite[Chapter 5.2.3]{Lovasz12}.

\begin{lemma}\label{lem:col_embs_to_homs}
    Let $(H,\nu_H,\xi_H)$ and $(G,\nu_G,\xi_G)$ be  $(\ell_1,\ell_2)$-coloured graphs. We have
    \[ \#\embs{(H,\nu_H,\xi_H)}{(G,\nu_G,\xi_G)} = \!\!\!\!\sum_{\rho \in \ppart(H,\nu_H,\xi_H)}\!\!\!\! \mu(\rho) \cdot \#\homs{(H,\nu_H,\xi_H)/\rho}{(G,\nu_G,\xi_G)} \]
    where $\mu=\mu_{(H,\nu_H,\xi_H)}$. \qed
\end{lemma}

\begin{corollary}\label{cor:col_sub_to_hom}
     Let $(H,\nu_H,\xi_H)$ and $(G,\nu_G,\xi_G)$ be  $(\ell_1,\ell_2)$-coloured graphs. We have
    \[ \#\subs{(H,\nu_H,\xi_H)}{(G,\nu_G,\xi_G)} = \#\auts(H,\nu_H,\xi_H)^{-1}\!\!\!\!\!\!\sum_{\rho \in \ppart(H,\nu_H,\xi_H)}\!\!\!\!\!\! \mu(\rho) \cdot \#\homs{(H,\nu_H,\xi_H)/\rho}{(G,\nu_G,\xi_G)} \]
    where $\mu=\mu_{(H,\nu_H,\xi_H)}$.
\end{corollary}
\begin{proof}
    Follows by Proposition~\ref{prop:col_subs_to_embs} and Lemma~\ref{lem:col_embs_to_homs}.
\end{proof}

Next, we establish some required algebraic properties of $(\ell_1,\ell_2)$-coloured graphs.
\begin{definition}[$\Gamma(\ell_1,\ell_2,S)$]
    Let $\ell_1,\ell_2$ be positive integers, and let $S$ be a set of size $\ell_1$. We define $\Gamma(\ell_1,\ell_2,S)$ as the set of all isomorphism types of $(\ell_1,\ell_2)$-coloured graphs $(H,\nu,\xi)$ with $\nu: V(H)\to S$ and $\xi: E(H) \to [\ell_2]$. Let also $\Gamma_{\text{\sf{inj}}}(\ell_1, \ell_2, S)$ denote the subset of $\Gamma(\ell_1,\ell_2,S)$ containing those graphs with \emph{injective} edge-colouring $\xi$. 
\end{definition}

\noindent Given $(H,\nu_H,\xi_H),(F,\nu_F,\xi_F) \in \Gamma(\ell_1,\ell_2,S)$, we define their \emph{Tensor product} $(H,\nu_H,\xi_H)\otimes(F,\nu_F,\xi_F)$ as follows:
\begin{itemize}
    \item[(1)] The vertex set is $\{ (u,v) \mid \nu_H(v)=\nu_F(v) \}$,
    and the vertex colouring $\nu_{H\otimes F}$ assigns a vertex $(u,v)$ the colour $\nu_H(u)(=\nu_F(v))$.
    \item[(2)] Two vertices $(u,v)$ and $(u',v')$ are made adjacent if $\{u,u'\}\in E(H)$ and $\{(v,v')\}\in E(F)$, and $\xi_H(\{u,u'\})=\xi_F(\{v,v'\})$. The edge colouring $\xi_{H\otimes F}$ assigns an edge $\{(u,v),(u',v')\}$ the colour $\xi_H(\{u,u'\})(=\xi_F(\{v,v'\}))$.
\end{itemize}

\begin{proposition}\label{prop:semigroup}
     Let $\ell_1,\ell_2$ be positive integers, and let $S$ be a set of size $\ell_1$. Then $(\Gamma(\ell_1,\ell_2,S),\otimes)$ is a semigroup.
\end{proposition}
\begin{proof}
    $\Gamma(\ell_1,\ell_2,S)$ is clearly closed under the Tensor product.
    Thus we only need to show associativity, which is also immediate: the isomorphism from $(H,\nu_H,\xi_H)\otimes\big((F,\nu_F,\xi_F) \otimes (G,\nu_G,\xi_G)\big)$ to $\big((H,\nu_H,\xi_H)\otimes(F,\nu_F,\xi_F)\big) \otimes (G,\nu_G,\xi_G)$ is given by $(a,(b,c))\mapsto ((a,b),c)$.
\end{proof}

\begin{proposition}\label{prop:linear}
     Let $(F,\nu_F,\xi_F),(G,\nu_G,\xi_G),(H,\nu_H,\xi_H) \in \Gamma(\ell_1,\ell_2,S)$. We have 
     \begin{align*}
         ~&~\#\homs{(F,\nu_F,\xi_F)}{(G,\nu_G,\xi_G)\otimes (H,\nu_H,\xi_H)} \\
         =&~\#\homs{(F,\nu_F,\xi_F)}{(G,\nu_G,\xi_G)} \cdot \#\homs{(F,\nu_F,\xi_F)}{(H,\nu_H,\xi_H)}\,. 
     \end{align*}
\end{proposition}
\begin{proof}
   Consider the mapping $b: h \mapsto (h_1,h_2)$ where $h_1(v)=\pi_1(h(v))$ and $h_2(v)=\pi_2(h(v))$ are the projections of $h$ to the first and second component, respectively. It is easy to see that $b$ is a bijection from $\homs{(F,\nu_F,\xi_F)}{(G,\nu_G,\xi_G)\otimes (H,\nu_H,\xi_H)}$ to  $\homs{(F,\nu_F,\xi_F)}{(G,\nu_G,\xi_G)} \times \homs{(F,\nu_F,\xi_F)}{(H,\nu_H,\xi_H)}$.
\end{proof}

The final ingredient is the following proposition. Its  proof follows the same lines as the classical argument of~\lovasz (see~\cite[Chapter 5.4]{Lovasz12}), but for reasons of self-containment we include a proof.

\begin{proposition}\label{prop:distinct}
    Let $(F,\nu_F,\xi_F),(H,\nu_H,\xi_H) \in \Gamma_{\text{\sf{inj}}}(\ell_1,\ell_2,S)$. If $(F,\nu_F,\xi_F)\ncong (H,\nu_H,\xi_H)$ then there exists $(G,\nu_G,\xi_G)\in \Gamma_{\text{\sf{inj}}}(\ell_1,\ell_2,S)$ such that
    \[ \#\homs{(F,\nu_F,\xi_F)}{(G,\nu_G,\xi_G)} \neq \#\homs{(H,\nu_H,\xi_H)}{(G,\nu_G,\xi_G)} \,.\]
\end{proposition}
\begin{proof}
    For any pair of $(\ell_1,\ell_2)$-coloured graphs $(A,\nu_A,\xi_A),(G,\nu_G,\xi_G) \in \Gamma_{\text{\sf{inj}}}(\ell_1,\ell_2,S)$ we set
    \[ \surhoms{(A,\nu_A,\xi_A)}{(G,\nu_G,\xi_G)} := \{ \varphi \in \homs{(A,\nu_A,\xi_A)}{(G,\nu_G,\xi_G)} \mid \varphi \text{ is surjective} \} \]
    Given a subset $J \subseteq V(G)$, we set $(G,\nu_G,\xi_G)[J] := (G[J],\nu_G|_{J},\xi_G|_{E(G[J])})$,
    where $G[J]$ is the subgraph of $G$ induced by $J$, that is, $V(G[J])=J$ and $E(G[J])=\{e \in E(G) \mid e \subseteq J\}$. Note that $(G,\nu_G,\xi_G)[J]$ is still an element of $\Gamma_{\text{\sf{inj}}}(\ell_1,\ell_2,S)$.
    By the principle of inclusion and exclusion, we have
    \[\#\surhoms{(A,\nu_A,\xi_A)}{(G,\nu_G,\xi_G)} = \sum_{J \subseteq V(G)} (-1)^{|V(G)\setminus J|} \cdot \#\homs{(A,\nu_A,\xi_A)}{(G,\nu_G,\xi_G)[J]} \,.\]
    Now assume for contradiction that for all $(G,\nu_G,\xi_G) \in \Gamma_{\text{\sf{inj}}}(\ell_1,\ell_2,S)$ we have
    \[ \#\homs{(F,\nu_F,\xi_F)}{(G,\nu_G,\xi_G)} = \#\homs{(H,\nu_H,\xi_H)}{(G,\nu_G,\xi_G)} \,.\]
    Then, in particular, we have
    \begin{align*}
        \#\surhoms{(F,\nu_F,\xi_F)}{(H,\nu_H,\xi_H)} &= \sum_{J \subseteq V(H)} (-1)^{|V(H)\setminus J|} \cdot \#\homs{(F,\nu_F,\xi_F)}{(H,\nu_H,\xi_H)[J]}\\
        ~& = \sum_{J \subseteq V(H)} (-1)^{|V(H)\setminus J|} \cdot \#\homs{(H,\nu_H,\xi_H)}{(H,\nu_H,\xi_H)[J]}\\
        ~&= \#\surhoms{(H,\nu_H,\xi_H)}{(H,\nu_H,\xi_H)} > 0\,.
    \end{align*}
    Similarly, we have $\#\surhoms{(H,\nu_H,\xi_H)}{(F,\nu_F,\xi_F)}>0$. 

    Finally, it is easy to see that $(F,\nu_F,\xi_F)$ and $(H,\nu_H,\xi_H)$ must already be isomorphic if there are surjective homomorphisms from $(F,\nu_F,\xi_F)$ to $(H,\nu_H,\xi_H)$ and from $(H,\nu_H,\xi_H)$ to $(F,\nu_F,\xi_F)$. This yields the desired contradiction and concludes the proof. 
\end{proof}

\subsection{Statement and Proof of the Equivalence}

\begin{lemma}\label{lem:main_equivalence_hard}
    Let $\mathcal{S}$ be a finite set of signatures. We have
    
    \[ \holantprob(\mathcal{S}) \fptinterlinred \holantprobstar(\mathcal{S}) \]
\end{lemma}
\begin{proof}
By Fact~\ref{fact:easy_direction_equivalence}, we have $\holantprobstar(\mathcal{S}) \fptlinred \holantprob(\mathcal{S})$. In what follows, we prove the backwards direction.
    Let $k$ and $(G,\xi,\{s_v\}_{v\in V(G)})$ be the input to $\textsc{p-Holant}(\mathcal{S})$.
    Let $\ell=|\mathcal{S}|$, fix any ordering $s_1,\dots,s_\ell$ of $\mathcal{S}$ and, for each $i\in[\ell]$, we set $n^G_i$ as the number of vertices $v$ of $G$ with $s_v=s_i$. 
    
    Our goal is to compute 
     \[ \holant(\Omega) = \sum_{\substack{A \subseteq E(G)\\ A \text{ colourful}}} \prod_{v\in V(G)} s_v(|A \cap E(v)|) \]
     It will be convenient for the proof to consider the function $\nu(v):=s_v$ as an $\ell$-vertex colouring of $G$. In particular, we will consider the $(\ell,k)$-coloured graph $(G,\nu,\xi)$.

     Let $A\subseteq E(G)$ be a colourful edge-subset of $G$, and let $G\{A\}$ be the graph obtained from $(V(G),A)$ by deleting isolated vertices, that is, we just delete all edges not contained in $A$, and then we delete all isolated vertices. We say that two colourful edge-subsets $A_1$ and $A_2$ are equivalent, denoted by $A_1 \sim A_2$, if $(G\{A_1\},\nu|_{V(G\{A_1\})},\xi|_{A_1}) \cong (G\{A_2\},\nu|_{V(G\{A_1\})},\xi|_{A_1})$. Given an $(\ell,k)$-coloured graph $(H,\nu_H,\xi_H)$ with $|E(H)|=k$, $\nu_H: V(H) \to \mathcal{S}$ and bijective $\xi_H: E(H) \to [k]$, we write $[(H,\nu_H,\xi_H)]$ for the equivalence class of $\sim$ containing all $A$ with $(G\{A\},\nu|_{V(G\{A\})},\xi|_{A}) \cong (H,\nu_H,\xi_H)$ --- note that, for avoiding notational clutter, we are slightly abusing notation here, since $[(H,\nu_H,\xi_H)]$ might be empty.

     Now fix such an $(\ell,k)$-coloured graph $(H,\nu_H,\xi_H)$. For each $i\in[\ell]$ let $n^H_i$ be number of vertices of $H$ coloured by $\nu_H$ with $s_i$. Recalling that $\nu_H$ maps the vertices of $H$ to signatures in $\mathcal{S}$, we observe that for each $A\in [(H,\nu_H,\xi_H)]$ we have
     \[ \prod_{v\in V(G)}s_v(|A \cap E(v)|)= \prod_{u \in V(H)}\nu_H(u)(d_H(u)) \cdot \prod_{i=1}^\ell s_i(0)^{n^G_i-n^H_i}  \,.\]

     Let us write $C_k$ for the set of all (isomorphism types of) $(\ell,k)$-coloured graph $(H,\nu_H,\xi_H)$ with $|E(H)|=k$, $\nu_H: V(H) \to \mathcal{S}$ and bijective $\xi_H: E(H) \to [k]$.
     Then we have
     \[\holant(\Omega)= \prod_{i=1}^\ell s_i(0)^{n^G_i}\cdot \sum_{(H,\nu_H,\xi_H)\in C_k}|[(H,\nu_H,\xi_H)]| \cdot \prod_{u \in V(H)}\nu_H(u)(d_H(u)) \cdot \prod_{i=1}^\ell s_i(0)^{-n^H_i}\,. \]
     Observe that, by definition of $\sim$, we have
     \[ [(H,\nu_H,\xi_H)]| = \#\subs{(H,\nu_H,\xi_H)}{(G,\nu,\xi)} \,,\]
     which, by Corollary~\ref{cor:col_sub_to_hom}, implies that
     \[[(H,\nu_H,\xi_H)]| = \#\auts(H,\nu_H,\xi_H)^{-1} \sum_{\rho \in\ppart(H,\nu_H,\xi_H)} \mu(\rho) \cdot \#\homs{(H,\nu_H,\xi_H)/\rho}{(G,\nu,\xi)} \,.\]
     Now, given an $(\ell,k)$-coloured graph $(F,\nu_F,\xi_F)$, we set
     \[ \mathsf{coeff}(F,\nu_F,\xi_F):= \sum_{(H,\nu_H,\xi_H)\in C_k} \prod_{u \in V(H)}\nu_H(u)(d_H(u)) \prod_{i=1}^\ell s_i(0)^{-n^H_i} \cdot \#\auts(H,\nu_H,\xi_H)^{-1} \sum_{\substack{\rho \in \ppart(H,\nu_H,\xi_H)\\ (H,\nu_H,\xi_H)/\rho \cong (F,\nu_F,\xi_F)}} \mu(\rho)  \,.\]
     Note that $(H,\nu_H,\xi_H)/\rho \cong (F,\nu_F,\xi_F)$ is only possible if $(F,\nu_F,\xi_F)$ is contained in $C_k$ as well --- in particular, $\xi_F$ must be surjective since taking a quotient of $(H,\nu_H,\xi_H)$ w.r.t.\ a colour-consistent partition can never delete an edge-colour. Thus, we obtain
     \begin{equation}
         \holant(\Omega)= \prod_{i=1}^\ell s_i(0)^{n^G_i} \cdot \sum_{(F,\nu_F,\xi_F) \in C_k} \mathsf{coeff}(F,\nu_F,\xi_F) \cdot \#\homs{(F,\nu_F,\xi_F)}{(G,\nu,\xi)}\,.
     \end{equation}

     Clearly, $\prod_{i=1}^\ell s_i(0)^{n^G_i}$ can be computed in near-linear time, since the $s_i(0)$ are constants only depending on $\mathcal{S}$, which is fixed. It thus remains to be proved that we can compute 
     \[ \Phi(\Omega):= \sum_{(F,\nu_F,\xi_F)\in C_k} \mathsf{coeff}(F,\nu_F,\xi_F) \cdot \#\homs{(F,\nu_F,\xi_F)}{(G,\nu,\xi)}\,,\]
     in FPT-near-linear time using our oracle for $\holantprobstar(\mathcal{S})$.

     First of all, for $(F,\nu_F,\xi_F)\in C_k$, the term $\mathsf{coeff}(F,\nu_F,\xi_F)$ only depends on $k$ and $\mathcal{S}$ and so it can be computed in time only depending on $k$, as $\mathcal{S}$ is fixed. Moreover, the same is true for the set $C_k$. Hence it suffices to compute the terms \[
     T[(F,\nu_F,\xi_F),(G,\nu,\xi)] := \mathsf{coeff}(F,\nu_F,\xi_F) \cdot \#\homs{(F,\nu_F,\xi_F)}{(G,\nu,\xi)}\] for all $(F,\nu_F,\xi_F)\in C_k$.

     To this end, recall that $\Gamma_{\text{\sf{inj}}}(\ell,k,\mathcal{S})$ contains all isomorphism types of $(\ell,k)$-coloured graphs $(H,\nu_H,\xi_H)$ with $\nu_H: V(H) \to \mathcal{S}$ and \emph{injective} $\xi_H:E(H)\to [k]$. 
     We write $\Omega\otimes (H,\nu_H,\xi_H)$ for the signature grid obtained from $(G,\nu,\xi)\otimes (H,\nu_H,\xi_H)$ by equipping a vertex $(u,v)$ with the signature $\nu(u)$ (which is equal to $\nu_H(v)$ by definition of $\otimes$). Let $\hat{G}$ be the underlying (uncoloured) graph of $(G,\nu,\xi)\otimes (H,\nu_H,\xi_H)$.\footnote{Note that $\hat{G}$ is \emph{not} the Tensor product $G\otimes H$ of the uncoloured graphs.}

     Now, the crucial property that allows for the use of our oracle is the fact that, for $(H,\nu_H,\xi_H) \in \Gamma_{\text{\sf{inj}}}(\ell,k,\mathcal{S})$, the graph $\hat{G}$ admits a canonical $H$-colouring $h$ by setting $h(u,v)=v$. 
     \begin{claim}\label{claim:help_in_equivalence}
     We have
    \[\holant(\Omega\otimes (H,\nu_H,\xi_H))= \begin{cases}\holant(\hat{G},h,\{s_v\}_{v\in V(\hat{G})}) & \xi_H \text{ is bijective}\\
     0 & \text{otherwise}
     \end{cases}\,.\]
     \end{claim}
     \begin{claimproof}
     The second case above refers to the situation in which $(H,\mu_H,\xi_H)$ is missing one of the $k$ edge-colours (since $\xi_H$ must be injective in any case). Then $\Omega\otimes (H,\nu_H,\xi_H)$ will miss this colour too, and the holant value is just $0$ since no colourful edge-subsets exists. 
    
     For the first case observe that, if $\xi_H$ is bijective then $H$ contains exactly $k$ edges, each one coloured with a unique colour in~$[k]$. Hence, a $k$-edge-subset of $\hat{G}$ is colourful w.r.t.\ the edge-colouring of the Tensor product if and only if it is colourful w.r.t.\ the $H$-colouring $h$; this shows the first case and concludes the proof of this claim.
     \end{claimproof}
     
     Next, let $(u,v),(u',v')$ be two vertices of $\hat{G}$ and assume $h(u,v)=h(u',v')$, implying that $v=v'$, and thus, by definition of $\otimes$, we have $\nu_{G\otimes H}(u,v)=\nu_H(v)=\nu_{G\otimes H}(u',v)$. This shows that $(\hat{G},h,\{s_v\}_{v\in V(\hat{G})})$ is an instance of $\holantprobstar(\mathcal{S})$. 
     
     Thus, for each $(H,\nu_H,\xi_H) \in \Gamma_{\text{\sf{inj}}}(\ell,k,\mathcal{S})$, the term ${\holant(\Omega\otimes (H,\nu_H,\xi_H))}\cdot ({\prod_{i=1}^\ell s_i(0)^{n^{\hat{G}}_i}})^{-1}$ can be computed by either outputting $0$ or by using our oracle, depending on the cases of Claim~\ref{claim:help_in_equivalence}. This quantity equals
     \begin{align*}
        \Phi(\Omega \otimes (H,\nu_H,\xi_H))) &= \sum_{(F,\nu_F,\xi_F)\in C_k}\mathsf{coeff}(F,\nu_F,\xi_F) \cdot \#\homs{(F,\nu_F,\xi_F)}{(G,\nu,\xi)\otimes (H,\nu_H,\xi_H)} \\
        ~&=\sum_{(F,\nu_F,\xi_F)\in C_k} T[(F,\nu_F,\xi_F),(G,\nu,\xi)] \cdot \#\homs{(F,\nu_F,\xi_F)}{(H,\nu_H,\xi_H)}
     \end{align*}
     where the last equation holds by Proposition~\ref{prop:linear}.
     Since $C_k \subseteq \Gamma_{\text{\sf{inj}}}(\ell,k,\mathcal{S})$, and by Propositions~\ref{prop:semigroup}-\ref{prop:distinct}, all conditions for the application of Dedekind Interpolation~\cite[Theorem 18]{BLR2023stoc} are satisfied, which allows us to compute all $T[(F,\nu_F,\xi_F),(G,\nu,\xi)]$ in FPT-(near)-linear time. While we omit the details of Dedekind Interpolation, we emphasize that the oracle we invoke is given by the following map \[\Phi(\Omega\,\otimes\,\star) : \Gamma(\ell, k, \mathcal{S}) \mapsto \sum_{(F,\nu_F,\xi_F)\in C_k} T[(F,\nu_F,\xi_F),(G,\nu,\xi)] \cdot \#\homs{(F,\nu_F,\xi_F)}{\star}\,.\]
     In particular, the oracle is only queried for graphs of the form $\bigotimes_{i=1}^m(H_i, \nu_{H_i}, \xi_{H_i})$, where $m$ depends only on $k$ and each $(H_i, \nu_{H_i}, \xi_{H_i})$ is a member of $\Gamma_{\text{\sf{inj}}}(\ell, k, \mathcal{S})$. Note that the signature grid $\Omega \otimes \left(\bigotimes_{i=1}^m(H_i, \nu_{H_i}, \xi_{H_i})\right)$ is the same as $\left(\Omega\,\otimes\,\left(\bigotimes_{i=1}^{m-1}(H_i, \nu_{H_i}, \xi_{H_i})\right)\right)\,\otimes\,(H_m, \nu_{H_m}, \xi_{H_m})$, which follows from associativity of the tensor product. Since $\xi_{H_m}$ is injective, we can compute $\Phi\left(\Omega \otimes \left(\bigotimes_{i=1}^m(H_i, \nu_{H_i}, \xi_{H_i})\right)\right)$ according to \Cref{claim:help_in_equivalence}. Finally, the promised runtime complexity is guaranteed by observing that the number of vertices (resp. the number of edges) of the signature grid $\Omega \otimes \left(\bigotimes_{i=1}^m(H_i, \nu_{H_i}, \xi_{H_i})\right)$ is at most $g(k)\cdot |V(\Omega)|$ (resp. $g(k)\cdot|E(\Omega)|)$, for some computable function $g$ depending only on $k$, which concludes the proof.
\end{proof}

\section{Classification for $\holantprobstar$}

We start with the following transformation, which can be considered a weighted version of the (first part of the) transformation in~\cite[Lemma 4.1]{PeyerimhoffRSSVW23}, and which follows similar arguments. However, due to various technicalities regarding the vertex signatures, we provide a proof nevertheless. 

\begin{lemma}\label{lem:holant_star_to_cpembs}
Let $H$, $(G,h,\{s_v\}_{v\in V(G)})$ be an instance of $\holantprobstar(\mathcal{S})$ for some finite set of signatures $\mathcal{S}$. Assume that $V(H)=\{v_1,\dots,v_z\}$, and, for each $i\in [z]$, set $n_i$ as the number of vertices of $G$ coloured by $h$ with $v_i$, and let $s_i$ be the signature of the vertices coloured by $h$ with $v_i$. 
    Then
    \[ \holant(G,h,\{s_v\}_{v\in V(G)}) = \prod_{i=1}^z s_i(0)^{n_i} \cdot \sum_{\vec{\sigma}\in \mathcal{F}(H)} \#\embscp(\fracture{H}{\vec{\sigma}} \to (G,h)) \cdot \left(\prod_{i=1}^z \prod_{B \in \vec{\sigma}(v_i)} \frac{s_i(|B|)}{s_i(0)}  \right) \]
\end{lemma}
\begin{proof}
    Let $k=|E(H)|$ and let $A$ be a colourful edge-subset of $G$ w.r.t.\ $h$, that is, $|A|=k$ and $h(A)=[k]$. 

    Observe that each such $A$ induces a fracture $\vec{\sigma}_A$ of $H$, defined as follows. Let $v\in V(H)$, and recall that $E_H(v)$ is the set of all edges of $H$ incident to $v$.
    Let us furthermore denote the elements of $E_H(v)$ by $e^H_1,\dots,e^H_d$, where $d$ is the degree of $v$ in $H$. Since $A$ is colourful w.r.t.\ $h$, there are (pairwise distinct) edges $e^G_1,\dots,e^G_d$ in $A$ such that $h(e^G_i)=e^H_i$ for all $i\in[d]$.
    
    Recall that $\vec{\sigma}_A(v)$ is a partition of $E_H(v)$. We put two edges $e^H_i$ and $e^H_j$ into the same block of $\vec{\sigma}_A(v)$ if and only if the endpoints of $e^G_i$ and $e^G_j$ that are coloured by $h$ with $v$ are equal. Formally, let $e^G_i=\{x_i,y_i\}$ and $e^G_j=\{x_j,y_j\}$. By definition, $h(\{x_i,y_i\})=e^H_i$ and $h(\{x_j,y_j\})=e^H_j$. Therefore one vertex of both edges $e^G_i$ and $e^G_j$ must be mapped by $h$ to $v$; assume w.l.o.g.\ that $h(x_i)=h(x_j)=v$. Then we put two edges $e^H_i$ and $e^H_j$ into the same block of $\vec{\sigma}_A(v)$ if and only if $x_i=x_j$ (meaning that $e^G_i$ and $e^G_j$ share a vertex coloured by $h$ with $v$).

    For what follows, we say that two colourful edge-subsets $A_1$ and $A_2$ of $G$ are equivalent, denoted by $A_1\sim A_2$ if $\vec{\sigma}_{A_1} = \vec{\sigma}_{A_2}$. Given a fracture $\vec{\sigma}$ of $H$, we denote $[\vec{\sigma}]$ for the set of all colourful $A$ with $\vec{\sigma}_{A} = \vec{\sigma}$.
    Now recall that
    \begin{align*}
        \holant(G,h,\{s_v\}_{v\in V(G)}) &= \sum_{\substack{A \subseteq E(G)\\ A\text{ colourful}}} \prod_{v\in V(G)} s_v(|A \cap E(v)|)\,.
    \end{align*}
    \begin{claim}
        We have
        \[ \prod_{v\in V(G)} s_v(|A \cap E(v)|) = \prod_{i=1}^z s_i(0)^{n_i - |\vec{\sigma}_A(v_i)|} \prod_{B\in \vec{\sigma}_A(v_i)} s_i(|B|)  \,.\]
    \end{claim}
    \begin{claimproof}
        For $i\in[z]$ let $V_i\subseteq V(G)$ be the set containing all vertices of $G$ coloured by $h$ with $v_i$ (in particular, this implies that $n_i=|V_i|$). Next consider $ \prod_{v\in V_i} s_v(|A \cap E(v)|)$ and recall that  all $v\in V_i$ have signature $s_v = s_i$. Now observe that for each block $B\in \vec{\sigma}_A(v_i)$ there is a vertex $v\in V_i$ incident to $|B|$ edges of $A$. Moreover, the remaining $n_i-|\vec{\sigma}_A(v_i)|$ vertices of $V_i$ are not incident to any edge in $A$. Thus
        \[\prod_{v\in V_i} s_v(|A \cap E(v)|) = s_i(0)^{n_i-|\vec{\sigma}_A(v_i)|} \cdot \prod_{B\in \vec{\sigma}_A(v_i)} s_i(|B|)\,,\]
        and, consequently,
        \[\prod_{v\in V(G)} s_v(|A \cap E(v)|) = \prod_{i=1}^z \prod_{v\in V_i} s_v(|A \cap E(v)|) = \prod_{i=1}^z s_i(0)^{n_i - |\vec{\sigma}_A(v_i)|} \prod_{B\in \vec{\sigma}_A(v_i)} s_i(|B|)\,. \]
    \end{claimproof}
    Therefore, grouping the colourful edge-subsets along their equivalence classes, we obtain
    \begin{align*}
        \holant(G,h,\{s_v\}_{v\in V(G)}) &= \sum_{\vec{\sigma}\in \mathcal{F}(H)} |[\vec{\sigma}]| \cdot \prod_{i=1}^z s_i(0)^{n_i - |\vec{\sigma}(v_i)|} \prod_{B\in \vec{\sigma}(v_i)} s_i(|B|)\\
        ~&= \prod_{i=1}^z s_i(0)^{n_i} \cdot \sum_{\vec{\sigma}\in \mathcal{F}(H)} |[\vec{\sigma}]| \cdot \left(\prod_{i=1}^z \prod_{B \in \vec{\sigma}(v_i)} \frac{s_i(|B|)}{s_i(0)}  \right) \,.
    \end{align*}
    Finally, recall that $|[\vec{\sigma}]|$ counts the number of colourful edge-subsets of $G$ that induce the fracture $\vec{\sigma}$. Equivalently, this is the number of subgraphs of $G$ that contain each edge colour exactly once, and that are isomorphic to $\fracture{H}{\vec{\sigma}}$. As was shown in the proof of~\cite[Lemma 4.1]{PeyerimhoffRSSVW23}, using the fact that $\fracture{H}{\vec{\sigma}}$ does not have non-trivial automorphisms as an $H$-coloured graph, this number is equal to $\#\embscp(\fracture{H}{\vec{\sigma}} \to (G,h))$, concluding the proof.
\end{proof}

\begin{lemma}\label{lem:holant_star_to_cphoms}
    Let $H$, $(G,h,\{s_v\}_{v\in V(G)})$ be an instance of $\holantprobstar(\mathcal{S})$ for some finite set of signatures $\mathcal{S}$. Assume that $V(H)=\{v_1,\dots,v_z\}$, and, for each $i\in [z]$, set $n_i$ as the number of vertices of $G$ coloured by $h$ with $v_i$, and let $s_i$ be the signature of the vertices coloured by $h$ with $v_i$. 
    Then
    \[ \holant(G,h,\{s_v\}_{v\in V(G)}) = \prod_{i=1}^z s_i(0)^{n_i} \cdot \sum_{\vec{\sigma}\in \mathcal{F}(H)} \sum_{\vec{\rho} \geq \vec{\sigma}} \vec{\mu}(\vec{\sigma},\vec{\rho}) \cdot \#\homscp(\fracture{H}{\vec{\rho}} \to (G,h)) \cdot \left(\prod_{i=1}^z \prod_{B \in \vec{\sigma}(v_i)} \frac{s_i(|B|)}{s_i(0)}  \right) \,.\]
\end{lemma}
\begin{proof}
    Follows immediately by Lemma~\ref{lem:holant_star_to_cpembs} and the following transformation~\cite[Equation (4.1)]{PeyerimhoffRSSVW23}:
    \[\#\embscp(\fracture{H}{\vec{\sigma}} \to (G,h)) = \sum_{\vec{\rho} \geq \vec{\sigma}} \vec{\mu}(\vec{\sigma},\vec{\rho}) \cdot \#\homscp(\fracture{H}{\vec{\rho}} \to (G,h))\,. \]
\end{proof}

Next we collect the coefficient for individual fractures in the previous lemma. 

\begin{definition}
   Let $H$ be a graph, and let $\mathcal{S}$ be a finite set of signatures. Assume that $V(H)=\{v_1,\dots,v_z\}$, and let $\vec{s}=(s_1,\dots,s_z)$ be a $z$-tuple of (not necessarily distinct) signatures in $\mathcal{S}$. For a fracture $\vec{\rho}$ of $H$ we define
   \[ \mathsf{coeff}_{H,\vec{s}}(\vec{\rho})= \sum_{\vec{\sigma} \leq \vec{\rho}} \vec{\mu}(\vec{\sigma},\vec{\rho}) \cdot \left(\prod_{i=1}^z \prod_{B \in \vec{\sigma}(v_i)} \frac{s_i(|B|)}{s_i(0)}  \right) \,.\]
   We might drop the subscript $H,\vec{s}$ if it is clear from the context.
\end{definition}

\begin{corollary}\label{cor:collect_coeffs}
    Let $(H,(G,h,\{s_v\}_{v\in V(G)}))$ be an instance of $\holantprobstar(\mathcal{S})$ for some finite set of signatures $\mathcal{S}$. Assume that $V(H)=\{v_1,\dots,v_z\}$, and, for each $i\in [z]$, set $n_i$ as the number of vertices of $G$ coloured by $h$ with $v_i$, and let $s_i$ be the signature of the vertices coloured by $h$ with $v_i$. Moreover, set $\vec{s}=(s_1,\dots,s_z)$.
    Then
    \[ \holant(G,h,\{s_v\}_{v\in V(G)}) \cdot \prod_{i=1}^z s_i(0)^{-n_i} = \sum_{\vec{\rho}\in \mathcal{F}(H)} \mathsf{coeff}_{H,\vec{s}}(\vec{\rho})  \cdot \#\homscp(\fracture{H}{\vec{\rho}} \to (G,h))  \,.\]
\end{corollary}
\begin{proof}
    We start from Lemma~\ref{lem:holant_star_to_cphoms}, collect the coefficients of $\#\homscp(\fracture{H}{\vec{\rho}} \to (G,h))$ for each fracture $\vec{\rho}$, and divide by $\prod_{i=1}^z s_i(0)^{n_i}$.
\end{proof}

We continue by analysing the coefficients $\mathsf{coeff}(\vec{\rho})$ in detail. For convenience, we first recall and expand the definition of signature fingerprints:

\begin{definition}
    Let $\rho$ be a partition of a finite set, and let $s$ be a signature. We define 
    \[\chi(\rho,s) := \sum_{\sigma \leq \rho} \mu(\sigma,\rho) \cdot \prod_{B \in \sigma} \frac{s(|B|)}{s(0)}\,.\]
    Moreover, given a positive integer $d$, we define the \emph{signature fingerprint} of $d$ and $s$ as follows:
    \[\chi(d,s) := \sum_{\sigma} (-1)^{|\sigma|-1} (|\sigma|-1)! \cdot \prod_{B \in \sigma} \frac{s(|B|)}{s(0)} \,,\]
    where the sum is over all partitions of $[d]$.
\end{definition}

\begin{lemma}\label{lem:just_distributivity}
    Let $\rho$ be a partition of a finite set, let $B_1,\dots,B_t$ be the blocks of $\rho$, and let $s$ be a signature. We have $\chi(\rho,s) =  \prod_{i=1}^t \chi(|B_i|,s)$.
\end{lemma}
\begin{proof}
    Given a partition $\sigma$ with $\sigma \leq \rho$, we set $\sigma^1,\dots,\sigma^t$ be the sub-partitions of $\sigma$ such that $\sigma^i$ partitions $B_i$ (that is, $\sigma^i \leq \{B_i\}$) for each $i\in[z]$. The explicit formula for the M\"obius function of the partition lattice (see e.g.\ \cite[Chapter 3]{Stanley11}) then states that
    \[ \mu(\sigma,\rho) = \prod_{i=1}^t (-1)^{|\sigma^i|-1}(|\sigma^i|-1)! \,.\]
    For $i\in[t]$ we set
    \[ T_i := (-1)^{|\sigma^i|-1}(|\sigma^i|-1)! \cdot \prod_{B \in \sigma^i} \frac{s(|B|)}{s(0)}\,, \]
    and we observe
    \begin{align*}
        \chi(\rho,s) &= \sum_{\sigma \leq \rho} \left(\prod_{i=1}^t (-1)^{|\sigma^i|-1}(|\sigma^i|-1)!\right) \cdot \prod_{B \in \sigma} \frac{s(|B|)}{s(0)}= \sum_{\sigma^1 \leq \{B_1\}} \dots \sum_{\sigma^t \leq \{B_t\}} \prod_{i=1}^t T_i\\
        ~&= \sum_{\sigma^1 \leq \{B_1\}}\left( T_1~ \dots \sum_{\sigma^{t-1} \leq \{B_{t-1}\}}\left( T_{t-1} \sum_{\sigma^t \leq \{B_t\}}T_t \right)\dots \right) = \sum_{\sigma^1 \leq \{B_1\}}\left( T_1~ \dots \sum_{\sigma^{t-1} \leq \{B_{t-1}\}} T_{t-1} \cdot \chi(|B_t|,s) \dots \right)\\
        ~&= \chi(|B_t|,s) \cdot \sum_{\sigma^1 \leq \{B_1\}}\left( T_1~ \dots \sum_{\sigma^{t-1} \leq \{B_{t-1}\}} T_{t-1}  \dots \right) = \dots = \prod_{i=1}^t \chi(|B_i|,s)\,,
    \end{align*}
    where the last equality follows by just inductively applying distributivity.
\end{proof}

\begin{lemma}\label{lem:coeffs_done}
    Let $H$ be a graph with vertices $V(H)=\{v_1,\dots,v_z\}$, let $\vec{\rho}$ be a fracture of $H$, and let $\vec{s}=(s_1,\dots,s_z)$ be a $z$-tuple of signatures. We have
    \[\mathsf{coeff}_{H,\vec{s}}(\vec{\rho}) = \prod_{i=1}^z \prod_{B\in \vec{\rho}(v_i)} \chi(|B|,s_i) \,.\]
\end{lemma}
\begin{proof}
    For ease of notation, assume w.l.o.g.\ that $V(H)=\{1,\dots,z\}$ and, given a fracture $\vec{\sigma}$ of $H$, set $\sigma_i := \vec{\sigma}(i)$. Recall that fractures $\vec{\sigma}$ of $H$ are then vectors $\vec{\sigma}=(\sigma_1,\dots,\sigma_z)$. Moreover, as was shown in~\cite{PeyerimhoffRSSVW23}, given two fractures $\vec{\sigma}$ and $\vec{\rho}$, we have
    \[\vec{\mu}(\vec{\sigma},\vec{\rho}) = \prod_{i=1}^z\mu(\sigma_i,\rho_i) \,.\]
    Recall that $\mu$ denotes the M\"obius function of the partition lattice. Thus we have
    \begin{align*}
        \mathsf{coeff}(\vec{\rho})&= \sum_{\sigma_1 \leq \rho_1} \dots \sum_{\sigma_z \leq \rho_z} \prod_{i=1}^z\mu(\sigma_i,\rho_i)\cdot \prod_{B \in \sigma_i} \frac{s_i(|B|)}{s_i(0)}\\
        ~&=\sum_{\sigma_1 \leq \rho_1}\left(\mu(\sigma_1,\rho_1) \cdot \prod_{B \in \sigma_1} \frac{s_1(|B|)}{s_1(0)}~ \dots \sum_{\sigma_{z-1} \leq \rho_{z-1}} \left(\mu(\sigma_{z-1},\rho_{z-1})\cdot \prod_{B \in \sigma_{z-1}} \frac{s_{z-1}(|B|)}{s_{z-1}(0)} \sum_{\sigma_z \leq \rho_z} \mu(\sigma_z,\rho_z)\cdot \prod_{B \in \sigma_z} \frac{s_z(|B|)}{s_z(0)}\right)\dots\right)\\
         ~&=\sum_{\sigma_1 \leq \rho_1}\left(\mu(\sigma_1,\rho_1) \cdot \prod_{B \in \sigma_1} \frac{s_1(|B|)}{s_1(0)}~ \dots \sum_{\sigma_{z-1} \leq \rho_{z-1}} \left(\mu(\sigma_{z-1},\rho_{z-1})\cdot \prod_{B \in \sigma_{z-1}} \frac{s_{z-1}(|B|)}{s_{z-1}(0)} \cdot \chi(\rho_z,s_z)\right)\dots\right)\\
         ~&= \chi(\rho_z,s_z) \cdot \sum_{\sigma_1 \leq \rho_1}\left(\mu(\sigma_1,\rho_1) \cdot \prod_{B \in \sigma_1} \frac{s_1(|B|)}{s_1(0)}~ \dots \sum_{\sigma_{z-1} \leq \rho_{z-1}} \left(\mu(\sigma_{z-1},\rho_{z-1})\cdot \prod_{B \in \sigma_{z-1}} \frac{s_{z-1}(|B|)}{s_{z-1}(0)}\right)\dots\right)\\
         ~&\vdots\\
         ~&= \prod_{i=1}^z \chi(\rho_i,s_i) = \prod_{i=1}^z \prod_{B\in \rho_i} \chi(|B|,s_i) \,,
    \end{align*}
    where the second-to-last equation follows just from inductively applying distributivity, and the final equation is Lemma~\ref{lem:just_distributivity}.
\end{proof}

Having understood the coefficients of the homomorphism expansion of $\holant(\mathcal{S})$ in terms of the signature fingerprints, we are now able to prove our main classification result.

First of all, we associate each (finite) set of signatures with a class of graphs defined as follows.

\begin{definition}[Homomorphism Supports $\homsupp(\mathcal{S})$ and $\homsupp_\top(\mathcal{S})$]
    Let $\mathcal{S}$ be a finite set of signatures. We say that a graph $F$ is \emph{weakly supported} by $\mathcal{S}$ if there is a graph $H$ with $V(H)=[z]$ (for some $z$), a $z$-tuple $\vec{s}=(s_1,\dots,s_z)$ of signatures in $\mathcal{S}$, and a fracture $\vec{\rho}$ of $H$ such that
    $F \cong \fracture{H}{\vec{\rho}}$ and
        \[\prod_{i=1}^z \prod_{B\in \vec{\rho}(v_i)} \chi(|B|,s_i)\neq 0\,.\]
The \emph{weak homomorphism support} of $\mathcal{S}$, denoted by $\homsupp(\mathcal{S})$, is defined to be the class of all graphs weakly supported by $\mathcal{S}$. 
        
Moreover, we say that a graph $H$ with $V(H)=[z]$  is \emph{strongly supported} by $\mathcal{S}$ if there is a $z$-tuple $\vec{s}=(s_1,\dots,s_z)$ of signatures in $\mathcal{S}$ such that 
    \[  \prod_{i=1}^z \chi(d_i,s_i) \neq 0\,,\]
    where $d_i$ is the degree of the $i$-th vertex. The \emph{strong homomorphism support} of $\mathcal{S}$, denoted by $\homsupp_\top(\mathcal{S})$, is defined to be the class of all graphs strongly supported by $\mathcal{S}$. 
\end{definition}

\begin{lemma}\label{lem:holant_sandwich}
    Let $\mathcal{S}$ be a finite set of signatures. We have
    \[\#\cphomsprob(\homsupp_\top(\mathcal{S})) \fptlinred \holantprobstar(\mathcal{S}) \fptlinred \#\colhomsprob(\homsupp(\mathcal{S})) \,.\]
\end{lemma}
\begin{proof}
    The second direction $\holantprobstar(\mathcal{S}) \fptlinred \#\cphomsprob(\homsupp(\mathcal{S}))$
    is the easier one: Given an instance $(H,(G,h,\{s_v\}_{v\in V(G)}))$ of $\holantprobstar(\mathcal{S})$,
    assume that $V(H)=\{v_1,\dots,v_z\}$, and, for each $i\in [z]$, set $n_i$ as the number of vertices of $G$ coloured by $h$ with $v_i$, and let $s_i$ be the signature of the vertices coloured by $h$ with $v_i$. Moreover, set $\vec{s}=(s_1,\dots,s_z)$. By Corollary~\ref{cor:collect_coeffs}, it suffices to compute  
    \[  \prod_{i=1}^z s_i(0)^{n_i} \cdot  \sum_{\vec{\rho}\in \mathcal{F}(H)} \mathsf{coeff}_{H,\vec{s}}(\vec{\rho})  \cdot \#\homscp(\fracture{H}{\vec{\rho}} \to (G,h)) \,.\]
    By Lemma~\ref{lem:coeffs_done}, we have that $\mathsf{coeff}_{H,\vec{s}}(\vec{\rho})\neq 0$ if and only if $\fracture{H}{\vec{\rho}}\in \homsupp(\mathcal{S})$. Thus all terms with a non-zero coefficient can be computed by just querying our oracle --- recall that $h_{\vec{\rho}}$ denotes the canonical $H$-colouring of $\fracture{H}{\vec{\rho}}$ and consider $(\fracture{H}{\vec{\rho}},h_{\vec{\rho}})$ as vertex coloured graphs; then $\homscp(\fracture{H}{\vec{\rho}} \to (G,h)) = \homs{(\fracture{H}{\vec{\rho}},h_{\vec{\rho}})}{(G,h)}$ and the task of computing its cardinality is clearly an instance of $\#\colhomsprob(\homsupp(\mathcal{S}))$.

    The first direction $\#\cphomsprob(\homsupp_\top(\mathcal{S})) \fptlinred \holantprobstar(\mathcal{S})$ is more interesting and holds by the principle of Complexity Monotonicity for fractured graphs as established in Section 4.1 in~\cite{PeyerimhoffRSSVW23}. Concretely, let $(H,(G,c_G))$ be an input instance of $\#\cphomsprob(\homsupp_\top(\mathcal{S}))$, that is, $H\in \homsupp_\top(\mathcal{S})$, and $(G,c_G)$ is an $H$-coloured graph. Assume w.l.o.g.\ that $V(H)=[z]$ for some positive integer $z$. Since $H\in \homsupp_\top(\mathcal{S})$ 
    there is a vector $\vec{s}=(s_1,\dots,s_z)$ of signatures such that
    $ \prod_{i=1}^z \chi(d_i,s_i) \neq 0$,
    where $d_i$ is the degree of the $i$-th vertex.
    
    The proof follows verbatim the proof of Lemma 4.6 in~\cite{PeyerimhoffRSSVW23}, with the only exception being that we use our oracle for computing the following linear combinations (see Corollary~\ref{cor:collect_coeffs}):
    \[\sum_{\vec{\rho}\in \mathcal{F}(H)} \mathsf{coeff}_{H,\vec{s}}(\vec{\rho})  \cdot \#\homscp(\fracture{H}{\vec{\rho}} \to (G',h'))\,,\]
    where $(G',h')$ are the (coloured) Tensor products of $(G,c_G)$ and the carefully selected graphs used in~\cite[Lemma 4.6]{PeyerimhoffRSSVW23}.
    Finally, note that, by Lemma~\ref{lem:coeffs_done}, we have $\mathsf{coeff}(\vec{\top}) = \prod_{i= 1}^z \chi(d_i,s_i)$; recall that $\vec{\top}$ is the fracture consisting of the coarsest partition (containing only one block) for each vertex of $H$. Since  $ \prod_{i=1}^z \chi(d_i,s_i) \neq 0$, the term 
    \[ \#\homscp(\fracture{H}{\vec{\top}} \to (G',h')) = \#\cphoms(H \to (G',h')) \]
    survives with a non-zero coefficient and can be isolated by solving the system of linear equations used in the proof of~\cite[Lemma 4.6]{PeyerimhoffRSSVW23}. Finally we conclude by pointing out that an inspection of the latter proof reveals that the reduction runs in FPT-near-linear time as promised.
\end{proof}

The final ingredient for our trichotomy result provides upper and lower bounds on the treewidth of graphs in $\homsupp(\mathcal{S})$ and $\homsupp_\top(\mathcal{S})$.

\begin{lemma}\label{lem:hom_support_characterisation}
    Let $\mathcal{S}$ be a finite set of signatures. 
    \item[(1)] If $\mathcal{S}$ is of type $\mathbb{T}[\mathsf{Lin}]$ then $\homsupp(\mathcal{S})$ only contains acyclic graphs.
    \item[(2)] If $\mathcal{S}$ is of type $\mathbb{T}[\omega]$ then $\homsupp_\top(\mathcal{S})$ contains the triangle, but $\homsupp(\mathcal{S})$ only contains graphs of treewidth at most $2$.
    \item[(3)] If $\mathcal{S}$ is of type $\mathbb{T}[\infty]$ then there exists a constant $d\geq 3$ such that $\homsupp_\top(\mathcal{S})$ contains all $d$-regular graphs
\end{lemma}
\begin{proof}
    We prove all cases separately.
    \begin{itemize}
        \item[(1)] If $\mathcal{S}$ is of type $\mathbb{T}[\mathsf{Lin}]$ then $\chi(d,s)=0$ for all $s\in \mathcal{S}$ and $d\geq 2$. Assume for contradiction that there is a graph $F\in \homsupp(\mathcal{S})$ that contains a vertex $u$ of degree at least $2$. Then there is also a graph $H$ with $V(H)=[z]$ for some $z$, a $z$-tuple $\vec{s}=(s_1,\dots,s_z)$ of signatures and a fracture $\vec{\rho}$ of $H$ such that $F \cong \fracture{H}{\vec{\rho}}$ and
        \begin{equation}\label{eq:types_tw_case_1}
            \prod_{i=1}^z \prod_{B\in \vec{\rho}(v_i)} \chi(|B|,s_i)\neq 0\,.
        \end{equation}
        Let $i\in [z]$ and $B \in \vec{\rho}(v_i)$ such that $u$ is mapped to $v^B$ by the isomorphism from $F$ to $\fracture{H}{\vec{\rho}}$. Thus $v^B$ has degree at least $2$ in $\fracture{H}{\vec{\rho}}$. But then $|B|\geq 2$ as well, and thus $\chi(|B|,s_i)=0$, contradicting~(\ref{eq:types_tw_case_1}). 
        Thus $\homsupp(\mathcal{S})$ contains in fact only matchings (i.e., graphs with degree $1$), which are clearly acyclic.
        \item[(2)] The claim that $\homsupp(\mathcal{S})$ contains only graphs of treewidth at most $2$ holds with an analogous argument as in Case (1): We obtain that no graph in $\homsupp(\mathcal{S})$ can contain a vertex of degree at least $3$, but graphs of degree at most $2$ cannot have treewidth $3$ or larger. 

        Hence it only remains to prove that $\homsupp_\top(\mathcal{S})$ contains the triangle.
        Let $s\in \mathcal{S}$ with $\chi(2,s)\neq 0$. Let $H$ be the triangle with $V(H)=\{1,2,3\}$ and set $\vec{s}=(s,s,s)$. Note that $d_1=d_2=d_3 =2$. Thus 
        \[ \prod_{i=1}^3 \chi(d_i,s_i) = \chi(2,s)^3 \neq 0\,,\]
        and hence the triangle is contained in $\homsupp_\top(\mathcal{S})$.
        \item[(3)] If $\mathcal{S}$ is of type $\mathbb{T}[\infty]$ then there are $s\in \mathcal{S}$ and $d\geq 3$ with $\chi(d,s)\neq 0$.
        Let $H$ be any $d$-regular graph and assume w.l.o.g.\ that $V(H)=[z]$ for some $z$. Set $\vec{s}=(s,\dots,s)$ to be $z$ tuple each entry of which is $s$. Since $H$ is $d$-regular we have $d_i=d$ for all $i\in [z]$, and thus
        \[ \prod_{i=1}^z \chi(d_i,s_i) = \chi(d,s)^3 \neq 0\,.\]
        and hence $H$ is contained in $\homsupp_\top(\mathcal{S})$.
    \end{itemize}
\end{proof}

We are now able to prove~\cref{main_thm}, which we restate for the readers convenience.
\begin{theorem}[\cref{main_thm}, restated]
    Let $\mathcal{S}$ be a finite set of signatures.
    \begin{itemize}
        \item[(I)] If $\mathcal{S}$ is of type $\mathbb{T}[\mathsf{Lin}]$, then $\textsc{p-Holant}(\mathcal{S})$ can be solved in FPT-near-linear time, that is, there is a computable function $f$ such that $\textsc{p-Holant}(\mathcal{S})$ can be solved in time $f(k)\cdot \tilde{\mathcal{O}}(|V(\Omega)|+|E(\Omega)|)$.
        \item[(II)] If $\mathcal{S}$ is of type $\mathbb{T}[\omega]$, then $\textsc{p-Holant}(\mathcal{S})$ can be solved in FPT-matrix-multiplication time, that is, there is a computable function $f$ such that $\textsc{p-Holant}(\mathcal{S})$ can be solved in time $f(k)\cdot \mathcal{O}(|V(\Omega)|^{\omega})$. Moreover, $\textsc{p-Holant}(\mathcal{S})$ cannot be solved in time $f(k)\cdot \tilde{\mathcal{O}}(|V(\Omega)|+|E(\Omega)|)$ for any function $f$, unless the Triangle Conjecture fails.
        \item[(III)] Otherwise, that is, if $\mathcal{S}$ is of type $\mathbb{T}[\infty]$, $\textsc{p-Holant}(\mathcal{S})$ is $\#\W[1]$-complete. Moreover, $\textsc{p-Holant}(\mathcal{S})$ cannot be solved in time $f(k)\cdot |V(\Omega)|^{o(k/\log k)}$ for any function $f$, unless ETH fails. 
    \end{itemize}
\end{theorem}
\begin{proof}
Thanks to Lemma~\ref{lem:main_equivalence_hard}, it suffices to prove the classification for $\holantprobstar(\mathcal{S})$. Recall that by Lemma~\ref{lem:holant_sandwich}, we have
\[\#\cphomsprob(\homsupp_\top(\mathcal{S})) \fptlinred \holantprobstar(\mathcal{S}) \fptlinred \#\colhomsprob(\homsupp(\mathcal{S})) \,.\]
For each of the three types, we can thus obtain the lower bound from $\#\cphomsprob(\homsupp_\top(\mathcal{S}))$ and the upper bound from $\#\colhomsprob(\homsupp(\mathcal{S}))$.
\begin{itemize}
    \item[(I)] If $\mathcal{S}$ is of type $\mathbb{T}[\mathsf{Lin}]$, then $\homsupp(\mathcal{S})$ only contains acyclic graphs by Lemma~\ref{lem:hom_support_characterisation}. Thus $\#\colhomsprob(\homsupp(\mathcal{S}))$ can be solved in FPT-linear-time by Fact~\ref{fact:colhoms_lintime}.
    \item[(II)] If $\mathcal{S}$ is of type $\mathbb{T}[\omega]$ then, by Lemma~\ref{lem:hom_support_characterisation}, $\homsupp_\top(\mathcal{S})$ contains the triangle, but $\homsupp(\mathcal{S})$ only contains graphs of treewidth at most $2$. 
    
    By Lemma~\ref{lem:cphom_lower_bounds}~(1), $\#\cphomsprob(\homsupp_\top(\mathcal{S}))$ can thus not be solved in FPT-near-linear time unless the Triangle Conjecture fails.  

    However, by Lemma~\ref{lem:colhom_matrix_multi}, $\#\colhomsprob(\homsupp(\mathcal{S}))$ can be solved in time $f(|H|)\cdot \mathcal{O}(V(|G|)^{\omega})$ for some computable function $f$.
    \item[(III)]
    If $\mathcal{S}$ is of type $\mathbb{T}[\infty]$ then, by Lemma~\ref{lem:hom_support_characterisation}, $\homsupp_\top(\mathcal{S})$ contains, for some $d\geq 3$, all $d$-regular graphs. In particular, it must contain as a subset each family of $d$-regular expander graphs. Since $d$-regular expanders have treewidth linear in their size\footnote{See for instance Proposition 1 in~\cite{Grohe&2009expansion} and set $\alpha=1/2$.}, the desired lower bounds for $\#\cphomsprob(\homsupp_\top(\mathcal{S}))$, and thus for $\holantprobstar(\mathcal{S})$, follow from Lemma~\ref{lem:cphom_lower_bounds}~(2).
\end{itemize}
\end{proof}

\subsection{Consequences for Modular Counting}\label{sec:modular}
Our analysis of $\holantprob(\mathcal{S})$ applies verbatim to the case of counting modulo a fixed prime $p$, if we restrict ourselves to signatures $s$ with $s(0)\neq 0\mod p$. For the formal statement, we define, for each prime $p$, the problem $\holantprob_p(\mathcal{S})$ to be the version of $\holantprob(\mathcal{S})$ where we output the value of the holant modulo $p$. Likewise, we define the types $\mathbb{T}_p[\mathsf{lin}]$, $\mathbb{T}_p[\omega]$, and $\mathbb{T}_p[\infty]$ by evaluating the fingerprints modulo $p$.

For our lower bounds, we require the parameterised complexity class $\mathsf{Mod}_p\text{-}\W[1]$, which consists of all parameterised counting problems reducible to the problem of counting $k$-cliques modulo $p$, parameterised by $k$ (see~\cite{CurticapeanDH21}). Moreover, we will rely on the \emph{randomised} Exponential Time Hypothesis rETH, which is identical to ETH except for additionally ruling out sub-exponential time \emph{randomised}, bounded-error algorithms for $3\textsc{-SAT}$. We note that some authors already state ETH in a way to account for randomised algorithms~\cite{CurticapeanDH21}; however, to avoid confusion, we emphasise the need for rETH in our result on modular counting. 

\begin{theorem}\label{thm:modular_classification}
     Let $p$ be a prime, and let $\mathcal{S}$ be a finite set of signatures with $s(0)\not\equiv 0 \mod p$ for each $s\in \mathcal{S}$.
    \begin{itemize}
        \item[(I)] If $\mathcal{S}$ is of type $\mathbb{T}_p[\mathsf{Lin}]$, then $\holantprob_p(\mathcal{S})$ can be solved in FPT-near-linear time.
        \item[(II)] If $\mathcal{S}$ is of type $\mathbb{T}_p[\omega]$, then $\holantprob_p(\mathcal{S})$ can be solved in FPT-matrix-multiplication time. Moreover, $\holantprob_p(\mathcal{S})$ cannot be solved in FPT-near-linear time, unless the Triangle Conjecture fails.
        \item[(III)] Otherwise, that is, if $\mathcal{S}$ is of type $\mathbb{T}_p[\infty]$, $\holantprob_p(\mathcal{S})$ is $\mathsf{Mod}_p\text{-}\W[1]$-hard. Moreover, $\holantprob(\mathcal{S})$ cannot be solved in time $f(k)\cdot |V(\Omega)|^{o(k/\log k)}$ for any function $f$, unless rETH fails.  \qed
    \end{itemize}
\end{theorem}
\begin{proof}
    As mentioned previously, the classification is proved almost verbatim as in the case of exact counting. There are only two places in the proof which require slight modifications or further explanation; we discuss them subsequently.
    \begin{itemize}
        \item Lemma~\ref{lem:cphom_lower_bounds} on lower bounds on $\#\textsc{cp-Hom}(\mathcal{H})$ remains true if colour-prescribed homomorphisms are counted modulo $p$: If $\mathcal{H}$ contains a triangle, we still reduce from finding a triangle in a graph $G$ by constructing the tensor $G\otimes K_3$. However, the number of triangles in $G\otimes K_3$ is divisible by $3$ and might be divisible by other primes too. For this reason, we use a version of the Schwartz-Zippel-Lemma due to Williams et al.\ \cite{WilliamsWWY15} to complete the reductions; the details are identical to the reduction in Lemma~2.3 in the full version~\cite{PeyerimhoffRSSVarxiv} of~\cite{PeyerimhoffRSV21}. 

        If $\mathcal{H}$ has unbounded treewidth, the claim follows immediately from Lemma~2.3 in the full version~\cite{PeyerimhoffRSSVarxiv} of~\cite{PeyerimhoffRSV21}. 
        \item The reduction $\#\cphomsprob(\homsupp_\top(\mathcal{S})) \fptlinred \holantprobstar(\mathcal{S})$ in Lemma~\ref{lem:holant_sandwich} relies on the Complexity Monotonicity framework for counting homomorphisms from fractured graphs in~\cite[Section 4.1]{PeyerimhoffRSSVW23}. However, as observed in the proof of Lemma~2.4 in the full version~\cite{PeyerimhoffRSSVarxiv} of~\cite{PeyerimhoffRSV21}, the framework applies verbatim to counting modulo a fixed prime.
    \end{itemize}
\end{proof}

We provide the following example application of our trichotomy for modular counting; recall that the problem $\oplus\textsc{ColMatch}$ gets as input a positive integer $k$ and a $k$-edge-coloured graph $G$, and the output is the parity of the number of edge-colourful $k$-matchings in $G$; the parameter is $k$.
\begin{corollary}
    $\oplus\textsc{ColMatch}$ can be solved in FPT-matrix-multiplication time. Moreover, it cannot be solved in FPT-near-linear time, unless the Triangle Conjecture fails.
\end{corollary}
\begin{proof}
    Let $\mathsf{hw}_{\leq 1}$ be the signature defined by $\mathsf{hw}_{\leq 1}(x)=1$ for $x\leq 1$ and $\mathsf{hw}_{\leq 1}(x)=0$ otherwise. Clearly, the problem of counting edge-colourful $k$-matchings modulo $2$ is identical to $\holantprob_p(\{\mathsf{hw}_{\leq 1}\})$. 
    
    Now note that we have $(|\sigma|-1)!=0$ modulo $2$ whenever $|\sigma| \geq 3$, and, for $s=\mathsf{hw}_{\leq 1}$, we have $s(|B|)=0$ whenever $|B|\geq 2$. Therefore, for any $d\geq 3$, we have that $\chi(d,s)=0$ modulo $2$.

    However, note also that $\chi(2,d)=1$ modulo $2$: The set $[2]$ only has two partitions $\bot_2=\{\{1\},\{2\}\}$ and $\top_2=\{\{1,2\}\}$, and observe that $\top_2$ contains a block $B$ of size $2$, hence $s(|B|)$ and the contribution of $\top_2$ to $\chi(2,s)$ vanishes. Therefore 
    \[ \chi(2,s) = (-1)^{|\bot_2|-1}(|\bot_2|-1)! \prod_{B\in \bot_2} \frac{s(|B|)}{s(0)} = 1 \mod 2 \,. \]
    Note that dividing by $s(0)$ means multiplying with the inverse of $s(0)$ w.r.t.\ arithmetic modulo $p$, which must exist since $s(0)\neq 0\mod p$ and $p$ is a prime.
    Thus $\mathcal{S}=\{\mathsf{hw}_{\leq 1}\}$ is indeed of type $\mathbb{T}_2[\mathsf{\omega}]$. The claim hence follows from Theorem~\ref{thm:modular_classification} which concludes our proof.
\end{proof}

\begin{proof}[Proof of Theorem~\ref{thm:colmatch_mod_2_intro}]
    Follows immediately from the previous result and the fact that \[\oplus\textsc{ColMatch}\fptlinred\oplus\textsc{Match}\] via inclusion-exclusion (see Lemma~\ref{lem:col_to_uncol}).
\end{proof}

Finally, as mentioned as example in the abstract, it is easy to see that Theorem~\ref{thm:modular_classification} also implies hardness of counting edge-colourful $k$-matchings modulo $p$, for any prime $p>2$, since the type changes to $\mathbb{T}_p[\infty]$ in that case. We omit stating this example as a theorem since it has already been shown to be hard in previous work, using different methods, by Curticapean, Dell, and Husfeldt~\cite{CurticapeanDH21}.

\section{Extension to Signatures Allowing $s(0) = 0$}\label{sec:sig0}

So far, we have only considered signatures $s$ restricted to $s(0) \neq 0$. In this section, we lift this restriction and establish, similar to \Cref{main_thm}, a trichotomy for $\holantprob(\mathcal{S})$ for any finite set $\mathcal{S}$ of signatures $s$ without requiring $s(0) \neq 0$ (cf. \Cref{def:signatures_intro}).

\subsection{List Homomorphisms}

For our proofs, we need to extend the notion of coloured homomorphisms to list homomorphisms given below. We then show that all algorithmic results for counting coloured homomorphisms, mentioned in \Cref{sec:algo_count_col_homs}, also apply to counting list homomorphisms.

\begin{definition}\label{def:listHoms}
Let $H, G$ be two graphs and let $\mathcal{L} = (L_v)_{v\in V(H)}$ be a collection of sets $L_v \subseteq V(G)$. We set
\[
\homs{H}{G}[\mathcal{L}] = \{\phi \in \homs{H}{G} \mid \forall v \in V(H) : \phi(v) \in L_v\}.
\]
\end{definition}

\begin{definition}
Let $\mathcal{H}$ be a class of graphs. The problem $\#\listhomsprob(\mathcal{H})$ expects as input a graph $H \in \mathcal{H}$, a graph $G$, and a collection $\mathcal{L} = (L_v)_{v\in V(H)}$ of sets $L_v \subseteq V(G)$, and outputs $\#\homs{H}{G}[\mathcal{L}]$.
\end{definition}

As in the case of coloured homomorphisms (see \Cref{fact:colhoms_lintime}), we can interpret the lists $L_v$ as unary predicates and show that computing $\#\homs{H}{G}[\mathcal{L}]$ reduces to counting answers to acyclic conjunctive queries without quantified variables, which can be done in FPT-near-linear time (see e.g.\ \cite[Theorem 7]{BeraGLSS22}).

\begin{lemma}\label{lem:listHomsLinear}
Let $\mathcal{H}$ be a class of acyclic graphs. Then, $\#\listhomsprob(\mathcal{H})$ can be solved in time $f(|H|)\cdot\tilde{\mathcal{O}}(|V(G)| + |E(G)|)$, for some computable function $f$.   \qed 
\end{lemma}

Next, we adapt a result due to Curticapean, Dell and Marx from uncoloured homomorphisms to list homomorphisms to obtain, for $\mathcal{H}$ containing graphs of treewidth at most $2$, an algorithm for $\#\listhomsprob(\mathcal{H})$ via fast matrix multiplication; the proof of the following lemma can be found in \Cref{lem:AppendixListHomsMatrix}.

\begin{lemma}\label{lem:listHomsMatrix}
Let $\mathcal{H}$ be a class of graphs of treewidth at most 2.
Then $\#\listhomsprob(\mathcal{H})$ can be solved in time $f(|H|)\cdot \mathcal{O}(|V(G)|^{\omega})$ for some computable function $f$. Here, $\omega$ is the matrix multiplication exponent. \qed
\end{lemma}

\subsection{The Tractable Cases}
The criterion for the classification of $\holantprob(\mathcal{S})$ is the type of the set of signatures $s \in \mathcal{S}$, for which $s(0) \neq 0$ holds. We first deal with the tractable cases of the classification. We show that these cases reduce to a tractable restriction of the holant problem considered in \Cref{main_thm}, that considers signature grids with a restricted number of vertices with signatures of type $\mathbb{T}[\infty]$ (introducing a parameter that upper-bounds the latter number). To this end, we first prove some algorithmic results, we will need next.

% \begin{definition}
% Let $G$ be a graph and $G'$ be a subgraph of $G$. For any vertex $v \in V(G)$, we set $\mathcal{N}_{G'}(v) = \{u \in V(G') \mid \{u, v\} \in E(G)\}$.
% \end{definition}

\begin{definition}
Let $G_1, G_2$ be two graphs with respective vertex-colorings $\nu_{G_1} : V(G_1) \to [\ell]$ and $\nu_{G_2} : V(G_2) \to [\ell]$, for some $\ell \in \mathbb{N}$. For a subset $X \subseteq V(G_1)$ and $\phi \in \homs{(G_1[X], \nu_{G_1}|_{X})}{(G_2, \nu_{G_2})}$, we define,
\begin{equation*}
\parthomsX{(G_1, \nu_{G_1})}{(G_2, \nu_{G_2})}{X} = \{h \in \homs{(G_1, \nu_{G_1})}{(G_2, \nu_{G_2})} : h|_{X} = \phi\}\,.    
\end{equation*}
\end{definition}

\begin{lemma}\label{lem:partialHomsFPT}
Let $H, G$ be two graphs with respective vertex colorings $\nu_H: V(H) \to [\ell]$ and $\nu_G : V(G) \to [\ell]$, for some $\ell \in \mathbb{N}$. Further, let $V(H) = X_H \,\dot\cup\, Y_H$ and $V(G) = X_G \,\dot\cup\, Y_G$ such that $\nu_H(X_H) \subseteq \nu_G(X_G)$, $\nu_H(Y_H) \subseteq \nu_G(Y_G)$, and $\nu_G(X_G)\cap\nu_G(Y_G) = \emptyset$, and let $\phi \in \homs{(H[X_H], \nu_H|_{X_H})}{(G[X_G], \nu_G|_{X_G})}$.
\begin{enumerate}
    \item If every vertex $v \in Y_H$ has degree 1 then we can compute $\#\parthomsX{(H, \nu_H)}{(G, \nu_G)}{X_H}$ in $f(|H|)\cdot\tilde{\mathcal{O}}(|V(G)| + |E(G)|)$ time, for some computable function $f$.
    \item If every vertex $v \in Y_H$ has maximum degree at most 2, then we can compute $\#\parthomsX{(H, \nu_H)}{(G, \nu_G)}{X_H}$ in $g(|H|)\cdot\mathcal{O}(|V(G)|^{\omega})$ time, for some computable function $g$.
\end{enumerate}
\end{lemma}
\begin{proof}
For any vertex $v \in Y_H$, recall that $\mathcal{N}_{H[X_H]}(v)$ denotes the set of neighbors of $v$ in $X_H$. Let $U_{\phi}(v)$ denote the set of all vertices $u \in V(G)$, that $v$ may be mapped to, by some homomorphism $h \in \homs{(H,\nu_H)}{(G,\nu_G)}$ that respects $\phi$. The set $U_{\phi}(v)$ can be computed as follows. \begin{equation*}
U_{\phi}(v) = \bigcap_{w \in \mathcal{N}_{H[X_H]}(v)}\{u \in \mathcal{N}_{G}(\phi(w)) : \nu_G(u) = \nu_H(v)\}
\end{equation*}

It is easy to verify that the problem of computing $\#\parthomsX{(H, \nu_H)}{(G, \nu_G)}{X_H}$ is equivalent to computing the number of homomorphisms in $\homs{H[Y_H]}{G[Y_G]}$ under the restriction that every vertex $v \in Y_H$ may only be mapped to vertices in $U_{\phi}(v)$. Equivalently,
\[\#\parthomsX{(H, \nu_H)}{(G, \nu_G)}{X_H} = \#\{h \in \homs{H[Y_H]}{G[Y_G]} \mid \forall v\in Y_H : h(v) \in U_{\phi}(v)\}.\]

Let $\mathcal{U}$ denote the collection of the sets $U_{\phi}(v)$, for all vertices $v \in Y_H$. By the equation above and \Cref{def:listHoms}, follows that 
\[\#\parthomsX{(H,\nu_H)}{(G,\nu_G)}{X_H} = \#\homs{H[Y_H]}{G[Y_G]}[\mathcal{U}].\]

It remains to argue that the time we need to compute $\mathcal{U}$ is bounded by the desired running times, since then, Cases 1 and 2 follow from \Cref{lem:listHomsLinear,lem:listHomsMatrix} respectively. To this end, assume that for some vertex $v\in Y_H$, the set $\mathcal{N}_{H[X_H]}(v)$ contains two vertices $a, b \in V(G)$. We wish to compute the intersection of the sets $A = \{u \in \mathcal{N}_{G}(\phi(a)) : \nu_G(u) = \nu_H(v)\}$ and $B = \{u \in \mathcal{N}_{G}(\phi(b)) : \nu_G(u) = \nu_H(v)\}$. This can be done in time $O(|V(G)|)$, as follows. We can implement the characteristic function $\boldsymbol{1}_B : V(G) \rightarrow \{0, 1\}$ of the set $B$, in linear time (e.g, via a $|V(G)|$-size array). Then, for each element $a' \in A$ we check in constant time whether $\boldsymbol{1}_B(a') = 1$.
\end{proof}

As already argued, restricting the holant problem so that the number of vertices with signatures of type $\mathbb{T}[\infty]$ is upper-bounded by some parameter, renders the problem tractable. All of the above are formally stated in the following theorem.

\begin{lemma}\label{lem:restrictedHolant}
Let $\mathcal{S}$ be a finite set of signatures such that for all $s \in \mathcal{S}$, $s(0) \neq 0$. We assume that $\mathcal{S}$ is the disjoint union of a set $\mathcal{S}^e$ which is not of type $\mathbb{T}[\infty]$ and a set $\mathcal{S}^h$ that can be of any type. Let $\holantprob(\mathcal{S}^e;\mathcal{S}^h, r)$ be the restriction of $\holantprob(\mathcal{S})$ on instances where the number of vertices with signatures in $\mathcal{S}^h$ is at most $r$. 
\begin{enumerate} 
\item If $\mathcal{S}^e$ is of type $\mathbb{T}[\mathsf{Lin}]$ then $\holantprob(\mathcal{S}^e;\mathcal{S}^h, r)$ can be solved in FPT-near-linear time, that is, there exists a computable function $f$ such that $\holantprob(\mathcal{S}^e;\mathcal{S}^h,r)$ can be solved in $f(k,r)\cdot\tilde{\mathcal{O}}(|V(\Omega)| + |E(\Omega)|)$ time. 
\item If $\mathcal{S}^e$ is of type $\mathbb{T}[\omega]$, then $\holantprob(\mathcal{S}^e; \mathcal{S}^h, r)$ can be solved in FPT-matrix-multiplication time, that is, there exists a computable function $f$ such that $\holantprob(\mathcal{S}^e;\mathcal{S}^h,r)$ can be solved in $f(k, r)\cdot\mathcal{O}(|V(\Omega)|^{\omega})$ time.
\end{enumerate}
\end{lemma}
\begin{proof} 
Let $\holantprobstar(\mathcal{S}^e;\mathcal{S}^h, r)$ denote the restriction of $\holantprobstar(\mathcal{S})$ on instances $(H, (G,\xi,\{s_v\}_{v\in V(G)}))$ where the number of vertices $v \in V(G)$ with signature $s_v \in \mathcal{S}^h$ is at most $r$. Set $V^h = \{v \in V(G) : s_v \in \mathcal{S}^h\}$. Recall that $G$ is an $H$-colored graph with respect to the $H$-coloring $\xi \in \homs{G}{H}$ (see \Cref{sec:col_graphs_fractures}). Assume that $V(H)=\{v_1,\dots,v_z\}$. For each $i\in [z]$, let $n_i$ denote the number of vertices of $G$ colored by $\xi$ with $v_i$, and let $s_i$ be the signature of the vertices colored by $\xi$ with $v_i$. From \Cref{cor:collect_coeffs}, the following identity holds.
\begin{equation}
\holant(G,\xi,\{s_v\}_{v\in V(G)}) \cdot \prod_{i=1}^z s_i(0)^{-n_i} = \sum_{\vec{\rho}\in \mathcal{F}(H)} \mathsf{coeff}(\vec{\rho})  \cdot \#\homscp(\fracture{H}{\vec{\rho}} \to (G,\xi))\,,   
\end{equation}
where $\mathsf{coeff}(\vec{\rho})$ is given by \Cref{lem:coeffs_done} as $\prod_{i = 1}^{z}\prod_{B \in \vec{\rho}(v_i)}\chi(|B|, s_i)$.

We show that whenever $\mathsf{coeff}(\vec{\rho}) \neq 0$, we can compute $\#\homscp(\fracture{H}{\vec{\rho}} \to (G, \xi))$  in FPT time with respect to the parameters $k, r$,  distinguishing between the following two cases regarding the types of signatures in $\mathcal{S}^e$.
\begin{enumerate}
    \item \textbf{$\mathcal{S}^e$ is of type} $\mathbb{T}[\mathsf{Lin}]$. We may assume that $\vec{\rho}(v_i) = \bot$, for all $v_i \in V(H)$ with $s_i \in \mathcal{S}^e$, since otherwise, (that is, if for a vertex $v_i \in V(H)$ and for some $B \in \vec{\rho}(v_i)$ we had $|B| > 1$) we would have $\mathsf{coeff}(\vec{\rho}) = 0$. This follows from the assumption that the signature $s_i$ is of type $\mathbb{T}[\mathsf{Lin}]$, which implies that $\chi(d, s_i) = 0$, for all $d > 1$. Let $V^e_{\vec{\rho}}$ denote the set of vertices $v \in V(\fracture{H}{\vec{\rho}})$ such that the vertices of $G$ coloured by $\xi$ with $v$ have signature in $\mathcal{S}^e$ and let $V^h_{\vec{\rho}}$ denote the respective vertex set for signatures in $\mathcal{S}^h$. Since $\vec{\rho}(v_i) = \bot$, for all $v_i \in V(H)$ with $s_i \in \mathcal{S}^e$, it follows that the degree of each vertex $v\in V^e_{\vec{\rho}}$ is 1. Hence, we can compute $\#\homscp(\fracture{H}{\vec{\rho}} \to (G,\xi))$ in FPT-near-linear time as follows. We enumerate all $\phi \in \homscp(\fracture{H}{\vec{\rho}}\,[V^h_{\vec{\rho}}] \to (G[V^h],\xi|_{V^h}))$ via brute-force in $|V^h|^{\mathcal{O}(|V_{\vec{\rho}}^h|)}$ time and then, for each $\phi$, compute the number of its extensions to $\#\homs{\fracture{H}{\vec{\rho}}}{(G,\xi)}$, according to \Cref{lem:partialHomsFPT}, Case 1, in $f_1(|\fracture{H}{\vec{\rho}}|)\cdot\tilde{\mathcal{O}}(|G|)$ time, for some computable function $f_1$.

    \item \textbf{$\mathcal{S}^e$ is of type} $\mathbb{T}[\omega]$. We may assume that for all vertices $v_i \in V(H)$ such that $s_i \in \mathcal{S}^e$ and all $B \in \vec{\rho}(v_i)$, it is $|B| \leq 2$, since otherwise, (that is, if for a vertex $v_i \in V(H)$ and for some $B \in \vec{\rho}(v_i)$, we had $|B| > 2$) we would have $\mathsf{coeff}(\vec{\rho}) = 0$. This follows from the assumption that the signature $s_i$ can also be of type $\mathbb{T}[\omega]$, which implies that $\chi(d, s_i) = 0$, for all $d > 2$. Let $V^e_{\vec{\rho}}, V^h_{\vec{\rho}}$ defined as in case 1. By the previous assumption on the fracture $\vec{\rho}$, it follows that, the degree of each vertex $v \in V^e_{\vec{\rho}}$ is at most 2. Hence, we can compute $\#\homscp(\fracture{H}{\vec{\rho}} \to (G,h))$ in FPT-matrix-multiplication time as follows. We enumerate all $\phi \in \homscp(\fracture{H}{\vec{\rho}}\,[V^h_{\vec{\rho}}] \to (G[V^h],\xi|_{V^h}))$ via brute-force in $|V^h|^{\mathcal{O}(|V_{\vec{\rho}}^h|)}$ time and then, for each $\phi$, compute the number of its extensions to $\#\homs{\fracture{H}{\vec{\rho}}}{(G,h)}$, according to \Cref{lem:partialHomsFPT}, Case 2, in $f_2(|\fracture{H}{\vec{\rho}}|)\cdot\mathcal{O}(|V(G)|^\omega)$ time, for some computable function $f_2$.
\end{enumerate}

From \Cref{lem:main_equivalence_hard}, we know that for any instance $\Omega \in \holantprob(\mathcal{S})$ we can compute $\holant(\Omega)$ via an FPT-(near)-linear reduction to $\holantprobstar(\mathcal{S})$. In particular, the oracle is queried with instances $\Omega^* \in \holantprobstar(\mathcal{S})$, where each $\Omega^* = \Omega\,\otimes\,\left(\bigotimes_{i=1}^{m}\,(H_i, \nu_{H_i}, \xi_{H_i})\right)$, for graphs $(H_i, \nu_{H_i}, \xi_{H_i}) \in \Gamma_{\text{\sf{inj}}}(\ell, k, \mathcal{S})$ and $m$ depending only on $k$. Recall that, for $\ell = |\mathcal{S}|$, $\Gamma_{\text{\sf{inj}}}(\ell, k, \mathcal{S})$ contains all isomorphism types of $(\ell, k)$-colored graphs $(H, \nu_H, \xi_H)$ with $\nu_H : V(H) \to \mathcal{S}$ and injective $\xi_H : E(H) \to [k]$. So, $\Omega^*$ contains at most $\max_{1\leq i\leq m}|V(H_i)|^m\cdot r \leq (2k)^m\cdot r$ vertices with signatures in $\mathcal{S}^h$, which implies that each $\Omega^* \in \holantprobstar(\mathcal{S}^e; \mathcal{S}^h, g(k)\cdot r)$, for some computable function $g$ that depends only on $k$. Recall that $|V(\Omega^*)|$ (resp. $|E(\Omega^*)|)$ is at most $g'(k)\cdot|V(\Omega)|$ (resp. $g'(k)\cdot|E(\Omega)|)$, for some computable function $g'$. Finally, we compute $\holant(\Omega) \in \holantprob(\mathcal{S}^e;\mathcal{S}^h, r)$ via the aforementioned reduction to $\holantprobstar(\mathcal{S}^e; \mathcal{S}^h, g(k)\cdot r)$ which in turn is computed according to the Cases (1) and (2) above, yielding the desired.
\end{proof}

Now, we are ready to study the tractable cases of the classification.

\begin{lemma}\label{lem:reductionRestrictedHolant}
Let $\mathcal{S}$ be a finite set of signatures. Let $\mathcal{S}_0 = \{s \in \mathcal{S} \mid s(0) = 0\}$. Let $\Omega \in \holantprob(\mathcal{S})$ and set $n_0$ to be the number of vertices in $V(\Omega)$ with signatures in $\mathcal{S}_0$.
\begin{enumerate}
\item If $\mathcal{S}\setminus\mathcal{S}_0$ is of type $\mathbb{T}[\mathsf{Lin}]$, then $\holant(\Omega)$ can be computed in FPT-near-linear time, with respect to the parameters $k, n_0$, that is, there is a computable function $f$ such that $\holant(\Omega)$ can be computed in $f(k,n_0)\cdot\tilde{\mathcal{O}}(|V(\Omega)| + |E(\Omega)|)$ time.
\item If $\mathcal{S}\setminus\mathcal{S}_0$ is of type $\mathbb{T}[\omega]$, then $\holant(\Omega)$ can be computed in FPT-matrix-multiplication time with respect to the parameters $k, n_0$, that is, there is a computable function $f$ such that $\holant(\Omega)$ can be computed in $f(k, n_0)\cdot\mathcal{O}(|V(\Omega)|^{\omega})$ time.
\end{enumerate}
\end{lemma}
\begin{proof}
Let $V_0$ denote the set of vertices $v \in V(\Omega)$ with signatures $s_v \in \mathcal{S}_0$. Note that, $|V_0| = n_0$. Let $\mathcal{A}(\Omega)$ denote the set of all colorful edge-subsets $A$ such that, for all $v \in V_0$, $A \cap E(v) \neq \emptyset$. We observe that the contribution of any $A \notin \mathcal{A}(\Omega)$ to $\holant(\Omega)$ is zero since for every $A \notin \mathcal{A}(\Omega)$ there is at least one vertex $v \in V_0$ such that $|A\cap E(v)| = 0$. 

For a signature $s$ and an algebraic complex number $\alpha$, let $\zsig{s}{\alpha}$ denote the following signature.
\[
\zsig{s}{\alpha}(d) = 
\begin{cases}
s(d), & d > 0  \\
\alpha, & d = 0 
\end{cases}
\]

For an algebraic complex number $\alpha$, we consider the signature grid $\Omega_{\alpha}$, which is obtained from $\Omega$ by replacing, for every vertex $v \in V_0$, its signature $s_v$ with the signature $\zsig{s_v}{\alpha}$ and retaining the same signatures for the rest of the vertices. Let $s_v'$ denote the signature of vertex $v$ in the new grid $\Omega_{\alpha}$. We have,
\begin{equation}\label{eq:holant-alpha}
\holant(\Omega_{\alpha}) = \sum_{A \in \mathcal{A}(\Omega)}\prod_{v \in V(\Omega_{\alpha})}s_v'(|A\cap E(v)|) + \sum_{A \notin \mathcal{A}(\Omega)}\prod_{v \in V(\Omega_{\alpha})}s_v'(|A\cap E(v)|)\,.
\end{equation} 
The left summand in \Cref{eq:holant-alpha}, that is,
the contribution of all colorful edge-subsets in $\mathcal{A}(\Omega)$ to $\holant(\Omega_{\alpha})$, is equal to $\holant(\Omega)$. This follows from the aforementioned fact that the contribution of $A\notin \mathcal{A}(\Omega)$ to $\holant(\Omega)$ is zero, and from the observation that for $A \in \mathcal{A}(\Omega)$ and a vertex $v \in V_0$, we have $A\,\cap\,E(v) \neq \emptyset$ and so $\zsig{s_v}{\alpha}(|A\,\cap\,E(v)|) = s_v(|A\,\cap\,E(v)|)$. Regarding the right summand, that is, the contribution of all $A \notin \mathcal{A}(\Omega)$ to $\holant(\Omega_{\alpha})$, we observe that we can express it as a polynomial in indeterminate $\alpha$ with zero constant coefficient; the latter holds since for every colorful edge-subset $A \notin \mathcal{A}(\Omega)$, there is at least one vertex $v \in V_0$, such that $A\,\cap\,E(v) = \emptyset$, and so $\zsig{s_v}{\alpha}(|A\cap E(v)|) = \zsig{s_v}{\alpha}(0) = \alpha$. Furthermore, the degree of the polynomial is at most the maximum number of vertices $v \in V(\Omega_{\alpha})$ such that $s'_v(0) = \alpha$. The latter is by definition equal to $n_0$. 

Hence, $\holant(\Omega_{\alpha})$ can be seen as a polynomial in $\alpha$ of degree $n_0$ with a constant coefficient equal to $\holant(\Omega)$. Using polynomial interpolation, assuming we are given the evaluation of $\holant(\Omega_{\alpha})$ for $n_0+1$ distinct values of $\alpha$, we can recover the coefficients of the polynomial, and hence compute $\holant(\Omega)$, with $\mathcal{O}(n_0^3)$ additional arithmetic operations.  

It remains to show how to compute $\holant(\Omega_{\alpha}$), for various $\alpha$. To this end, let $\alpha$ be an algebraic complex number and let $\mathcal{S}_{\alpha}$ denote the set of all signatures $\zsig{s}{\alpha}$ for all $s \in \mathcal{S}_0$. We can force $(\mathcal{S}\setminus\mathcal{S}_0)\,\cap\,\mathcal{S}_{\alpha} = \emptyset$, by further restricting $\alpha \neq s(0)$, for every $s \in \mathcal{S}\setminus \mathcal{S}_0$. The latter implies that the number of vertices in $V(\Omega_{\alpha})$ with signatures in $\mathcal{S}_{\alpha}$, is at most $n_0$ and hence  $\Omega_{\alpha} \in \holantprob(\mathcal{S}\setminus\mathcal{S}_0; \mathcal{S}_{\alpha}, n_0)$. So, $\holant(\Omega_{\alpha})$ can be computed according to \Cref{lem:restrictedHolant}, based on the type of $\mathcal{S}\setminus\mathcal{S}_0$, yielding the desired.
\end{proof}

\subsection{A Trichotomy for $\holantprob(\mathcal{S})$ for all signatures}

With the upper bounds obtained above and the lower bounds implied directly by \Cref{main_thm}, we can fully classify $\holantprob(\mathcal{S})$ for all finite signature sets.

\begin{theorem}\label{thm:holant_trichotomy_zero}
Let $\mathcal{S}$ be a finite set of signatures. Let $\mathcal{S}_0 = \{s \in \mathcal{S} \mid s(0) = 0\}$.
\begin{enumerate}
\item If $\mathcal{S}\setminus\mathcal{S}_0$ is of type $\mathbb{T}[\mathsf{Lin}]$, then $\holantprob(\mathcal{S})$ can be solved in FPT-near-linear time, that is, there is a computable function $f$ such that $\holantprob(\mathcal{S})$ can be solved in $f(k)\cdot\tilde{\mathcal{O}}(|V(\Omega)| + |E(\Omega)|)$ time.
\item If $\mathcal{S}\setminus\mathcal{S}_0$ is of type $\mathbb{T}[\omega]$, then $\holantprob(\mathcal{S})$ can be solved in FPT-matrix-multiplication time, that is, there is a computable function $f$ such that $\holantprob(\mathcal{S})$ can be solved in $f(k)\cdot\mathcal{O}(|V(\Omega)|^{\omega})$ time. Moreover, $\textsc{p-Holant}(\mathcal{S})$ cannot be solved in time $f(k)\cdot \tilde{\mathcal{O}}(|V(\Omega)|+|E(\Omega)|)$ for any function $f$, unless the Triangle Conjecture fails.
\item Otherwise, that is, if $\mathcal{S}\setminus\mathcal{S}_0$ is of type $\mathbb{T}[\infty]$, $\holantprob(\mathcal{S})$ is $\#\mathrm{W}[1]$-complete. Moreover, $\textsc{p-Holant}(\mathcal{S})$ cannot be solved in time $f(k)\cdot |V(\Omega)|^{o(k/\log k)}$ for any function $f$, unless the Exponential Time Hypothesis fails.
\end{enumerate}
\end{theorem}
\begin{proof}
Regarding the first two cases, let $\Omega \in \holantprob(\mathcal{S})$ and set $n_0$ to be the number of vertices in $\Omega$ with signatures in $\mathcal{S}_0$. Note that since any subset of $k$ edges can only cover up to $2k$ vertices, it follows that if $k < n_0/2$ then there is always at least one uncovered vertex in $V_0$, in which case, $\holant(\Omega)$ trivially evaluates to 0. So, we may assume that $n_0 \leq 2k$. 

Following \Cref{lem:reductionRestrictedHolant}, if $\mathcal{S}\setminus\mathcal{S}_0$ is of type $\mathbb{T}[\mathsf{Lin}]$ then, $\holant(\Omega)$ can be computed in $f(k, n_0) \cdot \tilde{\mathcal{O}}(|V(\Omega)|+|E(\Omega)|)$ time for some computable function $f$. Equivalently, if $\mathcal{S}\setminus\mathcal{S}_0$ is of type $\mathbb{T}[\omega]$, then $\holant(\Omega)$ can be computed in $g(k,n_0)\cdot\mathcal{O}(|V(\Omega)|^{\omega})$ time for some computable function $g$. Since $n_0 \leq 2k$, the first two cases follow.

For the last case, if $\mathcal{S}\setminus\mathcal{S}_0$ is of type $\mathbb{T}[\infty]$, then from \Cref{main_thm} follows that $\holantprob(\mathcal{S}\setminus\mathcal{S}_0)$ is $\#\mathrm{W}[1]$-complete, and so is $\holantprob(\mathcal{S})$.
\end{proof}

We obtain, as immediate consequence, the classification of the edge-coloured graph factor problem.
\begin{corollary}\label{cor:factor_classification_col}
    If $\mathcal{B}$ contains a set $\{0\} \subsetneq S \subsetneq \mathbb{N}$ then $\textsc{ColFactor}(\mathcal{B})$ is $\#\W[1]$-complete, and cannot be solved in time $f(k)\cdot n^{o(k/\log k)}$ for any function~$f$, unless the Exponential Time Hypothesis fails. Otherwise $\textsc{ColFactor}(\mathcal{B})$ is solvable FPT-near-linear time.
\end{corollary}
\begin{proof}
    For the lower bounds, assume there is an $S\in \mathcal{B}$ with $\{0\} \subsetneq S \subsetneq \mathbb{N}$. We define a signature 
    \[s(x) := \begin{cases}
        1 & x \in S\\
        0 & x \notin S\,,
    \end{cases}\]
    and set $\mathcal{S}=\{s\}$. Then, clearly, $\textsc{p-Holant}(\mathcal{S})\fptlinred \textsc{ColFactor}(\mathcal{B})$. 
    
    We show that $\{s\}$ is of type $\mathbb{T}[\infty]$. To this end, observe that $\chi(3,s)= 2s(1)^3 - 3s(1)s(2) + s(3)$, which is non-zero unless $(s(1),s(2),s(3))\in \{(0,0,0),(0,1,0),(1,1,1)\}$. We will consider these three cases separately:
\begin{itemize}[leftmargin=1.5cm]
    \item[$(0,1,0)$:] Consider $\chi(4,s)$. Since $s(1)=s(3)=0$, the only partitions contributing to $\chi(4,s)$ are $\{\{1,2\},\{3,4\}\}$, $\{\{1,3\},\{2,4\}\}$, $\{\{1,4\},\{2,3\}\}$, and $\{\{1,2,3,4\}\}$. The former three each contribute $(-1)^{2-1} (2-1)! = -1$, and the latter one contributes $s(4)$. Thus $\chi(4,s) = -3 +s(4)\neq 0$.
    \item[$(0,0,0)$:] Set $c=\min\{x>0 \mid s(x)= 1\}$. Note that $c$ must exist since $\{0\} \subsetneq S$, and that $c\geq 4$. Then, clearly, $\chi(c,s)= s(c)\neq 0$.
    \item[$(1,1,1)$:] Set $c=\min\{x>0 \mid s(x)= 0\}$. Note that $c$ must exist since $S\neq \mathbb{N}$, and that $c\geq 4$. Then
    \[\chi(c,s)= \sum_{\sigma < \top_c} (-1)^{|\sigma|-1} (|\sigma|-1)! = \sum_{\sigma} (-1)^{|\sigma|-1} (|\sigma|-1)! - \left((-1)^{|\top_c|-1} (|\top_c|-1)!\right) = 0 -1 = -1 \,,\]
    where the third equation holds again by the properties of the M\"obius function of the partition lattice (see e.g.\ \cite[Section 3.7]{Stanley11}).
\end{itemize}
This shows that $\{s\}$ is of type $\mathbb{T}[\infty]$. Given that $s(0)=1\neq 0$, the lower bounds thus follow immediately from \Cref{thm:holant_trichotomy_zero}.

Now, for the upper bound, assume that $\mathcal{B}=\{S_1,\dots,S_\ell\}$ for some $\ell>0$, such that none of the $S_i$ satisfies $\{0\}\subsetneq S_i \subsetneq \mathbb{N}$. Equivalently, this means that for each $i\in[\ell]$, either $0\notin S_i$, $S_i=\{0\}$, or $S_i=\mathbb{N}$. For each $i\in [\ell]$, define a signature $s_i$ by setting $s_i(x)=1$ if $x\in S_i$ and $s_i(x)=0$ otherwise. Moreover, let $\mathcal{S}=\{s_1,\dots,s_\ell\}$ and note that $\textsc{ColFactor}(\mathcal{B}) \fptlinred \textsc{p-Holant}(\mathcal{S})$. 
If $S_i=\{0\}$ or $S_i=\mathbb{N}$ then, clearly, $s_i(0)=1$, and $s_i$ is constant $0$ or $1$ on $\mathbb{N}_{>0}$. It is then easy to prove (see Lemma~\ref{lem:linear_type_sigs_equivalence} for a more general version) that the subset of $\mathcal{S}$ containing only signatures with $s(0)\neq 0$ is of type $\mathbb{T}[\mathsf{lin}]$. The claim hence follows from \Cref{thm:holant_trichotomy_zero}.
\end{proof}

\section{Parameterised Uncoloured Holants}\label{sec:uncoloured}
We begin by recalling the definition of parameterised uncoloured holant problems; to avoid notational clutter, given a finite set $E$, we define $\binom{E}{k}:=\{A \subseteq E \mid |A| = k\}$.
\begin{definition}
    Let $\mathcal{S}$ be a finite set of signatures.
    An (uncoloured) \emph{signature grid} over $\mathcal{S}$ is a pair of a graph~$G$ and a collection of signatures $\{s_v\}_{v\in V(G)}$ from $\mathcal{S}$. Given a signature grid $\Omega=(G,\{s_v\}_{v\in V(G)})$, and a positive integer $k$, we set
    \[ \mathsf{UnColHolant}(\Omega,k) = \sum_{A \in \binom{E(G)}{k}} \prod_{v\in V(G)} s_v(|A \cap E(v)|) \,.\]
    The problem $\text{\sc{p-UnColHolant}}(\mathcal{S})$ expects as input a positive integer $k$ and a signature grid $\Omega=(G,\{s_v\}_{v\in V(G)})$ over $\mathcal{S}$. The output is $\mathsf{UnColHolant}(\Omega,k)$, and the problem is parameterised by $k$. 
\end{definition}

The goal of this section is to prove our main classification theorem for $\text{\sc{p-UnColHolant}}(\mathcal{S})$:

\begin{theorem}[Theorem~\ref{thm:main_uncol}, restated]\label{thm:main_uncol_restate}
    Let $\mathcal{S}$ be a finite set of signatures. 
    \begin{itemize}
        \item[(I)] If $\mathcal{S}$ is of type $\mathbb{T}[\mathsf{Lin}]$, then $\text{\sc{p-UnColHolant}}(\mathcal{S})$ can be solved in FPT-near-linear time.
        \item[(II)] Otherwise $\text{\sc{p-UnColHolant}}(\mathcal{S})$ is $\#\W[1]$-complete. If, additionally, $\mathcal{S}$ is of type $\mathbb{T}[\infty]$, then $\text{\sc{p-UnColHolant}}(\mathcal{S})$ cannot be solved in time $f(k)\cdot |V(\Omega)|^{o(k/\log k)}$, unless ETH fails. \qed
    \end{itemize}
\end{theorem}

Similarly as in the previous sections, we will translate the uncoloured Holants into a linear combination of homomorphism counts. To this end, given a finite set of signatures $\mathcal{S}$, we denote by $\mathcal{G}(\mathcal S)$ the set of all (isomorphism types of) $\mathcal{S}$-vertex-coloured graphs $(H,\nu)$, where $H$ does not contain isolated vertices, and $\nu:V(H)\to \mathcal{S}$ assigns each vertex of $H$ to a signature in $\mathcal{S}$. Two $\mathcal{S}$-vertex-coloured graphs $(H_1,\nu_1)$ and $(H_2,\nu_2)$ are isomorphic, denoted by $(H_1,\nu_1)\cong(H_2,\nu_2)$ if there is an isomorphism $\iota$ from $H_1$ to $H_2$ that preserves colours, that is, $\nu_2(\iota(v))=\nu_1(v)$ for all $v\in V(H_1)$. Moreover, an automorphism of $(H,\nu)$ is an isomorphism from $(H,\nu)$ to itself, and we denote by $\mathsf{Aut}(H,\nu)$ the set of all automorphisms of $(H,\nu)$. Similarly, a homomorphism from $(H,\nu)$ to $\Omega=(G,\{s_v\}_{v\in V(G)})$ is a homomorphism $\varphi$ from $H$ to $G$ such that $s_{\varphi(v)} = \nu(v)$ for all $v\in V(H)$. We write $\homs{(H,\nu)}{\Omega}$ for the set of all homomorphisms from $\varphi$ to $\Omega$. Embeddings from $(H,\nu)$ to $\Omega=(G,\{s_v\}_{v\in V(G)})$ are defined, likewise, as embeddings from $H$ to $G$ that preserve the signatures of the vertices, and we denote by $\embs{(H,\nu)}{\Omega}$ the set of all embeddings from $(H,\nu)$ to $\Omega$.

We begin by translating embeddings to homomorphisms. To this end, given $(H,\nu)$, a partition $\rho$ of $V(H)$ is called \emph{colour-consistent} if vertices in the same block must have the same colour w.r.t.\ $\nu$, and we denote the set of all colour-consistent partitions of $V(H)$ by $\mathsf{colPart}(H)$. For $\rho \in \mathsf{colPart}(H)$ we set $(H,\nu)/\rho = (H/\rho,\nu/\rho)$, where $(H/\rho)$ is the usual quotient graph, and $\nu/\rho$ assigns a vertex $v_B$ of $H/\rho$ to the signature of the vertices in the block $B$, which is well-defined since $\rho$ is colour-consistent.

We start by translating embeddings to homomorphisms; the proof is standard and deferred to \Cref{sec:app-uncol}.
\begin{lemma}\label{lem:colEmb_to_colHom_S-coloured}
    Let $\mathcal{S}$ be a finite set of signatures, let $(H,\nu)$ be an $\mathcal{S}$-vertex-coloured graph, and let $\Omega=(G,\{s_v\}_{v\in V(G)})$ be a signature grid over $\mathcal{S}$, we have
    \[ \#\embs{(H,\nu)}{\Omega} = \sum_{\rho \in \mathsf{colPart}(H)} \mu(\bot,\rho)\cdot \#\homs{(H,\nu)/\rho}{\Omega}\,,\]
    where $\mu(\bot,\rho)=\prod_{B\in \rho}(-1)^{|B|-1}(|B|-1)!$ is the (usual) M\"obius function of partitions.\qed
\end{lemma}

Now let $\mathcal{G}_k(\mathcal{S})$ denote the set of all (isomorphism types of) $\mathcal{S}$-vertex-coloured graphs $(H,\nu)$, where $H$ has exactly $k$ edges and does not contain isolated vertices, and $\nu:V(H)\to \mathcal{S}$ assigns each vertex of $H$ to a signature in $\mathcal{S}$.
\begin{lemma}\label{lem:uncoloured_hombasis_general}
    Let $\mathcal{S}=\{s_1,\dots,s_\ell\}$ be a finite set of signatures, and let
    $\Omega=(G,\{s_v\}_{v\in V(G)})$ be a signature grid over $\mathcal{S}$. Set $n_i=|\{v\in V(G)\mid s_v=s_i\}|$ for each $i\in [\ell]$, and let $k$ be a positive integer. We have
    \[ \mathsf{UnColHolant}(\Omega,k) = \prod_{i\in \ell}s_i(0)^{n_i} \cdot \sum_{(H,\nu)\in \mathcal{G}(\mathcal{S})} \zeta_{\mathcal{S},k}(H,\nu) \cdot \#\homs{(H,\nu)}{\Omega}\,, \]
    where 
    \[\zeta_{\mathcal{S},k}(H,\nu)= \sum_{(F,\nu_F)\in\mathcal{G}_k(\mathcal{S})} \frac{\prod_{v\in V(F)}{s_v(d_F(v))}/{s_v(0)}}{\#\mathsf{Aut}(F,\nu_F)} \cdot \sum_{\substack{\rho \in \mathsf{colPart}(F,\nu_F)\\(F,\nu_F)/\rho \cong (H,\nu)}}\mu(\bot,\rho)  \,.\]
    % Here $d_F(v)$ denotes the degree of $v$ in $F$.
\end{lemma}
\begin{proof}
    Let us write $\mathsf{colSub}((F,\nu_F)\to \Omega)$ for the set of all subgraphs of $G$ that are isomorphic to $(F,\nu_F)$ as an $\mathcal{S}$-vertex-coloured graph. Then we can partition the $k$-edge subsets of $E(G)$ by the ($S$-vertex-colored) subgraphs they induce and obtain:
    \begin{align*}
        \mathsf{UnColHolant}(\Omega,k) &= \sum_{(F,\nu_F) \in \mathcal{G}_k(\mathcal{S}) }\#\mathsf{colSub}((F,\nu_F) \to \Omega) \cdot \prod_{s_i:i\in [\ell]}s_i(0)^{n_i-n^F_i} \cdot \prod_{v\in V(F)}s_v(d_F(v)) \\
    ~&= \prod_{i \in [\ell]}s_i(0)^{n_i} \cdot \sum_{(F,\nu_F) \in \mathcal{G}_k(\mathcal{S}) }\#\mathsf{colSub}((F,\nu_F) \to \Omega) \cdot \prod_{v\in V(F)}\frac{s_v(d_F(v))}{s_v(0)}\,,
    \end{align*}
    where $n_i$ and $n^F_i$ denote the numbers of vertices of $G$ and $F$, respectively, with signature $s_i$.

    Next observe that, similarly as for uncoloured graphs, $\mathsf{Aut}(F,\nu_F)$ acts on $\mathsf{colEmb}((F,\nu_F) \to \Omega)$ and the orbits of this action correspond precisely to the elements in $\mathsf{colSub}((F,\nu_F) \to \Omega)$. Thus
    \[\#\mathsf{colSub}((F,\nu_F) \to \Omega) = \#\mathsf{Aut}(F,\nu_F)^{-1}\cdot \#\mathsf{colEmb}((F,\nu_F) \to \Omega)  \,.\]
    The lemma then follows by invoking Lemma~\ref{lem:colEmb_to_colHom_S-coloured} and collecting for isomorphic terms.
\end{proof}

Next, our goal is to establish complexity monotonicity for counting homomorphisms from $\mathcal{S}$-vertex-coloured graphs into signature grids over $\mathcal{S}$. The formal statement sufficient for our purposes is provided below and the proof is deferred to \Cref{sec:app-uncol}, since it is an easy consequence of previous works on linear combinations of homomorphism counts.

\begin{lemma}\label{lem:monotonicity_S-coloured}
    Let $\mathcal{S}$ be a finite set of signatures and let $\mathcal{C}$ be the class of all graphs $H$ for which there is a positive integer $k$ and a colouring $\nu:V(H)\to \mathcal{S}$ such that $\zeta_{\mathcal{S},k}(H,\nu)\neq 0$.
    \begin{itemize}
        \item[(1)] If all graphs in $\mathcal{C}$ are acyclic, then  $\text{\sc{p-UnColHolant}}(\mathcal{S})$ can be solved in FPT-near-linear time.
        \item[(2)] If $\mathcal{C}$ has unbounded treewidth, then $\text{\sc{p-UnColHolant}}(\mathcal{S})$ is $\#\W[1]$-complete.\qed
    \end{itemize}
\end{lemma}

The heart, and in fact the most challenging part, of our investigation of the uncoloured holant problem boils down to understanding the coefficient function $\zeta_{\mathcal{S},k}$. We provide the detailed analysis encapsulated in Section~\ref{sec:analysis_of_zeta}. In what follows, we present our key result on $\zeta_{\mathcal{S},k}$ and invoke it in the proof of Theorem~\ref{thm:main_uncol_restate}.

\begin{remark}[On $s(0)=1$]\label{rem:s_zero_equals_1}
Let $\mathcal{S}$ be a finite set of signatures $s$ such that $s(0) \neq 0$. Let $\Omega = (G, \{s_v\}_{v\in V(G)}) \in$ \text{\sc{p-UnColHolant}}$(\mathcal{S})$. Consider the signature grid $\Omega' = (G, \{s'_v\}_{v \in V(G)})$ obtained by replacing every signature $s_v$, $v \in V(G)$ with the signature $s_v' = s_v/s_v(0)$. For any $k$, it can be readily verified that 
\[\textup{\textsf{UnColHolant}}(\Omega', k) = \textup{\textsf{UnColHolant}}(\Omega, k) \cdot \prod_{v \in V(G)}s_v(0)^{-1}.\] 
Furthermore, for any algebraic complex number $q$, the signature sets $\{s\}$ and $\{q\cdot s\}$ have the same type, because the signatures $s$ and $q\cdot s$ have the same fingerprints. Hence, it suffices to classify \text{\sc{p-UnColHolant}}$(\mathcal{S})$ for finite signature sets such that, for each $s \in \mathcal{S}$, $s(0) = 1$.
\end{remark}

\begin{lemma}\label{lem:main_result_on_zeta}
    Let $\mathcal{S}$ be a finite set of signatures and let $\mathcal{C}$ be the class of all graphs $H$ for which there is a positive integer $k$ and a colouring $\nu:V(H)\to \mathcal{S}$ such that $\zeta_{\mathcal{S},k}(H,\nu)\neq 0$. If $\mathcal{S}$ is of type $\mathbb{T}[\mathsf{lin}]$, then all graphs in $\mathcal{C}$ are acyclic. Otherwise $\mathcal{C}$ has unbounded treewidth.\qed
\end{lemma}

Note that Theorem~\ref{thm:main_uncol_restate} almost immediately follows from the combination of Lemma~\ref{lem:monotonicity_S-coloured} and Lemma~\ref{lem:main_result_on_zeta}. The only missing part is the ETH-based lower bound for signatures of type $\mathbb{T}[\infty]$. Fortunately, for this case we can easily reduce from the coloured version:

\begin{lemma}\label{lem:col_to_uncol}
    Let $\mathcal{S}$ be a finite set of signatures. We have 
    \[ \holantprob(\mathcal{S}) \fptlinred \text{\sc{p-UnColHolant}}(\mathcal{S})\,.\]
\end{lemma}
\begin{proof}
Let $\Omega=(G,\xi,\{s_v\}_{v\in V(G)})$ be a $k$-edge-coloured signature grid, and let $\tilde{\Omega}=(G,\{s_v\}_{v\in V(G)})$ be the underlying uncoloured signature grid.
For each $S\subseteq [k]$, define
\begin{align*}
f(S) := \sum_{\substack{A \in \binom{E(G)}{k}\\ \xi(A) = S}}~ \prod_{v\in V(G)} s_v(|A \cap E(v)|)  ~~~~\text{and}~~~~  g(S) := \sum_{\substack{A \in \binom{E(G)}{k}\\ \xi(A) \subseteq S}}~ \prod_{v\in V(G)} s_v(|A \cap E(v)|) \,. 
\end{align*} 
Moreover, we write $\tilde{\Omega}[S]$ for the signature grid obtained from $\tilde{\Omega}$ by deleting all edges $e$ with $\xi(e)\notin S$.
We next make the following three observations: 
\begin{enumerate}
    \item[(a)] For each $S\subseteq [k]$ we have $g(S)=\mathsf{UnColHolant}(\tilde{\Omega}[S],k)$.
    \item[(b)] $f([k]) = \holant(\Omega)$.
    \item[(c)] For each $S\subseteq [k]$ we have $g(S)=\sum_{T\subseteq S}f(T)$.
\end{enumerate}
By M\"obius inversion over the subset lattice, also called the weighted inclusion-exclusion principle, Observation (c) implies that $f([k]) = \sum_{T\subseteq [k]} (-1)^{k-|T|} g(T)$ (see e.g.\ \cite[Example 3.8.3]{Stanley11}). In combination with Observations (a) and (b), this yields
\[ \holant(\Omega) = \sum_{T\subseteq [k]} (-1)^{k-|T|} \cdot \mathsf{UnColHolant}(\tilde{\Omega}[T],k) \,.\]
Therefore, $\holant(\Omega)$ can be computed using $2^k$ calls to our oracle for $\textsc{p-UnColHolant}(\mathcal{S})$. Moreover, each oracle call is of size at most $|\Omega|$, and does not increase the parameter ($k$). This concludes the proof. 
\end{proof}

\begin{proof}[Proof of Theorem~\ref{thm:main_uncol_restate}]
Let $\mathcal{C}$ be the class of all graphs $H$ for which there is a positive integer $k$ and a colouring $\nu:V(H)\to \mathcal{S}$ such that $\zeta_{\mathcal{S},k}(H,\nu)\neq 0$. 
\begin{itemize}
    \item[(I)] If $\mathcal{S}$ of of type $\mathbb{T}[\mathsf{lin}]$, then each graph in $\mathcal{C}$ is acyclic by Lemma~\ref{lem:main_result_on_zeta}. Thus $\textsc{p-UnColHolant}(\mathcal{S})$ can be solved in FPT-near-linear time by Lemma~\ref{lem:monotonicity_S-coloured}.
    \item[(II)] Otherwise $\mathcal{C}$ has unbounded treewidth by Lemma~\ref{lem:main_result_on_zeta}, and thus $\textsc{p-UnColHolant}(\mathcal{S})$ is $\#\W[1]$-hard by Lemma~\ref{lem:monotonicity_S-coloured}. If, additionally, $\mathcal{S}$ is of type $\mathbb{T}[\infty]$, then, by Theorem~\ref{main_thm}, $\holantprob(\mathcal{S})$ cannot be solved in time $f(k)\cdot |\Omega|^{o(k/\log k)}$ for any function $f$, unless ETH fails. The same must then be true for $\textsc{p-UnColHolant}(\mathcal{S})$ since $\holantprob(\mathcal{S}) \fptlinred \text{\sc{p-UnColHolant}}(\mathcal{S})$ by Lemma~\ref{lem:col_to_uncol}.
\end{itemize}
\end{proof}

\subsection{Analysis of $\zeta_{\mathcal{S},k}$ and the Proof of Lemma~\ref{lem:main_result_on_zeta}}\label{sec:analysis_of_zeta}

For reasons of accessibility, we analyze the coefficients $\zeta_{\mathcal{S}, k}(H, \nu_H)$ in multiple fashions. First, we assume that the vertex-colouring of $H$ is given by $\nu_H : V(H) \to \{s\}$, that is, all vertices of $H$ have been assigned the same signature $s$. Under this assumption, we first study the case where $|E(H)| = k$ and then lift the results to the case where $|E(H)| \leq k$. In each case, we show how to adapt the results above to the setting with multiple signatures. 

When we deal with graphs coloured by a single signature, we can simplify notation based on the following observation. 

\begin{remark}\label{rem:simplified_single_signature}
Let $(H, \nu_H) \in \mathcal{G}(\mathcal{S})$, where $\mathcal{S} = \{s\}$, for some signature $s$, such that $s(0) = 1$. We assume that $|E(H)| \leq k$, for some $k > 0$. Let $\mathcal{G}_k$ be the set of all uncoloured simple graphs with no isolated vertices and precisely $k$ edges and set
\[\zeta(H, k):= \sum_{\substack{F\in \mathcal{G}_k}} \frac{\prod_{v\in V(F)}{s(d_F(v))}}{\#\mathsf{Aut}(F)} \cdot \sum_{\substack{\rho \in \mathsf{Part}(F)\\F/\rho \cong H}}\mu(\bot,\rho)  \,.\]
Then, it can be readily verified that $\zeta_{\mathcal{S}, k}(H, \nu_H) = \zeta(H, k)$, which allows us to assume that, in the setting of $\mathcal{S}$ containing only one signature, all graphs we consider are uncoloured and thus simplify the notation.
\end{remark}

\begin{theorem} \label{Thm:special_case}
Let $(H, \nu_H)$ be a vertex-coloured graph without isolated vertices and precisely $k$ edges. Its vertex-colouring is given by $\nu_H : V(H) \to \mathcal{S}$, where $\mathcal{S} = \{s\}$, for some signature $s$ such that $s(0) = 1$. Then we have
\begin{equation}
    \zeta_{\mathcal{S}, k}(H,\nu_H)=\frac{1}{\#\mathsf{Aut}(H, \nu_H)}\prod_{v \in V(H)} \chi(d_H(v), s)
\end{equation}
\end{theorem}

By \Cref{rem:simplified_single_signature}, it suffices to analyze $\zeta(H, k)$ and show
\[\zeta(H, k) = \frac1{\#\mathsf{Aut}(H)}\prod_{v \in V(H)}\chi(d_H(v), s).\]

We introduce one auxiliary notion which is used in the proof:
\begin{definition}
Given graphs $F,H$, a homomorphism $\phi : F \to H$ is called \emph{edge preserving} if it induces a bijection $E(F) \to E(H)$. 
\end{definition}

\begin{proof}[Proof of Theorem~\ref{Thm:special_case}]
Fix graphs $F,H \in \mathcal{G}_k$, and recall that neither $F$ nor $H$ have isolated vertices by definition of $\mathcal{G}_k$. Then any homomorphism $\phi : F \to H$, induces a partition 
$$R(\phi) = \{\phi^{-1}(v) : v \in V(H)\} \in \mathsf{Part}(F)$$
on the set of vertices of $F$. Consider the map
\begin{align} \label{eqn:mor_to_part}
    M_{F,H}=\{\phi: F \to H : \phi \text{ surjective, edge preserving}\} \xrightarrow{R} \{\rho \in \mathsf{Part}(F) : F/\rho \cong H \}\,.
\end{align}
First we note that the map is well-defined: indeed the image of the graph homomorphism $\phi$ is isomorphic to $F/R(\phi)$, since it's obtained by identifying vertices with the same image under $\phi$. On the other hand, $\phi$ being surjective implies that its image is all of $H$, so indeed we have $F/R(\phi) \cong H$.

The automorphism group $\mathsf{Aut}(H)$ acts freely on the left-hand side $M_{F,H}$ of \eqref{eqn:mor_to_part}. Indeed, if an automorphism $\psi: H \to H$ satisfies $\psi \circ \phi = \phi$, then it must be the identity on the vertices in the image $\phi(V(F))$. Since $\phi$ is surjective, we must have $\psi=\mathrm{id}_F$.

Given a map $f:X\to Y$, its \emph{fibers} are the sets in $f^{-1}(y):=\{x\in X\mid f(x)=y\}$ for an element $y\in Y$, which is a partition of $X$.
We claim that the map $R$ is surjective, and the orbits of the $\mathsf{Aut}(H)$-action on $M_{F,H}$ are precisely the fibers of $R$. To see the surjectivity, just note that for $\rho$ with $F/\rho \cong H$ we clearly get a map $\phi : F \to F/\rho \cong H$ which is contained in $M_{F,H}$ and satisfies $R(\phi)=\rho$.

Assume now that we have $\phi_1, \phi_2 \in M_{F,H}$ with $R(\phi_1) = R(\phi_2)$. Then we claim that there is a unique $\psi \in \mathsf{Aut}(H)$ with $\phi_2 = \psi \circ \phi_1$. Given $v \in V(H)$, let $w \in V(F)$ be a vertex with $\phi_1(w)=v$. Then we want to define $\psi(v)=\phi_2(w)$. This gives a well-defined map of sets $V(H) \to V(H)$ because $R(\phi_1)=R(\phi_2)$, and so the preimage $\phi_1^{-1}(v)$ is one block of the set of preimages of the map $\phi_2$. We claim that the map $\psi: V(H) \to V(H)$ of sets gives rise to a graph homomorphism $\psi: H \to H$. Indeed, if two vertices $v_1, v_2$ are connected in $H$, then there must be two preimages $w_1, w_2 \in V(F)$ under $\phi_1$ which are connected as well (since $\phi$ is edge-preserving). But then $\psi(v_1) = \phi_2(w_1)$ and $\psi(v_2)=\phi_2(w_2)$ are also connected by an edge in $H$ since $\phi_2$ is a graph homomorphism.

To summarize: we have shown that the map $R$ in \eqref{eqn:mor_to_part} is surjective, with fibers given by the free orbits of the group $\mathsf{Aut}(H)$.
What this intermediate result allows us to do, is to rewrite $\zeta(H, k)$ as
\begin{equation} \label{eqn:zeta_intermediate_form1}
\zeta(H,k) = \#\mathsf{Aut}(H)^{-1} \cdot \sum_{F \in \mathcal{G}_k} \left(\#\mathsf{Aut}(F)^{-1}\cdot  \prod_{u\in V(F)}s(d_F(u))\right)\cdot \sum_{\phi \in M_{F,H}}\left( (-1)^{|V(F)|-|V(H)|}\cdot \prod_{B \in R(\phi)} (|B|-1)!\right) \,.
\end{equation}

To simplify this expression further, we would like to find a combinatorially easier way to enumerate the surjective, edge-preserving homomorphisms $\phi: F \to H$ with $F$ having no isolated vertices. The crucial idea is to obtain them from \emph{fractures} of the graph $H$ (see \Cref{def:fracture})). 

For the subsequent analysis, it will be very convenient to consider \emph{half-edges} or \emph{flags} of a vertex in a graph.
Fixing the graph $H$ and a vertex $v \in V(H)$, for each edge incident to $v$, we consider the corresponding half-edge for which we only specify the endpoint $v$; formally, we can associate with each $e$ incident to $v$ a half-edge $e_v$. Then we denote by $\he(H,v)$ the set of all half-edges at $v$, and we can view a fracture $\vec{\rho}$ equivalently as a tuple containing a partition of $\he(H, v)$ for each vertex of $H$. Then, recalling that $\mathcal{F}(H)$ is the set of fractures of $H$ (see \Cref{def:fracture}), we have $\mathcal{F}(H)=\prod_{v \in V(H)} \mathsf{Part}(\he(H, v))$.

We claim that there is a well-defined map
\begin{equation} \label{eqn:second_map}
S : \mathcal{M} = \dot\bigcup_{F \in \mathcal{G}_k} M_{F,H} \to \mathcal{F}(H)\,.
\end{equation}
Given a pair $(F, \phi: F \to H)$ in the domain of \eqref{eqn:second_map}, note that $\phi$ induces a bijection $\phi^\he : \he(F) \to \he(H)$ from the set of (all) half-edges of $F$ to the set of (all) half-edges of $H$. This follows again since $\phi$ is assumed to be edge-preserving. Then given 
$v \in V(H)$ and $h_1, h_2 \in \he(H,v)$, they are in the same block $A$ inside $S(F, \phi)_v$ if and only if $(\phi^\he)^{-1}(h_1)$ and $(\phi^\he)^{-1}(h_2)$ are adjacent to the same vertex in $F$.

\emph{Note:} To avoid later confusion, we emphasize that in the disjoint union on the left-hand side of \eqref{eqn:second_map} we fix, once and for all, a \emph{representative} $F$ from each isomorphism class inside $\mathcal{G}_k$.

Next we analyze the properties of the map $S$. The first observation is that the map $S$ is surjective. Given a fracture $\vec\rho = (\rho_v)_{v \in V(H)} \in \mathcal{F}(H)$, consider the fractured graph $\fracture{H}{\vec\rho}$ according to \Cref{def:fract_graph}.
$\fracture{H}{\vec\rho}$ admits a natural surjective, edge-preserving homomorphism $\phi: \fracture{H}{\vec\rho} \to H$, sending the vertex $v^B \in \vec\rho(v) \subseteq V(\fracture{H}{\vec\rho})$ to $v \in V(H)$. 
Let $F \in \mathcal{G}_k$ be the chosen representative of the isomorphism class of $\fracture{H}{\vec\rho}$ and let $\eta : F \xrightarrow{\sim} \fracture{H}{\vec\rho}$ be such an isomorphism. Then the pair $(F, \phi \circ \eta:  F \to H)$ is an element of the left-hand side of \eqref{eqn:second_map} and by the construction of the fractured graph it maps to $\vec\rho$ under $S$.

What we note in this last step is that the choice of identification $\eta$ above is only unique up to pre-composing $\eta$ with an automorphism $\widetilde \eta \in \mathsf{Aut}(F)$. In fact, the group $\mathsf{Aut}(F)$ acts naturally on $M_{F,H} \subseteq \mathcal{M}$ (by pre-composition) and clearly the map $S$ is invariant under this action (two half-edges being adjacent to the same vertex of $F$ can equivalently be checked after applying an automorphism of $F$). 
Moreover, we claim, again, that the automorphisms of $F$ act freely on $M_{F,H}$. Indeed, any automorphism $\widetilde \eta$ of $F$ fixing a map $\phi: F \to H \in M_{F,H}$ must act as the identity on the set of half-edges of $F$ (since they map bijectively to the half-edges of $H$ under $\phi^\he$). But by the assumption that $F$ has no isolated vertices, any vertex is adjacent to at least one half-edge, and so $\widetilde \eta$ also has to fix all vertices of $F$ (and thus must be equal to the identity $\widetilde \eta = \mathrm{id}_F$).

The final property of $S$ that we want to show is that the fibers of $S$ are precisely the orbits under these automorphism groups.
To state this more rigorously: let $(F, \phi: F \to H)$ and $(F', \phi': F' \to H)$ be two elements of $\mathcal{M}$. We claim
\begin{equation}  \label{eqn:S_equivalence_claim}
    S(F, \phi: F \to H) = S(F', \phi': F' \to H) \iff F = F' \text{ and there exists $\widetilde \eta \in \mathsf{Aut}(F)$ with }\phi' = \phi \circ \widetilde \eta\,.
\end{equation}
To see this, assume that $S(F, \phi: F \to H) = \vec\rho$. Then we claim that we have a natural factorization of $\phi$ as
\[
\phi = (F \xrightarrow{\widetilde \phi} \fracture{H}{\vec\rho} \to H)
\]
and $\widetilde \phi$ is an isomorphism. Indeed, given a vertex $w \in V(F)$ with $\phi(w)=v \in V(H)$ we have by definition that $A=\he(F,w)$ is one of the blocks of the partition $\vec\rho(v)$. On the other hand, we constructed the vertices of $\fracture{H}{\vec\rho}$ to exactly be those blocks, and then $\widetilde \phi$ is just defined by $\widetilde \phi(w)=A$. It's a straightforward check that $\widetilde \phi$ is a graph homomorphism and in fact an isomorphism (again using that $\phi$ is edge-preserving).

Now we can finish the proof of the equivalence \eqref{eqn:S_equivalence_claim}: if $S(F, \phi: F \to H) = \xi = S(F', \phi': F' \to H)$, then we obtain isomorphisms $\widetilde \phi, \widetilde \phi'$ fitting into the diagram
\[
\begin{tikzcd}
F \arrow[dr, "\widetilde \phi"]  \arrow[ddr, bend right, "\phi", swap] & & F' \arrow[ld, "\widetilde \phi'", swap] \arrow[ddl, bend left, "\phi'"]\\
& \fracture{H}{\vec\rho} \arrow[d] &\\
& H &
\end{tikzcd}
\]
Since $F,F'$ were chosen as unique representatives of the isomorphism classes $\mathcal{G}_k$, this proves that $F=F'$ since they are in fact isomorphic (via the composition $\widetilde \eta = (\widetilde \phi)^{-1} \circ \widetilde \phi'$). On the other hand, the diagram also proves that for $\widetilde \eta\in \mathsf{Aut}(F)$ we have $\phi' = \phi \circ \widetilde \eta$. This relies on the fact that we used the same natural map $H[\xi] \to H$ in both constructions.

To summarize again: the map $S : \mathcal{M} = \dot\bigcup_{F} M_{F,H} \to \mathcal{F}(H)$ is surjective, each of the $M_{F,H}$ admits a free action of $\mathsf{Aut}(F)$ and the orbits of these group actions are precisely the fibers of the map $S$.

To conclude the desired formula for $\zeta(H,k)$ it remains to note that all the terms in the formula \eqref{eqn:zeta_intermediate_form1}, which a priori depend on the full tuple $(F, \phi: F \to H) \in \mathcal{M}$, in fact only depend on its image $\vec\rho = S(F, \phi)$ in $\mathcal{F}(H)$. Indeed, firstly we have
\begin{align*}
\prod_{u \in V(F)} s(d_F(u)) = \prod_{v \in V(H), B \in \vec\rho(v)} s(|B|)
\end{align*}
since the degree $d_F(u)$ is precisely the cardinality of the set $B$ of half-edges adjacent to $u$.

Secondly, we have
\[
|V(F)| - |V(H)| =\left( \sum_{v \in V(H)} |\vec\rho(v)| \right) - \left(\sum_{v \in V(H)} 1 \right)  = \sum_{v \in V(H)} |\vec\rho(v)| -1\,,
\]
since the vertices of $F$ over $v \in V(H)$ are in bijection with the blocks of $\vec\rho(v)$. In particular, we have
\[
(-1)^{|V(F)| - |V(H)|} = \prod_{v \in V(H)} (-1)^{|\vec\rho(v)| -1}
\]
Finally, we have
\[
\prod_{B \in R(\phi)} (|B|-1)! = \prod_{v \in V(H)} (|\vec\rho(v)|-1)!\,,
\]
again for the same reason: the blocks $B$ correspond to preimages of vertices $v$ of $H$, which are described by $\vec\rho(v)$.

Putting this all together, we can rewrite \eqref{eqn:zeta_intermediate_form1} as
\begin{align} \label{eqn:zeta_penultimate_form}
\zeta(H,k) &= \#\mathsf{Aut}(H)^{-1} \cdot \sum_{(F,\phi) \in \mathcal{M}} \#\mathsf{Aut}(F)^{-1}\cdot \left(  \prod_{u\in V(F)}s(d_F(u)) \right) \cdot  (-1)^{|V(F)|-|V(H)|}\cdot \prod_{B \in R(\phi)} (|B|-1)!\\
&= \#\mathsf{Aut}(H)^{-1} \cdot \sum_{(F,\phi) \in \mathcal{M}} \#\mathsf{Aut}(F)^{-1}\cdot  \underbrace{\prod_{v\in V(H)}\left(\prod_{B \in \vec\rho(v)} s(|B|) \cdot (-1)^{|\vec\rho(v)| -1} \cdot (|\vec\rho(v)|-1)! \right)}_{=:\mathsf{Cont}(\vec\rho)}\,, \nonumber
\end{align}
where $\vec\rho=S(F,\phi)$ in the expression of $\mathsf{Cont}(\vec\rho)$. Now we simply rearrange the sum according to the possible values of $\vec\rho$ and find
\[
\zeta(H,k)= \#\mathsf{Aut}(H)^{-1} \cdot \sum_{\vec\rho \in \mathcal{F}(H)} \mathsf{Cont}(\vec\rho) \cdot  \underbrace{\sum_{(F,\phi) \in S^{-1}(\vec\rho)} \#\mathsf{Aut}(F)^{-1}}_{=1}\,,
\]
where we (finally) use that the fibers of $S$ are exactly free orbits under the groups $\mathsf{Aut}(F)$. 

Recall that the set of fractures can be written as $\mathcal{F}(H)=\prod_{v \in V(H)} \mathsf{Part}(\he(H, v)).$
Approaching the finish line of this rather involved calculation, we conclude
\begin{align}
\zeta(H,k) &= \#\mathsf{Aut}(H)^{-1} \cdot \sum_{\vec\rho\in \prod_{v \in V(H)} \mathsf{Part}(\he(H, v))} \ \ \ \prod_{v\in V(H)}\left(\prod_{B \in \vec\rho(v)} s(|B|) \cdot (-1)^{|\vec\rho(v)| -1} \cdot (|\vec\rho(v)|-1)! \right) \nonumber\\
&= \#\mathsf{Aut}(H)^{-1} \cdot \prod_{v\in V(H)} \sum_{\vec\rho(v) \in \mathsf{Part}(\he(H, v))} \left(\prod_{A \in \vec\rho(v)} s(|A|) \cdot (-1)^{|\vec\rho(v)| -1} \cdot (|\vec\rho(v)|-1)! \right) \label{eqn:distributivity_trick}\\
&= \#\mathsf{Aut}(H)^{-1} \cdot \prod_{v\in V(H)} \chi(\underbrace{|\he(H, v)|}_{=d_H(v)}, s)\,, \nonumber
\end{align}
where the second equality uses distributivity of the product over $v \in V(H)$
\end{proof}

\begin{theorem}\label{thm:k_edges_multiple_signatures}
Let $(H, \nu_H)$ be a vertex-coloured graph with no isolated vertices and precisely $k$ edges. Its vertex colouring is given by $\nu_H : V(G) \to \mathcal{S}$, where $\mathcal{S}$ is a finite set of signatures $s$ such that $s(0) = 1$. Then, we have
\[\zeta_{\mathcal{S}, k}(H, \nu_H) = \frac1{\#\mathsf{Aut}(H, \nu_H)}\prod_{v\in V(H)}\chi(d_H(v), \nu_H(v)).\]
\end{theorem}

\begin{proof}
Recall that 
\[\zeta_{\mathcal{S}, k}(H,\nu_H) = \!\!\!\sum_{(F,\nu_F) \in \mathcal{G}_k(\mathcal{S})} \left(\#\mathsf{Aut}(F, \nu_F)^{-1}\cdot\prod_{u\in V(F)}\nu_F(u)(d_F(u))\right)\cdot \sum_{\substack{\rho \in \mathsf{Part}(F, \nu_F)\\(F, \nu_F)/\rho \cong (H, \nu_H)}}\left( (-1)^{|V(F)|-|V(H)|}\cdot \prod_{B \in \rho} (|B|-1)!\right). \]
We now consider homomorphisms $\phi : (F, \nu_F) \to (H, \nu_H)$ that preserve colours. Colour-preserving homomorphisms induce colour-consistent partitions $\rho \in \mathsf{Part}(F, \nu_F)$ and if, in addition, $\phi$ is surjective, then $(F, \nu_F)/\rho \cong (H, \nu_H)$. Let $M_{(F, \nu_F), (H, \nu_H)}$ denote the restriction of $M_{F, H}$ on homomorphisms $\phi \in (F, \nu_F) \to (H, \nu_H)$ that preserve colours. Then, $R$ defines a surjection
\begin{equation} \label{eqn:R_coloured}
R: M_{(F, \nu_F), (H, \nu_H)} \to \{\rho \in \mathsf{Part}(F, \nu_F) : (F, \nu_F)/\rho \cong (H, \nu_H)\}
\end{equation}
For $\phi_1, \phi_2 \in M_{(F, \nu_F), (H, \nu_H)}$ with $R(\phi_1) = R(\phi_2)$, there is a unique $\psi \in \mathsf{Aut}(H)$, such that $\phi_2 = \psi \circ \phi_1$. By construction, $\psi$ preserves the colours, and so $\psi \in \mathsf{Aut}(H, \nu_H)$. Furthermore, $\mathsf{Aut}(H, \nu_H)$ acts freely on $M_{(F, \nu_F), (H, \nu_H)}$ as a subgroup of $\mathsf{Aut}(H)$ acting freely on $M_{F, H} \supseteq M_{(F, \nu_F), (H, \nu_H)}$. Hence, the fibers of the map \eqref{eqn:R_coloured} coincide with the free orbits of the group $\mathsf{Aut}(H, \nu_H)$. So,
\begin{align*}
\zeta_{\mathcal{S}, k}(H,\nu_H) = \#\mathsf{Aut}(H, \nu_H)^{-1}  \cdot \sum_{(F, \nu_F)\in \mathcal{G}_k(\mathcal{S})} &\left(\#\mathsf{Aut}(F, \nu_F)^{-1}\cdot   \prod_{u\in V(F)}\nu_F(u)(d_F(u))\right) \\ & \cdot \sum_{\phi \in M_{(F,\nu_F),(H,\nu_H)}}\left( (-1)^{|V(F)|-|V(H)|}\cdot \prod_{B \in R(\phi)} (|B|-1)!\right).
\end{align*}
Next we observe that there is an one-to-one correspondence between the sets $\mathcal{M} = \dot\bigcup_{F \in \mathcal{G}_k}M_{F, H}$ and $\mathcal{C} = \dot\bigcup_{(F, \nu_F) \in \mathcal{G}_k(\mathcal{S})}M_{(F, \nu_F), (H, \nu_H)}$. This holds, because for $(F, \phi) \in \mathcal{M}$, we can uniquely define a vertex-colouring $c_{\phi} : V(F) \to \mathcal{S}$ where for $u \in V(F)$, $c_{\phi}(u) = \nu_H(\phi(u))$ such that $((F, c_{\phi}), \phi) \in \mathcal{C}$. For a fracture $\vec\rho \in \mathcal{F}(H)$ of the uncolored graph $H$, the fractured graph $\fracture{H}{\vec\rho}$ admits a natural vertex-colouring $\nu_{\vec\rho}$ that assigns each block $B \in \vec\rho(v)$ the colour $\nu_H(v), v \in V(H)$. 
Recall the map $S : \mathcal{M} \to \mathcal{F}(H)$. By slightly abusing notation, we define $S : \mathcal{C} \to \mathcal{F}(H)$, taking $S((F, \nu_F), \phi) = S(F, \phi)$. We observe that by definition $S$ preserves colours in the sense that $(F, \nu_F) \cong (\fracture{H}{\vec\rho}, \nu_{\vec\rho})$, where $\vec\rho = S(F, \phi)$. Similarly to the proof of \Cref{Thm:special_case}, if $S((F, \nu_F), \phi)) = S((F', \nu_{F'}), \phi')$, then $(F, \nu_F) \cong (F', \nu_{F'})$ via an isomorphism $\widetilde{\eta}$ $(\in \mathsf{Aut}(F, \nu_F))$ which can be used to show that $\phi' = \phi \circ \widetilde{\eta}$. 
Furthermore, $\mathsf{Aut}(F, \nu_F)$ acts freely on $M_{(F, \nu_F), (H, \nu_H)}$. We have
\begin{align*}
\zeta_{\mathcal{S}, k}(H, \nu_H) &= \#\mathsf{Aut}(H, \nu_H)^{-1} \cdot \sum_{((F,\nu_F),\phi) \in \mathcal{C}} \#\mathsf{Aut}(F, \nu_F)^{-1}\cdot  \underbrace{\prod_{v\in V(H)}\left(\prod_{B \in \vec\rho(v)} \nu_H(v)(|B|) \cdot (-1)^{|\vec\rho(v)| -1} \cdot (|\vec\rho(v)|-1)! \right)}_{=:\mathsf{Cont}(\vec\rho)} \\ &= \#\mathsf{Aut}(H, \nu_H)^{-1}\cdot\sum_{\vec\rho \in \mathcal{F}(H)}\mathsf{Cont}(\vec\rho)\cdot\sum_{((F, \nu_F), \phi) \in S^{-1}(\vec\rho)}\#\mathsf{Aut}(F, \nu_F)^{-1}\,,
\end{align*}
where,
\begin{align*}
\sum_{((F, \nu_F), \phi) \in S^{-1}(\vec\rho)}\#\mathsf{Aut}(F, \nu_F)^{-1} &= \sum_{\substack{(F, \nu_F) \\ \exists \phi : ((F, \nu_F), \phi) \in S^{-1}(\vec\rho)}}\mathsf{Aut}(F, \nu_F)^{-1}\cdot\sum_{\phi: ((F, \nu_F), \phi) \in S^{-1}(\vec\rho)}1 \\ &= \#\{{(F, \nu_F) \mid \exists\phi : ((F, \nu_F), \phi) \in S^{-1}(\vec\rho)}\} = 1\,,    
\end{align*}
since we consider unique representatives of isomorphism classes of coloured graphs.
The rest of the proof follows verbatim the distributivity argument of \Cref{Thm:special_case}.
\end{proof}

For the general case, we introduce some additional notation:
\begin{itemize}
    \item A \emph{partition} $\lambda$ of a positive integer $d$ is a decomposition of $d$ into an (unordered) sum of positive integers (e.g. $\lambda=3+2+2+1$ is a partition of $d=8$). Sometimes these are written in exponential notation, e.g.
    \[
    4+3+3+3+2+1+1+1+1+1 = 4^1 3^3 2^1 1^5.
    \]
    \item We write $|\lambda|=d$ for the sum of all parts of the partition and $\mathsf{len}(\lambda)$ for the number of its summands. E.g. we have
    \[
    |3+2+2+1| = 8 \text{ and }\mathsf{len}(3+2+2+1)=4.
    \]
    \item Given two partitions $\lambda_1, \lambda_2$, their \emph{union} $\lambda_1 \cup \lambda_2$ is formed by combining the summands from each of the partitions:
    \[
    (3+2+2+1) \cup (2+1+1+1) = 3 + 2+2+2 + 1 + 1 + 1 + 1.
    \]
    Similarly, given a finite family $(\lambda_i)_{i \in I}$, we write $\bigcup_{i \in I} \lambda_i$ for their union.
    \item We write $\mathcal{P}_d$ for the set of all partitions of $d$ and $\mathcal{P} = \dot\bigcup_{d \geq 1} \mathcal{P}_d$ for the set of all partitions.
    \item Given a set-partition $\rho$ of $[d]$, its \emph{shape} is the partition of $d$ obtained from the size of the blocks $B \in \rho$, e.g.
    \[
    \mathsf{shape}(\{1,4\} \cup \{2,5,6\} \cup \{3\} \cup \{7\}) = 3 + 2 + 1 + 1.
    \]
    \item The set of \emph{automorphisms} of a partition $\lambda$ is the set of permutations of its summands leaving its shape unchanged. Its cardinality can be computed as the product of factorials from its exponential notation:
    \[
    \# \mathsf{Aut}(4^1 3^3 2^1 1^5) = 1! \cdot 3! \cdot 1! \cdot 5!.
    \]
    \item For an integer partition $\lambda$ of $d$, we define a multiplicity associated to $\lambda$ as
    \[
    \mathsf{mult}(\lambda) = \sum_{\substack{\rho \in P(d):\\\mathsf{shape}(\rho)=\lambda}} \mu(\{1\} \cup \{2\} \cup \ldots \cup \{d\}, \rho)\,,
    \]
    where $\mu$ is the Moebius function of the set partition lattice. This can be explicitly calculated as
    \begin{equation}
        \mathsf{mult}(\lambda) = \underbrace{\frac{1}{\# \mathsf{Aut}(\lambda)} \binom{d}{\lambda}}_{=\#\{\rho \in P(d): \mathsf{shape}(\rho)=\lambda\}} \cdot \underbrace{(-1)^{d-\mathsf{len}(\lambda)} \prod_{\lambda_i \in \lambda} (\lambda_i-1)!}_{=\mu(\{1\} \cup \{2\} \cup \ldots \cup \{d\}, \rho)}\,,
    \end{equation}
    where $\binom{d}{\lambda} = d! / \prod_{\lambda_i \in \lambda} \lambda_i!$ is the multinomial coefficient.
    \item We also define a generalization of the function $\chi(-, s)$ from above which takes partitions as arguments. Let $\lambda = \lambda_1 + \ldots + \lambda_\ell$ be an integer partition. Then we define
    \begin{equation} \label{eqn:chi_generalization}
        \chi(\lambda , s) := \sum_{\sigma\in P(\ell)} (-1)^{|\sigma|-1} (|\sigma|-1)! \cdot \prod_{B \in \sigma} s\left(\sum_{i \in B} \lambda_i \right)\,.
    \end{equation}
    One checks that for $\lambda = 1 + 1 + \ldots + 1 = 1^d$ we indeed have $\chi(\lambda, s)=\chi(d,s)$ in the notation above.
\end{itemize}

\begin{theorem}\label{thm:zeta_at_most_k_edges}
Let $H$ be a vertex-coloured graph without isolated vertices and at most $k \geq 0$ edges. Its vertex-colouring is given by $\nu_H : V(H) \to \mathcal{S}$, where $\mathcal{S} = \{s\}$, for some signature $s$ such that $s(0) = 1$. Then we have
\begin{equation} \label{eqn:Theorem_general1}
    \zeta_{\mathcal{S}, k}(H,\nu_H)=\frac{1}{\#\mathsf{Aut}(H, \nu_H)} \cdot \sum_{\substack{\lambda: E(H) \to \mathcal{P}\\\sum_e |\lambda(e)| = k }}\ \  \prod_{e \in E(H)} \frac{\mathsf{mult}(\lambda(e))}{|\lambda(e)|!}  \prod_{v \in V(H)} \chi(\mathsf{deg}(H, v, \lambda), s)\,,
\end{equation}
where the partitions $\mathsf{deg}(H, v, \lambda)$ are defined as
\[
\mathsf{deg}(H, v, \lambda) = \bigcup_{\substack{e \in E(H):\\ e \text{ incident to }v}} \lambda(e)\,. 
\]
\end{theorem}
\begin{proof}
Using again \Cref{rem:simplified_single_signature} for simplifying notation, it suffices to analyze $\zeta(H, k)$ and show
\[
\zeta(H,k)=\frac{1}{\#\mathsf{Aut}(H)} \cdot \sum_{\substack{\lambda: E(H) \to \mathcal{P}\\\sum_e |\lambda(e)| = k }}\ \  \prod_{e \in E(H)} \frac{\mathsf{mult}(\lambda(e))}{|\lambda(e)|!}  \prod_{v \in V(H)} \chi(\mathsf{deg}(H, v, \lambda), s).
\]

Very similar to the first part of the proof of Theorem \ref{Thm:special_case} we see that
\begin{equation} \label{eqn:zeta_intermediate_form}
\zeta(H,k) = \#\mathsf{Aut}(H)^{-1} \cdot \sum_{F \in \mathcal{G}_k} \left(\#\mathsf{Aut}(F)^{-1}\cdot  \prod_{u\in V(F)}s(d_F(u))\right)\cdot \sum_{\phi \in M_{F,H}}\left( (-1)^{|V(F)|-|V(H)|}\cdot \prod_{B \in R(\phi)} (|B|-1)!\right) \,.
\end{equation}
where however we now need to define
\[
M_{F, H} = \{\phi: F \to H : \phi \text{ surjective on vertices and edges}\}\,.
\]
In the case where $F,H$ have the same number $k$ of edges, this condition is equivalent to being edge-preserving. To start approaching the shape of \eqref{eqn:Theorem_general1} we note that given $\phi \in M_{F,H}$ we obtain a map $d_\phi: E(H) \to \mathbb{Z}_{>0}$ satisfying $\sum_e d_\phi(e) = k$ by
\begin{equation}
    d_\phi(e) = \# \{\widehat e \in E(F) : \phi_E(\widehat e) = e\}\,,
\end{equation}
where $\phi_E : E(F) \to E(H)$ is the map on edges induced by the graph homomorphism $\phi$. Given any map $d : E(H) \to \mathbb{Z}_{>0}$ with $\sum_e d(e)=k$, we define
\begin{equation} \label{eqn:zeta_d}
\zeta_d(H, k) = \frac{1}{\#\mathsf{Aut}(H)} \cdot \sum_{\substack{\lambda: E(H) \to \mathcal{P}\\|\lambda(e)| = d(e) }} \ \  \prod_{e \in E(H)} \frac{\mathsf{mult}(\lambda(e))}{|\lambda(e)|!}  \prod_{v \in V(H)} \chi(\mathsf{deg}(H, v, \lambda), s)\,,
r\end{equation}
so that $\zeta(H,k) = \sum_{d} \zeta_d(H, k)$. We show the finer statement that $\zeta_d(H, k)$ computes the sum of all contributions $(F, \phi)$ of the summation \eqref{eqn:zeta_intermediate_form} with $d_\phi = d$, which would finish the proof.

To show this equality, we first define an auxiliary object: a graph $H(d)$ with multi-edges replacing each edge $e$ of $H$ with $d(e)$ copies of the edge:
\[
H = \begin{tikzpicture}
    \draw[fill] (0,0) circle (1pt);
    \draw[fill] (1,0) circle (1pt);
    \draw[fill] (2,0) circle (1pt);

    \draw (0,0) --node[midway, above] {$e_1$} (1,0) --node[midway, above] {$e_2$} (2,0);
\end{tikzpicture}, d(e_1)=3, d(e_2)=2 \implies H(d) = 
\begin{tikzpicture}
    \draw[fill] (0,0) circle (1pt);
    \draw[fill] (1,0) circle (1pt);
    \draw[fill] (2,0) circle (1pt);

    \draw (0,0) -- (1,0);
    \draw (0,0) to[bend left] (1,0);
    \draw (0,0) to[bend right] (1,0);
    \draw (1,0) to[bend left] (2,0);
    \draw (1,0) to[bend right] (2,0);
\end{tikzpicture}
\]
For each half edge $h \in \he(H, v)$ forming part of an edge $e=\{h,h'\} \in E(H)$, the new graph $H(d)$ has $d(e)$ many half-edges $h_1, \ldots, h_{d(e)} \in \he(H(d), v)$ adjacent to the vertex $v$.

The utility of the multi-edge graph $H(d)$ comes from the fact that any homomorphism $\phi: F \to H$ with $d_\phi=d$ factors (in a non-unique way) as 
\[
\begin{tikzcd}
    F \arrow[r, "\phi'"] \arrow[rr, bend right, "\phi", swap] & H(d) \arrow[r] & H
\end{tikzcd}
\]
 with $\phi'$ edge-preserving, and so $F$ can be obtained from $H(d)$ by a suitable fracture
\[
\vec\rho \in \mathcal{F}(H(d)) = \prod_{v \in V(H)} \mathsf{Part}(\he(H(d), v))\,.
\]
This is illustrated by the following picture:
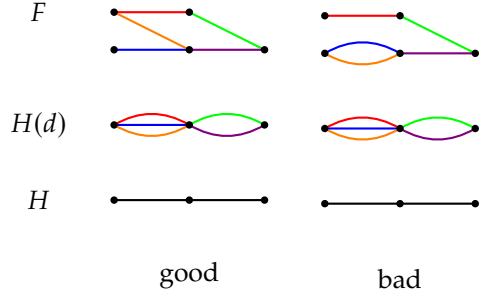
\[
\begin{tikzpicture}[thick]
\begin{scope}[yshift=1cm]
    \draw[blue] (0,0) -- (1,0);
    \draw[red] (0,0.5) to (1,0.5);
    \draw[orange] (0,0.5) to (1,0);
    \draw[green] (1,0.5) to (2,0);
    \draw[violet] (1,0) to (2,0);
    \draw[fill] (0,0) circle (1pt);
    \draw[fill] (0,0.5) circle (1pt);
    \draw[fill] (1,0) circle (1pt);
    \draw[fill] (1,0.5) circle (1pt);
    \draw[fill] (2,0) circle (1pt);
\end{scope}
    \draw[blue] (0,0) -- (1,0);
    \draw[red] (0,0) to[bend left] (1,0);
    \draw[orange] (0,0) to[bend right] (1,0);
    \draw[green] (1,0) to[bend left] (2,0);
    \draw[violet] (1,0) to[bend right] (2,0);
    \draw[fill] (0,0) circle (1pt);
    \draw[fill] (1,0) circle (1pt);
    \draw[fill] (2,0) circle (1pt);
    \draw (1,-2) node {\text{good}};
\begin{scope}[yshift=-1cm]
    \draw[fill] (0,0) circle (1pt);
    \draw[fill] (1,0) circle (1pt);
    \draw[fill] (2,0) circle (1pt);

    \draw (0,0) -- (1,0) -- (2,0);
\end{scope}

    \draw (-1,1.5) node {$F$};
    \draw (-1,0) node {$H(d)$};
    \draw (-1,-1) node {$H$};
\end{tikzpicture}
\quad \quad
\begin{tikzpicture}[thick]
\begin{scope}[yshift=1cm]
    \draw[blue] (0,0) to[bend left] (1,0);
    \draw[red] (0,0.5) to (1,0.5);
    \draw[orange] (0,0) to[bend right] (1,0);
    \draw[green] (1,0.5) to (2,0);
    \draw[violet] (1,0) to (2,0);
    
    \draw[fill] (0,0) circle (1pt);
    \draw[fill] (0,0.5) circle (1pt);
    \draw[fill] (1,0) circle (1pt);
    \draw[fill] (1,0.5) circle (1pt);
    \draw[fill] (2,0) circle (1pt);
\end{scope}
    
    \draw[blue] (0,0) -- (1,0);
    \draw[red] (0,0) to[bend left] (1,0);
    \draw[orange] (0,0) to[bend right] (1,0);
    \draw[green] (1,0) to[bend left] (2,0);
    \draw[violet] (1,0) to[bend right] (2,0);
    
    \draw[fill] (0,0) circle (1pt);
    \draw[fill] (1,0) circle (1pt);
    \draw[fill] (2,0) circle (1pt);
\begin{scope}[yshift=-1cm]
    \draw[fill] (0,0) circle (1pt);
    \draw[fill] (1,0) circle (1pt);
    \draw[fill] (2,0) circle (1pt);

    \draw (0,0) -- (1,0) -- (2,0);
\end{scope}
        
        \draw (1,-2) node {\text{bad}};
\end{tikzpicture}
\]
We can see that for some partitions $\vec\rho$ the fractured graph $F=\fracture{H(d)}{\vec\rho}$ is an honest (simple) graph, whereas for others it would still contain multi-edges, and thus does not appear in the summation \eqref{eqn:zeta_d}. These bad cases are precisely characterized by the property that there exist two vertices $v, v' \in V(H(d))$ connected by edges $\{h_1, h_1'\}, \{h_2, h_2'\}$ such that $h_1, h_2$ are in the same block of $\vec\rho(v)$ and $h_1', h_2'$ are in the same block of $\vec\rho(v')$. 

In fact, we get a well-defined map
\begin{equation} \label{eqn:ep_map}
\mathsf{ep} : \mathcal{F}(H(d)) \to \prod_{e \in E(H)} \mathsf{Part}(\{e_1, \ldots, e_{d(e)}\}) =: \mathsf{Part}(d)\,,
\end{equation}
where for each $e \in E(H)$ we again denote $e_1, \ldots, e_{d(e)} \in E(H(d))$ the edges lying over $e$. The map $\mathsf{ep}$ sends a fracture $\vec\rho$ to the unique collection of set-partitions $\mathsf{ep}(\vec\rho) = (\sigma_e)_{e \in E(H)}$ such that two edges $e_i, e_j \in E(H(d))$ lying over the same edge $e \in E(H)$ are in the same block of $\sigma_e$ if and only if the lifts of $e_i, e_j$ in $F=\fracture{H(d)}{\vec\rho}$ connect the same two vertices in $F$.
Denote by 
\begin{equation} \label{eqn:F_adm_def}
    \mathcal{F}^\textup{adm}(H(d)) = \mathsf{ep}^{-1}\left((\{e_1\} \cup \ldots \cup \{e_{d(e)}\})_{e \in E(H)} \right) \subseteq \mathcal{F}(H(d))
\end{equation}
the set of "good" $\vec\rho$ for which $F=\fracture{H(d)}{\vec\rho}$ is a graph without multi-edges. This condition is indeed equivalent to $\mathsf{ep}(\vec\rho)$ being the product of finest set-partitions, only consisting of singleton sets $\{e_i\}$, i.e. no two edges of $\fracture{H(d)}{\vec\rho}$ have the same start- and end-vertex.

Then similar arguments as before show that the sum of all terms in \eqref{eqn:zeta_intermediate_form} with $d_\phi = d$ is given by
\begin{equation} \label{eqn:zeta_d_intermediate_form}
\frac{1}{\#\mathsf{Aut}(H) \cdot \prod_{e \in E(H)} d(e)!}  \cdot \sum_{\vec\rho \in \mathcal{F}^\textup{adm}(H(d))} \ \ \ \underbrace{\prod_{v\in V(H)}\left(\prod_{B \in \vec\rho(v)} s(|B|) \cdot (-1)^{|\vec\rho(v)| -1} \cdot (|\vec\rho(v)|-1)! \right)}_{=:\mathsf{Cont}(\vec\rho)} \,.
\end{equation}
Here the additional factor $1/\prod_{e \in E(H)} d(e)!$ comes from the set of permutations of the edges of $H(d)$ that leave all vertices fixed. Indeed, since in a lift of the map $\phi: F \to H$ to an edge-preserving homomorphism $\phi': F \to H(d)$ we do not see which edges of $F$ map to which edges of $H(d)$, there are exactly $\prod_{e \in E(H)} d(e)!$ such lifts. 

Compared to the situation of Theorem \ref{Thm:special_case} this sum is more complicated, precisely because the condition $\vec\rho \in \mathcal{F}^\textup{adm}(H(d))$ \emph{cannot} be checked independently at each vertex (and so $\mathcal{F}^\textup{adm}(H(d))$ does not decompose as a product over vertices). 
To overcome this challenge, we calculate the sum using an inclusion-exclusion formula on the whole set $\mathcal{F}(H(d))$ of fractions, where each part of the inclusion-exclusion summation \emph{will} feature an index set which is a product over the vertices.

More concretely, if we consider the variant of formula \eqref{eqn:zeta_d_intermediate_form} where the sum goes over \emph{all} $\vec\rho \in \mathcal{F}(H(d))$, then the same argument as in \eqref{eqn:distributivity_trick}, via distributivity, shows that the result equals
\[
\frac{1}{\#\mathsf{Aut}(H) \cdot \prod_{e \in E(H)} d(e)!} \cdot \prod_{v\in V(H)} \chi(|\he(H(d), v)|, s)\,.
\]
Comparing with the formula \eqref{eqn:zeta_d} that we want to prove, this is precisely the part of that formula associated to the map
\[
\lambda : E(H) \to \mathcal{P}, e \mapsto 1 + 1 + \ldots + 1 = 1^{d(e)}\,.
\]
Indeed, we have $|\lambda(e)|=d(e)$, $|d(H,v,\lambda)|=|\he(H(d), v)|$ and from the definition of $\mathsf{mult}$ it is easy to verify that $\mathsf{mult}(1+1+\ldots+1) = 1$. 

As mentioned before, we need to correct this formula by subtracting all terms associated to $\vec\rho \in \mathcal{F}(H(d)) \setminus \mathcal{F}^\textup{adm}(H(d))$. Our approach is to do this by considering the map $\mathsf{ep}$ from \eqref{eqn:ep_map} and performing an inclusion-exclusion calculation over the target of that map, which is the product $\mathsf{Part}(d)$ of partition lattices $\mathsf{Part}(\{e_1, \ldots, e_{d(e)}\})$. 

Let $\sigma = (\sigma_e)_{e \in E(H)} = (B_1^e \dot\cup B_2^e \dot\cup \ldots \dot\cup B_{m_e}^e)_e$ be an element of $\mathsf{Part}(d)$. Consider the set $\mathsf{Part}(d)^{\geq \sigma} \subseteq \mathsf{Part}(d)$ of all partitions which are \emph{coarsenings} of $\sigma$ (i.e. obtained by possibly combining some of the blocks $B_i^e, B_j^e$). Given any subset $S \subseteq \mathsf{Part}(d)$ let 
\[
\boldone_S : \mathsf{Part}(d) \to \mathbb{Z}, \sigma \mapsto \begin{cases}
    1& \text{ if }\sigma \in S,\\
    0& \text{ otherwise}
\end{cases}
\]
be the characteristic function of $S$. Then one (fancy) version of an inclusion-exclusion statement states that
\begin{equation}  \label{eqn:char_function_moebius}
\boldone_{ \{(\{e_1\} \cup \ldots \cup \{e_{d(e)}\})_{e \in E(H)}\}  } = \sum_{\rho \in \mathsf{Part}(d)}  \left(\prod_{e \in E(H)} \mu(\{e_1\} \cup \ldots \cup \{e_{d(e)}\}, \rho_e) \right) \cdot \boldone_{\mathsf{Part}(d)^{\geq \sigma}}\,.
\end{equation}
To obtain the desired quantity \eqref{eqn:zeta_d_intermediate_form} we then use the following chain of equalities:
\begin{align}
\sum_{\vec\rho \in \mathcal{F}^\textup{adm}(H(d))} \mathsf{Cont}(\xi) &\overset{\eqref{eqn:F_adm_def}}{=} \sum_{\vec\rho \in \mathcal{F}(H(d))} \mathsf{Cont}(\vec\rho) \cdot \boldone_{ \{(\{e_1\} \cup \ldots \cup \{e_{d(e)}\})_{e \in E(H)}\}  }(\mathsf{ep}(\vec\rho))\\
&\overset{\eqref{eqn:char_function_moebius} }{=} \sum_{\sigma \in \mathsf{Part}(d)}  \left(\prod_{e \in E(H)} \mu(\{e_1\} \cup \ldots \cup \{e_{d(e)}\}, \sigma(e)) \right) \cdot \sum_{\vec\rho \in \mathcal{F}(H(d))} \mathsf{Cont}(\vec\rho) \cdot \boldone_{\mathsf{Part}(d)^{\geq \sigma}}(\mathsf{ep}(\vec\rho))\\
&= \sum_{\sigma\in \mathsf{Part}(d)}  \left(\prod_{e \in E(H)} \mu(\{e_1\} \cup \ldots \cup \{e_{d(e)}\}, \sigma(e)) \right) \cdot \sum_{\vec\rho \in \mathsf{ep}^{-1}(\mathsf{Part}(d)^{\geq \sigma})} \mathsf{Cont}(\vec\rho)\,. \label{eqn:Contxi_sum}
\end{align}
To analyze this sum further, note that $\mathsf{ep}^{-1}(\mathsf{Part}(d)^{\geq \sigma)}$ is the set of fractures $\vec\rho$ on $H(d)$ such that any two edges $e, e' \in E(H(d))$ in the same block of $\sigma$ share the same end-vertices inside the fractured graph $\fracture{H(d)}{\vec\rho}$. 

But such fractures can equivalently be described as follows: take the graph $H(d(e))$, combine all edges from blocks $B_i^e$ of $\sigma$ into one edge $E_{e,i}$, which gets weighted by the cardinality $\mathsf{wt}(E_{e,i})=|B_i^e|$, to obtain a graph $H(d(e))/\sigma$ with weighted multi-edges. We extend the weight function $\mathsf{wt}$ to the set of half-edges of $H(d(e))/\sigma$ by sending each half-edge to the weight of the edge of which it forms a part.
Then there is a one-to-one correspondence
\begin{equation} \label{eqn:chi_chibar_correspondence}
\vec\rho \in\mathsf{ep}^{-1}(\mathsf{Part}(d)^{\geq \sigma})  \leftrightarrow \vec\rho' \in \prod_{v \in V(H)} \mathsf{Part}(\he(H(d)/\sigma, v))
\end{equation}
to the set of fractures on the graph $H(d(e))/\sigma$. One important formula that is straightforward to check from the definitions is that under the correspondence \eqref{eqn:chi_chibar_correspondence}, the contribution $\mathsf{Cont}(\vec\rho)$ is equal to the formula
\begin{equation}
\mathsf{Cont}(\vec\rho') := \prod_{v\in V(H)}\left(\prod_{A \in \vec\rho'(v)} s(\sum_{h \in A} \mathsf{wt}(h) ) \cdot (-1)^{|\vec\rho'(v)| -1} \cdot (|\vec\rho'(v)|-1)! \right)
\end{equation}

To make the connection to the formula \eqref{eqn:zeta_d}, note that from the data of $\sigma$ above, we obtain a map $\lambda_\sigma: E(H) \to \mathcal{P}$ as in \eqref{eqn:zeta_d}, where the partition $\lambda_\sigma(e)$ is defined by
\[
\lambda_\sigma(e) = |B_1^e| + \ldots + |B_{m_e}^e| \in \mathcal{P}(d(e))\,.
\]
\textbf{Claim:} For $\lambda: E(H) \to \mathcal{P}$ with $|\lambda(e)|=d(e)$, the sum of contributions in \eqref{eqn:Contxi_sum} from those $\sigma$ with $\lambda_\sigma = \lambda$ equals
the term
\begin{equation} \label{eqn:final_desired_form}
\prod_{e \in E(H)} \mathsf{mult}(\lambda(e)) \cdot   \prod_{v \in V(H)} \chi(\mathsf{deg}(H, v, \lambda), s)
\end{equation}
appearing in \eqref{eqn:Theorem_general1}.

If we prove this claim, then multiplying with the common factor
\[
\frac{1}{\#\mathsf{Aut}(H) \cdot \prod_{e \in E(H)} |\lambda(e)|!}
\]
appearing in both expressions \eqref{eqn:zeta_d_intermediate_form} and \eqref{eqn:Theorem_general1}, and summing over all $\lambda$ yields the statement of the theorem.

\noindent {\bf Proof of Claim:} 
Fixing $\sigma$ with $\lambda_\sigma=\lambda$, we can now use the distributivity trick from before and obtain
\begin{align}
\sum_{\vec\rho' \in \prod_{v \in V(H)} \mathsf{Part}(\he(H(d)/\sigma, v))} \mathsf{Cont}(\vec\rho') = \prod_{v \in V(H)} \underbrace{\sum_{\vec\rho'(v) \in \mathsf{Part}(\he(H(d)/\sigma, v))} \left(\prod_{A \in \vec\rho'(v)} s(\sum_{h \in A} \mathsf{wt}(h) ) \cdot (-1)^{|\vec\rho'(v)| -1} \cdot (|\vec\rho'(v)|-1)! \right)}_{=\chi(\mathsf{deg}(H, v, \lambda), s)}
\end{align}
The identification with $\chi(\mathsf{deg}(H, v, \lambda), s)$ just uses the definition of the function $\chi$ from \eqref{eqn:chi_generalization} and an identification of the half-edges $h \in \he(H(d)/\sigma, v)$ with the summands $\lambda(e)_i$ in the partition $\mathsf{deg}(H, v, \lambda)$ which respects the weighting, i.e. satisfies $\mathsf{wt}(h)=\lambda(e)_i$. This basically comes from the definition of the function $\mathsf{wt}$ on one side and the partition-valued function $\lambda$ on the other side.

Putting the final pieces together, we have that the sum of contributions in \eqref{eqn:Contxi_sum} from those $\sigma$ with $\lambda_\sigma = \lambda$ equals
\begin{align}
\sum_{\substack{\sigma = (\sigma_e)_e \in \prod_{e \in E(H)} \mathsf{Part}(d(e)) \\ \lambda_\sigma = \lambda } }  \left(\prod_{e \in E(H)} \mu(\{e_1\} \cup \ldots \cup \{e_{d(e)}\}, \sigma_e) \right) \cdot \underbrace{\prod_{v \in V(H)} \chi(\mathsf{deg}(H, v, \lambda), s)}_{(*)}
\end{align}
In this expression, we see that the term $(*)$ is independent of $\sigma$ and can be drawn out. On the other hand, the condition $\lambda_\sigma=\lambda$ can be checked independently for each edge $e \in E(H)$ and there is just equivalent to requiring that $\mathsf{shape}(\sigma_e)=\lambda(e)$. Then we can use distributivity a final time and obtain the desired expression
\begin{equation}
    \prod_{v \in V(H)} \chi(\mathsf{deg}(H, v, \lambda), s) \cdot \prod_{e \in E(H)} \underbrace{\sum_{\substack{\sigma_e\in \mathsf{Part}(d(e)) \\ \mathsf{shape}(\sigma_e) = \lambda(e) } }  \mu(\{e_1\} \cup \ldots \cup \{e_{d(e)}\}, \sigma_e)}_{=\mathsf{mult}(\lambda(e))}
\end{equation}
from \eqref{eqn:final_desired_form}. This shows the claim and thus finishes the entire proof.
\end{proof}

\begin{theorem}\label{thm:general_coefficient}
Let $(H, \nu_H)$ be a vertex-coloured graph with no isolated vertices and at most $k \geq 0$ edges. Its vertex colouring is given by $\nu_H : V(H) \to \mathcal{S}$, where $\mathcal{S}$ is a finite set of signatures $s$ such that $s(0) = 1$. Then we have
\begin{equation} \label{eqn:Theorem_general}
    \zeta_{\mathcal{S},k}(H,\nu_H)=\frac{1}{\#\mathsf{Aut}(H, \nu_H)} \cdot \sum_{\substack{\lambda: E(H) \to \mathcal{P}\\\sum_e |\lambda(e)| = k }}\ \  \prod_{e \in E(H)} \frac{\mathsf{mult}(\lambda(e))}{|\lambda(e)|!}  \prod_{v \in V(H)} \chi(\mathsf{deg}(H, v, \lambda), \nu_H(v))\,,
\end{equation}
where the partitions $\mathsf{deg}(H, v, \lambda)$ are again defined as
\[
\mathsf{deg}(H, v, \lambda) = \bigcup_{\substack{e \in E(H):\\ e \text{ incident to }v}} \lambda(e)\,. 
\]    
\end{theorem}
\begin{proof}
For the adaptation of \Cref{thm:zeta_at_most_k_edges} to the setting with multiple signatures, it suffices to derive analogous equations to \eqref{eqn:zeta_intermediate_form} and \eqref{eqn:zeta_d_intermediate_form}. The rest of the argument then follows verbatim the proof of \Cref{thm:zeta_at_most_k_edges}, with the additional observation that all auxiliary graphs we define come with a very natural vertex-colouring. 

For the analogous of \eqref{eqn:zeta_intermediate_form}, similar arguments employed in \Cref{thm:k_edges_multiple_signatures} give 
\begin{align}\label{eqn:zeta_intermediate_coloured}
\zeta_{\mathcal{S},k}(H,\nu_H) = \#\mathsf{Aut}(H, \nu_H)^{-1}  \cdot \sum_{(F, \nu_F)\in \mathcal{G}_k(\mathcal{S})} &\left(\#\mathsf{Aut}(F, \nu_F)^{-1}\cdot   \prod_{u\in V(F)}\nu_F(u)(d_F(u))\right) \nonumber \\ & \cdot \sum_{\phi \in M_{(F,\nu_F),(H,\nu_H)}}\left( (-1)^{|V(F)|-|V(H)|}\cdot \prod_{B \in R(\phi)} (|B|-1)!\right) \!\!\,,
\end{align}
where now $M_{(F, \nu_F), (H, \nu_F)}$ is the set of all surjective (on vertices and edges) homomorphisms $\phi : (F, \nu_F) \to (H, \nu_H)$ that preserve colours.

Following the proof of \Cref{thm:zeta_at_most_k_edges}, given $\phi \in M_{(F, \nu_F), (H, \nu_H)}$ we obtain a map $d_\phi: E(H) \to \mathbb{Z}_{>0}$ satisfying $\sum_e d_\phi(e) = k$ by
\begin{equation}
    d_\phi(e) = \# \{\widehat e \in E(F) : \phi_E(\widehat e) = e\}\,,
\end{equation}
where $\phi_E : E(F) \to E(H)$ is the map on edges induced by the graph homomorphism $\phi$. Given any map $d : E(H) \to \mathbb{Z}_{>0}$ with $\sum_e d(e)=k$, we define
\begin{equation} \label{eqn:zeta_d_multiple_signatures}
\zeta_{\mathcal{S}, k, d}(H, \nu_H) = \frac{1}{\#\mathsf{Aut}(H, \nu_H)} \cdot \sum_{\substack{\lambda: E(H) \to \mathcal{P}\\|\lambda(e)| = d(e) }} \ \  \prod_{e \in E(H)} \frac{\mathsf{mult}(\lambda(e))}{|\lambda(e)|!}  \prod_{v \in V(H)} \chi(\mathsf{deg}(H, v, \lambda), \nu_H(v))\,,\end{equation}
so that $\zeta_{\mathcal{S}, k}(H,\nu_H) = \sum_{d} \zeta_{\mathcal{S}, k, d}(H, \nu_H)$. Again, it suffices to show that $\zeta_{\mathcal{S}, k, d}(H, \nu_H)$ computes the sum of all contributions $((F, \nu_F), \phi)$ of the summation \eqref{eqn:zeta_intermediate_coloured} with $d_\phi = d$, which would finish the proof.

To this end, recall that $H(d)$ is the multigraph obtained by $H$ by replacing each edge $e$ of $H$ with $d(e)$ copies of the edge, that comes with a natural colouring $\nu_{H(d)}$. Let $M_{F, H}(d) = \{\phi \in M_{F, H} : d_{\phi} = d\}$ and $M_{(F, \nu_F), (H, \nu_H)}(d) = \{\phi \in M_{(F, \nu_F), (H, \nu_H)} : d_{\phi} = d\}$. We also define $\mathcal{M}(d) = \dot\bigcup_{F \in \mathcal{G}_k}M_{F, H}(d)$ and $\mathcal{C}(d) = \dot\bigcup_{(F, \nu_F) \in \mathcal{G}_{k}(\mathcal{S})}M_{(F, \nu_F), (H, \nu_H)}$. By the proof of \Cref{thm:zeta_at_most_k_edges} we can define a map $S : \mathcal{M}(d) \to \mathcal{F}^{\textup{adm}}(H(d))$ that is surjective and the fibers of which coincide with the orbits of $M_{F, H}(d)$ under the free group actions of $\mathsf{Aut}(F, \nu_F)$. Further observe, that there is a one-to-one correspondence between $\mathcal{M}(d)$ and $\mathcal{C}(d)$, so we can, by slightly abusing notation, define $S : \mathcal{C} \to \mathcal{F}^{\textup{adm}}(H(d))$ given by $S((F, \nu_F), \phi) = S(F, \phi)$. Again, by \Cref{thm:k_edges_multiple_signatures}, $S$ preserves the colours in the sense that if $S((F, \nu_F), \phi) = \vec\rho \in \mathcal{F}^{\textup{adm}}(H(d))$, then $(F, \nu_F) \cong (\fracture{H(d)}{\vec\rho}, \nu_{\vec\rho})$. Furthermore, if $S((F, \nu_F), \phi) = S((F', \nu_{F'}), \phi')$ then $(F, \nu_F) \cong (F', \nu_{F'})$ and $\phi' = \phi \circ \widetilde{\eta}$, for $\widetilde{\eta} \in \mathsf{Aut}(F, \nu_F)$. We have,
\[\zeta_{\mathcal{S}, k, d}(H, \nu_H) = \#\mathsf{Aut}(H, \nu_H)^{-1}\cdot\!\!\!\!\!\!\sum_{((F, \nu_F), \phi) \in \mathcal{M}(d)}\!\!\!\!\!\!\#\mathsf{Aut}(F, \nu_F)^{-1}\cdot\underbrace{\prod_{v\in V(H)}\left(\prod_{B \in \vec\rho(v)} \nu_H(v)(|B|) \cdot (-1)^{|\vec\rho(v)| -1} \cdot (|\vec\rho(v)|-1)! \right)}_{=:\mathsf{Cont}(\vec\rho)}\,,\]
where $\vec\rho = S((F, \nu_F), \phi)$.
Thus, we obtain the analogous expression to \eqref{eqn:zeta_d_intermediate_form}, given by
\[\zeta_{\mathcal{S}, k, d}(H, \nu_H) = \frac1{\#\mathsf{Aut}(H, \nu_H)\cdot\prod_{e \in E(H)}d(e)\,!}\cdot\sum_{\xi \in \mathcal{F}^{\textup{adm}}(H(d))}\mathsf{Cont}(\vec\rho)\cdot\underbrace{\sum_{((F, \nu_F), \phi) \in S^{-1}(\xi)}\#\mathsf{Aut}(F, \nu_F)^{-1}}_{= 1}.\]
\end{proof}

\begin{lemma}\label{lem:linear_type_sigs_equivalence}

For a signature $s$ (with $s(0)=1$), the following are equivalent:
\begin{enumerate}
    \item[a)] $\{s\}$ is of type $\mathbb{T}[\mathsf{lin}]$
    \item[b)] $s(n) = s(1)^n$ for all $n\in \mathbb{N}_{>0}$
    \item[c)] $\chi(\lambda, s)=0$ for all partitions $\lambda$ with $\mathsf{len}(\lambda) \geq 2$
\end{enumerate}
\end{lemma}
\begin{proof}
We prove the implications separately:
\begin{itemize}
    \item a) $\implies$ b): This holds by induction: The claim is trivial for $n=1$. Moreover, for $n>1$, using that $\{s\}$ is of type $\mathbb{T}[\mathsf{lin}]$, we have
    \[0 = \chi(n,s)= \sum_{\sigma} (-1)^{|\sigma|-1}\cdot (|\sigma|-1)! \cdot \prod_{B\in \sigma}s(|B|)\,,\]
    where the sum is over all partitions of $[n]$. Next, setting $\top_n=\{[n]\}$ as the coarsest partition, and using the induction hypothesis, we obtain
    \[ s(n)= - \sum_{\sigma \neq \top_n} (-1)^{|\sigma|-1}\cdot (|\sigma|-1)! \cdot \prod_{B\in \sigma}s(1)^{|B|} = s(1)^{n}\cdot (- \sum_{\sigma \neq \top_n} (-1)^{|\sigma|-1}\cdot (|\sigma|-1)!) \,.\]
    Now recall that taking the sum over all evaluations of the M\"obius function over the partitions of $[n]$, for $n>1$, yields $0$ (see~\cite[Chapter 3]{Stanley11}), that is,
    \[\sum_{\sigma} (-1)^{|\sigma|-1}\cdot (|\sigma|-1)! = 0\,.\]
    Finally, observe that $(-1)^{|\top_n|-1}\cdot (|\top_n|-1)! =1$, and thus
    $(- \sum_{\sigma \neq \top_n} (-1)^{|\sigma|-1}\cdot (|\sigma|-1)!) = 1$.
    \item  b) $\implies$ c): Set $\ell=\mathsf{len}(\lambda)\geq 2$. From b) it follows that $\prod_{B \in \sigma} s(\sum_{i \in B} \lambda_i) = s(1)^{|\lambda|}$. Thus
    \[\chi(\lambda , s) = \sum_{\sigma\in \mathsf{Part}([\ell])} (-1)^{|\sigma|-1} (|\sigma|-1)! \cdot \prod_{B \in \sigma} s(\sum_{i \in B} \lambda_i) = s(1)^{|\lambda|} \cdot  \sum_{\sigma\in \mathsf{Part}([\ell])} (-1)^{|\sigma|-1} (|\sigma|-1)! = 0\,.\]
    \item c) $\implies$ a): We have to show that $\chi(d,s)=0$ for all $d\geq 2$. Set $\lambda = \underbrace{ 1+\dots+1}_{d\text{ times}}$ and observe \[\chi(d,s)=\chi(\lambda,s)=0\,.\] 
\end{itemize}
\end{proof}

For the following result, recall the simplification of the coefficient function $\zeta$ for the case of a unique signature $\mathcal{S}=\{s\}$ given in Remark~\ref{rem:simplified_single_signature}.
\begin{theorem} \label{Thm:Conj_2_false}
Let $\{s\}$ be of type $\mathbb{T}[\omega]$ and $d \geq 2$. Then for the complete graph $K_{d+1}$ on $d+1$ vertices, we have
\begin{equation} \label{eqn:zeta_complete_graph}
\zeta(K_{d+1}, d^2-1) = (-1)^{(d+1)(d-2)/2} \cdot A \cdot (s(2)-s(1)^2)^{(d+1)\cdot(d-1)} \text{ for some } A \in \mathbb{Q}_{>0}\,.
\end{equation}
In particular, $\zeta(K_{d+1}, d^2-1) \neq 0$ whenever $\{s\}$ is  of type $\mathbb{T}[\mathsf{\omega}]$. \qed
\end{theorem}

To prove the theorem, we need some more information about the function $\chi$. The first is the following recursive formula.

\begin{lemma} \label{Lem:chi_recursion}
Let $\lambda = a + \lambda_1 + \ldots + \lambda_m$ be a partition of length $m+1$ for $a \geq 2$ and consider a splitting $a  = a_1 + a_2$ with $a_1, a_2 \geq 1$. Then we have
\begin{equation}  \label{eqn:chi_recursion}
\chi(\lambda, s) = \chi(a_1 + a_2 + \lambda_1 + \ldots + \lambda_m, s) + \sum_{S \subseteq [m]} \chi(a_1 \cup \lambda_S, s) \cdot \chi(a_2 \cup \lambda_{S^c}, s)\,,
\end{equation}
where $S^c = [m] \setminus S$ is the complement of $S$ in $[m]=\{1, \ldots, m\}$ and
\[
\lambda_S = \sum_{i \in S} \lambda_i \text{ and } \lambda_{S^c} = \sum_{j \in S^c} \lambda_j\,.
\]
\end{lemma}
\begin{proof}
Let $\mu = a_1 + a_2 + \lambda_1 + \ldots + \lambda_m$, then by definition we have
\begin{equation} \label{eqn:chi_mu}
\chi(\mu, s) = \sum_{\sigma \in \mathsf{Part}(m+2)} (-1)^{|\sigma|-1} (|\sigma|-1)! \cdot \prod_{B \in \sigma} s(\sum_{i \in B} \mu_i)\,.
\end{equation}
For the partitions $\sigma$ there are two possibilities:
\begin{itemize}
    \item If the indices $1,2$ associated to the summands $a_1, a_2$ are in the \emph{same} block $B$ of $\sigma$, then $\sigma$ can equivalently be described via a set partition $\overline \sigma$ for the integer partition $a + \lambda_1 + \ldots + \lambda_m=\lambda$, where the indices for $a_1, a_2$ are combined in a single index for the summand $a$. The sum over all such $\sigma$ in \eqref{eqn:chi_mu} yields $\chi(\lambda, s)$.
    \item For all other $\sigma$, the indices $1,2$ are in different blocks of the partition $\sigma$. To finish the argument, we observe that we can interpret the expression \eqref{eqn:chi_mu} as the summation over \emph{ordered} set partitions $\vec \sigma = (B_1, \ldots, B_{|\sigma|})$ satisfying $1 \in B_1$, weighted by 
    \begin{equation} \label{eqn:ordered_set_part_weight}
        (-1)^{|\sigma|-1} \prod_{j=1}^{|\sigma|} s(\sum_{i \in B_j} \mu_i)\,.
    \end{equation}
    Indeed, for an unordered $\sigma$ there is precisely one block $B \in \sigma$ with $1 \in B$, which becomes $B_1$, and there are $(|\sigma|-1)!$ possibilities of ordering the other blocks. Given $\vec \sigma$, let $i_0 \geq 2$ be the unique index such that $2 \in B_{i_0}$ and let $S = (\bigcup_{i=1}^{i_0-1} B_i) \setminus \{1\}$ and $S^c = (\bigcup_{i=i_0}^{|\sigma|} B_i) \setminus \{2\}$. Then $\vec \sigma$ decomposes uniquely into
    \[
    \vec \sigma = \underbrace{(B_1, \ldots, B_{i_0-1})}_{=\vec \sigma_S} \cup \underbrace{(B_{i_0}, \ldots, B_{|\sigma|})}_{=\vec \sigma_{S^c}}\,.
    \]
    The ordered set partitions $\vec \sigma_S, \vec \sigma_{S^c}$ are exactly terms appearing in the corresponding formulas for $\chi(a_1 \cup \lambda_S, s)$ and $\chi(a_2 \cup \lambda_{S^c}, s)$. It's straightforward to check that the product of their corresponding weights \eqref{eqn:ordered_set_part_weight} gives $(-1)$ times the weight \eqref{eqn:ordered_set_part_weight}, since $(|\vec \sigma_S|-1)+(|\vec \sigma_{s^c}|-1)=(|\vec \sigma|-1)-1$. Thus summing over these remaining $\sigma$ gives the missing term in \eqref{eqn:chi_recursion}. \qedhere
\end{itemize}
\end{proof}

\begin{proposition} \label{Prop:chi_vanishing_property}
Let $\{s\}$ be of type $\mathbb{T}[\omega]$, then
we have
\begin{enumerate}
    \item[a)] $\chi(\lambda, s) = 0$ if $|\lambda| < 2\cdot  \mathsf{len}(\lambda)-2$,
    \item[b)] $\chi(\lambda, s) = a_\lambda \cdot (s(2)-s(1)^2)^{\mathsf{len}(\lambda)-1}$ for some $a_\lambda \in \mathbb{Z}_{>0}$ if $|\lambda| = 2 \cdot \mathsf{len}(\lambda)-2$.
\end{enumerate}
\end{proposition}
\begin{remark}
Computer experiments suggest that the formula of $a_\lambda$ is
\begin{equation}
a_\lambda = (d-2)! \cdot \prod_{\lambda_i \in \lambda} \lambda_i\,.
\end{equation}
We would be interested in seeing a proof of this, but content ourselves in proving the weaker statement (that $a_\lambda$ is a positive integer) below.
\end{remark}
\begin{proof}[Proof of Proposition \ref{Prop:chi_vanishing_property}]
a) We prove the statement by a double-induction: the outer induction is on increasing $d=|\lambda|$, whereas the inner induction is on decreasing $\ell = \mathsf{len}(\lambda)$ starting with the maximal case of $\lambda = 1 + 1 + \ldots + 1 = 1^d$.

The cases $d=0,1,2$ are empty, since any partition satisfies $|\lambda|\geq \mathsf{len}(\lambda)$. For $d=3$ the only case is $\lambda=1+1+1$ where indeed $\chi(1+1+1,s)=0$ since $s$ is assumed to be of type $\mathbb{T}[\omega]$. This concludes the induction start on $d$.
Thus assume the statement is proven for $|\lambda|=0,1, \ldots, d-1$. The downward-induction start on $\mathsf{len}(\lambda)$ is given by observing $\chi(1+1+ \ldots + 1,s)=\chi(1^d, s)=0$, again since $s$ is of type $\mathbb{T}[\omega]$. 

Hence let $\lambda$ be a partition with $|\lambda|=d$ and $\mathsf{len}(\lambda)=\ell<d$, satisfying $d - 2 \ell +2<0$, and assume the claim is proven for all partitions of degree less than $d$, or equal to $d$ but length greater than $\ell$. We can write $\lambda = a + \lambda_1 + \ldots + \lambda_{\ell-1}$ with $a=a_1+a_2 \geq 2$, since $\mathsf{len}(\lambda)<d$ forces $\lambda$ to have at least one part greater than $1$. Then by Lemma \ref{Lem:chi_recursion} we know
\begin{equation}  \label{eqn:chi_recursion_recall}
\chi(\lambda, s) = \underbrace{\chi(a_1 + a_2 + \lambda_1 + \ldots + \lambda_{\ell-1}, s)}_{=0} + \sum_{S \subseteq [\ell-1]} \chi(a_1 \cup \lambda_S, s) \cdot \chi(a_2 \cup \lambda_{S^c}, s)\,,
\end{equation}
where the first vanishing follows by induction since 
\[
|a_1+a_2+\lambda_1+\ldots+\lambda_{\ell-1}|=d \text{ and }\mathsf{len}(a_1+a_2+\lambda_1+\ldots+\lambda_{\ell-1})=\ell+1\,.
\]
On the other hand, for any $S \subseteq [\ell-1]$ we have $|a_1 \cup \lambda_S|, |a_2 \cup \lambda_{S^c}|<d$ and
\begin{equation} \label{eqn:Delta_S_Delta_Sc}
\underbrace{|a_1 \cup \lambda_S| - 2 \cdot \mathsf{len}(a_1 \cup \lambda_S)+2}_{=:\Delta_S} + \underbrace{|a_2 \cup \lambda_{S}^c| - 2\cdot \mathsf{len}(a_2 \cup \lambda_{S^c})+2}_{=:\Delta_{S^c}} = |\lambda| - 2 \cdot \mathsf{len}(\lambda)-2+4\,,
\end{equation}
which is a negative number by the assumption $d - 2 \ell +2<0$. Thus one of the two summands $\Delta_S, \Delta_{s^c}$ must be negative, and hence the corresponding term $\chi(a_1 \cup \lambda_S, s)$ or $\chi(a_2 \cup \lambda_{S^c}, s)$ vanishes by induction. This shows the vanishing of each summand in \eqref{eqn:chi_recursion_recall} and finishes the proof.

\noindent b) We prove the statement by increasing induction on $d=|\lambda|$, with the base case $\lambda=1+1$ indeed satisfying $\chi(1+1,s)=1 \cdot (s(2)-s(1)^2)$. 

For the induction step, let $\lambda$ be a partition with $|\lambda| \geq 3$ satisfying $|\lambda|-2\cdot \mathsf{len}(\lambda)+2 = 0$. Again we can write it as $\lambda = a+\lambda_1 + \ldots + \lambda_{\ell-1}$ where we decompose $a=a_1+a_2$ for the particular choice $a_1=1$, $a_2=a-1$. Applying Lemma \ref{Lem:chi_recursion}, we obtain exactly the same equation \eqref{eqn:chi_recursion_recall}. Note that the indicated vanishing there still holds by the same argument as before and the proven case a), since $d-2(\ell+1)+2<0$. 
Repeating also the second part of the argument yields the equation \eqref{eqn:Delta_S_Delta_Sc}, which vanishes by assumption. Thus for each $S$ there are two cases
\begin{itemize}
    \item If one of the numbers $\Delta_S, \Delta_{S^c}$ is negative, we have that the corresponding term $\chi(a_1 \cup \lambda_S, s)$ or $\chi(a_2 \cup \lambda_{S^c}, s)$ vanishes by part a).
    \item If both numbers $\Delta_S, \Delta_{S^c}$ are non-negative, then they must be equal to zero since they sum up to zero. In that case, by induction we have
    \[
    \chi(a_1 \cup \lambda_S, s) = a_{a_1 \cup \lambda_S} \cdot (s(2)-s(1)^2)^{|S|} \text{ and }\chi(a_2 \cup \lambda_{S^c}, s) = a_{a_2 \cup \lambda_{S^c}} \cdot (s(2)-s(1)^2)^{|S^c|}
    \]
    for positive integers $a_{a_1 \cup \lambda_S}$ and $a_{a_2 \cup \lambda_{S^c}}$.
\end{itemize}
Since $|S|+|S^c|=\mathsf{len}(\lambda)-1$, we see that all non-zero terms in \eqref{eqn:chi_recursion_recall} are positive multiples of $(s(2)-s(1)^2)^{\mathsf{len}(\lambda)-1}$. The only missing step is showing that there is at least one such term (satisfying $\Delta_S=\Delta_{S^c}=0$). For this we just note that the condition $|\lambda|=2 \cdot \mathsf{len}(\lambda)-2$ means that $\lambda$ must have some parts smaller than $2$ (i.e. equal to $1$), e.g. $\lambda_{\ell-1}=1$. Then we choose $S=\{\ell-1\}$, so that $a_1 \cup \lambda_S = 1+1$ which satisfies
\[
\Delta_S = |1+1|-2 \cdot \mathsf{len}(1+1) +2 = 0\,.
\]
This forces $\Delta_{S^c}=0$ by the vanishing of \eqref{eqn:Delta_S_Delta_Sc} and hence there is at least one surviving term in \eqref{eqn:chi_recursion_recall}, finishing the proof.
\end{proof}

\begin{proof}[Proof of Theorem \ref{Thm:Conj_2_false}]
Recall the formula \eqref{eqn:Theorem_general} :
\begin{equation} \zeta_{\mathcal{S},k}(H,\nu_H)=\frac{1}{\#\mathsf{Aut}(H, \nu_H)} \cdot \sum_{\substack{\lambda: E(H) \to \mathcal{P}\\\sum_e |\lambda(e)| = k }}\ \  \prod_{e \in E(H)} \frac{\mathsf{mult}(\lambda(e))}{|\lambda(e)|!}  \prod_{v \in V(H)} \chi(\mathsf{deg}(H, v, \lambda), \nu_H(v))\,,
\end{equation}
where the partitions $\mathsf{deg}(H, v, \lambda)$ are defined as
\[
\mathsf{deg}(H, v, \lambda) = \bigcup_{\substack{e \in E(H):\\ e \text{ incident to }v}} \lambda(e)\,. 
\]    

Let us apply this formula for $k=d^2-1$, $\mathcal{S}=\{s\}$, $H=K_{d+1}$, and the colouring $\nu_H$ mapping all vertices of $K_{d+1}$ to $s$. Any $\lambda$ in this formula induces the partition $\mathsf{deg}(K_{d+1}, v, \lambda)$ at the vertices $v$ of $K_{d+1}$. Since the total sum of weights of the $\lambda(e)$ is $d^2-1$, we see that, by a weighted version of the Handshaking-Lemma,
\begin{equation}
\sum_{v \in V(K_{d+1})} |\mathsf{deg}(K_{d+1}, v, \lambda)| = 2 \cdot \sum_{e \in E(K_{d+1})} |\lambda(e)| = 2(d^2-1)\,.
\end{equation}
Dividing by the number $d+1$ of vertices, we see that the \emph{average} weight $|\mathsf{deg}(K_{d+1}, v, \lambda)|$ equals $2d-2$. If any vertex $v_0$ had a \emph{strictly higher} weight than that, then for balancing reasons there would have to be another vertex $v$ with $|\mathsf{deg}(K_{d+1}, v, \lambda)|<2d-2$. Then Proposition \ref{Prop:chi_vanishing_property} a) implies that the corresponding factor $\chi(\mathsf{deg}(K_{d+1}, v, \lambda),s)$ of  \eqref{eqn:Theorem_general} vanishes since
\[
|\mathsf{deg}(K_{d+1}, v, \lambda)| < 2 \cdot \underbrace{\mathsf{len}(\mathsf{deg}(K_{d+1}, v, \lambda))}_{\geq d} - 2\,.
\]
Thus in the summation \eqref{eqn:Theorem_general} we can restrict to those $\lambda$ satisfying $|\mathsf{deg}(K_{d+1}, v, \lambda)|=2d-2$ for all $v \in V(K_{d+1})$. Applying Proposition \ref{Prop:chi_vanishing_property} a) again, any contribution from such a $\lambda$ for which $\mathsf{len}(\mathsf{deg}(K_{d+1}, v, \lambda))>d$ at any vertex $v$ will also vanish. Since
\[
\mathsf{len}(\mathsf{deg}(K_{d+1}, v, \lambda)) = \sum_{\substack{e \in E(K_{d+1}):\\ e \text{ incident to }v}} \mathsf{len}(\lambda(e))\,,
\]
and each of the $d$ summands in this sum is a positive integer, we have that each partition $\lambda(e)$ must be of length exactly $1$, i.e. $\lambda(e)=(d(e))$ for some $d(e) \in \mathbb{Z}_{\geq 1}$ satisfying 
$$\sum_{e \in E(K_{d+1})} d(e) = d^2 -1\,.$$
For each such degree distribution $d$ (and induced map $\lambda$), we have
\[
\mathsf{sign}\left(\prod_{e \in E(K_{d+1})} \mathsf{mult}(\lambda(e)) \right) = \prod_{e \in E(K_{d+1})} (-1)^{d(e)-1} = (-1)^{d^2-1 - d(d+1)/2} = (-1)^{(d+1)(d-2)/2}
\]
Applying Proposition \ref{Prop:chi_vanishing_property} b) to the formula 
\eqref{eqn:Theorem_general}, we thus see that any nonzero summand there is precisely of the form \eqref{eqn:zeta_complete_graph} claimed in our theorem. 

The only remaining thing to prove is that there \emph{exists} at least one degree distribution $d(e)$ satisfying that the sum of degrees at each vertex is precisely $2d-2$. This is easy: choose a $d-2$ regular subgraph $G$ of $K_{d+1}$ and declare $d(e)=2$ for $e \in E(G)$ and $d(e)=1$ otherwise. The existence of such a subgraph follows since $(d-2)(d+1)$ is even,  e.g.\ by using a special case of the Erdős–Gallai Theorem~\cite{ErdosG60,Choudum86}.
\end{proof}

We are finally able to prove Lemma~\ref{lem:main_result_on_zeta}, which we restate for convenience.
\begin{lemma}[Lemma~\ref{lem:main_result_on_zeta}, restated]
    Let $\mathcal{S}$ be a finite set of signatures and let $\mathcal{C}$ be the class of all graphs $H$ for which there is a positive integer $k$ and a colouring $\nu:V(H)\to \mathcal{S}$ such that $\zeta_{\mathcal{S},k}(H,\nu)\neq 0$. If $\mathcal{S}$ is of type $\mathbb{T}[\mathsf{lin}]$, then all graphs in $\mathcal{C}$ are acyclic. Otherwise $\mathcal{C}$ has unbounded treewidth.
\end{lemma}
\begin{proof}
    To avoid any confusion, we recall that, by Remark~\ref{rem:s_zero_equals_1}, we can assume w.l.o.g.\ that $s(0)=1$ for all $s\in \mathcal{S}$. 

    Assume first that $\mathcal{S}$ is of type $\mathbb{T}[\mathsf{lin}]$; in particular, this implies that $\{s\}$ is of type $\mathbb{T}[\mathsf{lin}]$ for all $s\in \mathcal{S}$. By Lemma~\ref{lem:linear_type_sigs_equivalence}, it follows that $\chi(\lambda,s)=0$ whenever $\mathsf{len}(\lambda)\geq 2$ and $s\in \mathcal{S}$. Assume for contradiction that there is a graph $H \in \mathcal{C}$ containing a cycle. Then $H$ contains a vertex $u$ of degree at least $2$. Since $H\in \mathcal{C}$, there are $k$ and $\nu$ such that $\zeta_{\mathcal{S},k}(H,\nu)\neq 0$. Note that $k\geq |E(H)|$, since otherwise $\zeta_{\mathcal{S},k}(H,\nu)$ is trivially zero (see Lemma~\ref{lem:uncoloured_hombasis_general}; no graph $F$ with less than $|E(H)|$ edges can have $H$ as a quotient). Then, by Theorem~\ref{thm:general_coefficient}, we have 
    \[ \zeta_{\mathcal{S},k}(H,\nu_H)=\frac{1}{\#\mathsf{Aut}(H, \nu_H)} \cdot \sum_{\substack{\lambda: E(H) \to \mathcal{P}\\\sum_e |\lambda(e)| = k }}\ \  \prod_{e \in E(H)} \frac{\mathsf{mult}(\lambda(e))}{|\lambda(e)|!}  \prod_{v \in V(H)} \chi(\mathsf{deg}(H, v, \lambda), \nu_H(v))\,,\]
    where
    \[
\mathsf{deg}(H, v, \lambda) = \bigcup_{\substack{e \in E(H):\\ e \text{ incident to }v}} \lambda(e)\,. 
\]
However, note that $\mathsf{len}(\mathsf{deg}(H, u, \lambda))\geq 2$ since $u$ has degree at least $2$ in $H$. Therefore, by our initial assumption, we have that $\chi(\mathsf{deg}(H, u, \lambda), \nu_H(u))=0$, and hence $\zeta_{\mathcal{S},k}(H,\nu_H)=0$ since each summand then contains a factor of $0$. This yields the desired contradiction. 

Next assume that there is $\mathcal{S}$ is not of type $\mathbb{T}[\mathsf{lin}]$. In particular, this means that, for some $s\in \mathcal{S}$, we have that $\{s\}$ is not of type $\mathbb{T}[\mathsf{lin}]$. By the second part of Theorem~\ref{Thm:Conj_2_false} it then follows that arbitrary large cliques survive in $\mathcal{C}$, concluding the proof.
\end{proof}

\subsection{Extensions to Signatures Allowing $s(0) = 0$}
As in \Cref{sec:sig0}, we can lift the restriction on $s(0) \neq 0$ and establish the following dichotomy.

\begin{theorem}
Let $\mathcal{S}$ be a finite set of signatures. Let $\mathcal{S}_0 = \{s \in \mathcal{S} \mid s(0) = 0\}$. 
\begin{enumerate}
\item If $\mathcal{S}\,\backslash\, \mathcal{S}_0$ is of type $\mathbb{T}[\mathsf{Lin}]$, then $\text{\sc{p-UnColHolant}}(\mathcal{S})$ can be solved in FPT-near-linear time.
\item Otherwise $\text{\sc{p-UnColHolant}}(\mathcal{S})$ is $\#\mathrm{W}[1]$-hard. If, additionally, $\mathcal{S}\setminus \mathcal{S}_0$ is of type $\mathbb{T}[\infty]$, then $\text{\sc{p-UnColHolant}}(\mathcal{S})$ cannot be solved in time $f(k)\cdot |\Omega|^{o(k/\log k)}$, unless ETH fails.
\end{enumerate}
\end{theorem}
\begin{proof}
The second claim follows directly from the second part of  \Cref{thm:main_uncol_restate}. So, we may assume that $\mathcal{S}\,\backslash\,\mathcal{S}_0$ is of type $\mathbb{T}[\mathsf{Lin}]$.

The proof of the first claim follows almost verbatim the proof of the respective claim in the colored setting (cf. \Cref{lem:reductionRestrictedHolant}). The only minor subtlety we have to deal with is the fact that we rely on \Cref{lem:linear_type_sigs_equivalence}, which only applies to signatures satisfying $s(0)=1$, hence we have to perform a rescaling. For what follows, we can assume w.l.o.g.\ that $\mathcal{S}$ does not contain the constant-zero signature (if a signature grid contains a vertex equipped with this signature, the Holant value is trivially $0$).

In a nutshell, following the proof of \Cref{lem:reductionRestrictedHolant}, computing $\mathsf{UnColHolant}(\Omega, k)$ reduces to computing $\mathsf{UnColHolant}(\Omega_{\alpha}, k)$, for at most $2k$ appropriately defined auxiliary instances $(\Omega_{\alpha}, k)$ --- we will choose the appropriate $\alpha$ momentarily. Recall that, for $\alpha>0$, we set $\mathcal{S}_\alpha = \{t_\alpha \mid t\in \mathcal{S}_0\}$, where $t_\alpha=t|_{0\mapsto \alpha}$. The vertices in each signature grid $\Omega_{\alpha}$ have been assigned signatures from $(\mathcal{S}\,\backslash\,\mathcal{S}_0)\,\cup\,\mathcal{S}_{\alpha}$, where ${\alpha}$ is defined so that that the following are guaranteed: $(1)$ $(\mathcal{S}\,\backslash\,\mathcal{S}_0) \cap \mathcal{S}_{\alpha} = \emptyset$ and $(2)$ there are at most $2k$ vertices in $V(\Omega_{\alpha})$ with signatures in $\mathcal{S}_{\alpha}$; we denote the latter set by $V_{\alpha}$. Hence, it suffices to show that computing $\mathsf{UnColHolant}(\Omega_{\alpha}, k)$ can be done in FPT-near-linear time. 

To this end, consider one such instance $\Omega_{\alpha}$. As mentioned previously, to be consistent with the analysis of this section, we need to scale all signatures $s \in (\mathcal{S} \setminus \mathcal{S}_0) \cup \mathcal{S}_{\alpha}$ and consider new signatures $s' = s / s(0)$. That way, we also obtain a new signature grid $\Omega_{\alpha}'$ where each signature $s$ has been replaced by $s'$. By \Cref{rem:s_zero_equals_1}, it then suffices to derive the desired result for $\Omega_{\alpha}'$. Recall again that the argument of the proof of \Cref{lem:reductionRestrictedHolant}, which we follow, requires $(\mathcal{S}\setminus\mathcal{S}_0) \cap \mathcal{S}_{\alpha} = \emptyset$. Hence we have to make sure that this condition holds now for the scaled signature sets, that is, for all pairs $s \in \mathcal{S}\setminus\mathcal{S}_0$ and $t \in \mathcal{S}_{0}$, we have to make sure that $s/s(0) \neq t_\alpha/t_\alpha(0)(=t_\alpha/\alpha)$. To this end, using that $t \in \mathcal{S}_0$ and the fact that $\mathcal{S}$ (and thus $\mathcal{S}_0$) do not contain the constant-zero function, there exists a constant $d_t>0$ such that $t_\alpha(d_t)=t(d_t)>0$.

We want to only consider $\alpha$ satisfying that $s(d_t)/s(0) \neq t(d_t)/\alpha$ is guaranteed for all pairs $s, t$. In particular, set
\[
\alpha_{s,t} := \begin{cases}
    s(0)\cdot t(d_{t})/s(d_{t}) & s(d_t)\neq 0\\
    1 & s(d_t) = 0\,,
\end{cases}
\]
and we observe that for any $\alpha\notin\cup_{s,t}\{\alpha_{s,t}\}$ we have that  $s(d_t)/s(0) \neq t(d_t)/\alpha$ for all pairs $s$ and $t$: If $s(d_t)=0$, then $t(d_t)/\alpha \neq 0 = s(d_t)/s(0)$, and if $s(d_t)\neq 0$, then $t(d_t)/\alpha \neq t(d_t)/\alpha_{s,t} = s(d_t)/s(0)$.

Note that $\cup_{s,t}\{\alpha_{s,t}\}$ depends only on $\mathcal{S}$, and not on the problem input. Hence, for what follows, we can assume that all our $2k$ choices of $\alpha$ are not contained in $\cup_{s,t}\{\alpha_{s,t}\}$.

Now, by \Cref{lem:uncoloured_hombasis_general}, we have,
\[ \mathsf{UnColHolant}(\Omega_{\alpha}', k) = \sum_{(H,\nu)\in \mathcal{G}(\mathcal{S}')} \zeta_{\mathcal{S}',k}(H,\nu) \cdot \#\homs{(H,\nu)}{\Omega_{\alpha}'}\,. \]
where $\mathcal{S}' = (\mathcal{S}\,\backslash\,\mathcal{S}_0)' \cup \mathcal{S}_{\alpha}' = \{s_1', \ldots, s_\ell'\}$ are the scaled signatures sets. Recall that the sum runs over vertex coloured graphs $(H, \nu)$ with at most $k$ edges where $\nu : V(H) \to (\mathcal{S}\,\backslash\,\mathcal{S}_0)'\,\cup\,\mathcal{S}_{\alpha}'$.

Let $V_1$ denote the set of vertices $v \in V(H)$ such that $\nu(v) \in (\mathcal{S}\setminus\mathcal{S}_0)'$ and $V_2$ denote the set of vertices $v \in V(H)$ such that $\nu(v) \in \mathcal{S}_{\alpha}'$. Note that, if some $v \in V_1$ has degree at least 2, then $\zeta_{\mathcal{S}',k}(H, \nu)$ is zero (see \Cref{lem:linear_type_sigs_equivalence}, Case (c)). Hence, we may assume that all vertices in $V_1$ have degree 1 and then compute $\#\homs{(H, \nu)}{\Omega_{\alpha}'}$ in FPT-near-linear time as follows. We enumerate all $\phi \in \homs{(H[V_2], \nu|_{V_2})}{\Omega_{\alpha}'}$ via brute-force in $|V_{\alpha}|^{\mathcal{O}(|V_2|)}$ time since vertices in $V_{2}$ can only be mapped to vertices in $V_{\alpha}$. Finally, for each such $\phi$, we compute the number of its extensions to $\#\homs{(H, \nu)}{\Omega_{\alpha}'}$, according to \Cref{lem:partialHomsFPT}, Case 1.
\end{proof}

We obtain, as immediate consequence, the classification of the uncoloured graph factor problem.
\begin{corollary}\label{cor:factor_classification_uncol}
    If $\mathcal{B}$ contains a set $\{0\} \subsetneq S \subsetneq \mathbb{N}$ then $\textsc{Factor}(\mathcal{B})$ is $\#\W[1]$-complete, and cannot be solved in time $f(k)\cdot n^{o(k/\log k)}$ for any function~$f$, unless the Exponential Time Hypothesis fails. Otherwise $\textsc{Factor}(\mathcal{B})$ is solvable in FPT-near-linear time.
\end{corollary}
\begin{proof}
    For the lower bounds, assume there is an $S\in \mathcal{B}$ with $\{0\} \subsetneq S \subsetneq \mathbb{N}$. We define a signature 
    \[s(x) := \begin{cases}
        1 & x \in S\\
        0 & x \notin S\,,
    \end{cases}\]
    and set $\mathcal{S}=\{s\}$. Then, clearly, $\textsc{p-UnColHolant}(\mathcal{S})\fptlinred \textsc{Factor}(\mathcal{B})$. 
    
    We show that $\{s\}$ is of type $\mathbb{T}[\infty]$. To this end, observe that $\chi(3,s)= 2s(1)^3 - 3s(1)s(2) + s(3)$, which is non-zero unless $(s(1),s(2),s(3))\in \{(0,0,0),(0,1,0),(1,1,1)\}$. We will consider these three cases separately:
\begin{itemize}[leftmargin=1.5cm]
    \item[$(0,1,0)$:] Consider $\chi(4,s)$. Since $s(1)=s(3)=0$, the only partitions contributing to $\chi(4,s)$ are $\{\{1,2\},\{3,4\}\}$, $\{\{1,3\},\{2,4\}\}$, $\{\{1,4\},\{2,3\}\}$, and $\{\{1,2,3,4\}\}$. The former three each contribute $(-1)^{2-1} (2-1)! = -1$, and the latter one contributes $s(4)$. Thus $\chi(4,s) = -3 +s(4)\neq 0$.
    \item[$(0,0,0)$:] Set $c=\min\{x>0 \mid s(x)= 1\}$. Note that $c$ must exist since $\{0\} \subsetneq S$, and that $c\geq 4$. Then, clearly, $\chi(c,s)= s(c)\neq 0$.
    \item[$(1,1,1)$:] Set $c=\min\{x>0 \mid s(x)= 0\}$. Note that $c$ must exist since $S\neq \mathbb{N}$, and that $c\geq 4$. Then
    \[\chi(c,s)= \sum_{\sigma < \top_c} (-1)^{|\sigma|-1} (|\sigma|-1)! = \sum_{\sigma} (-1)^{|\sigma|-1} (|\sigma|-1)! - \left((-1)^{|\top_c|-1} (|\top_c|-1)!\right) = 0 -1 = -1 \,,\]
    where the third equation holds again by the properties of the M\"obius function of the partition lattice (see e.g.\ \cite[Section 3.7]{Stanley11}).
\end{itemize}
This shows that $\{s\}$ is of type $\mathbb{T}[\infty]$. Given that $s(0)=1\neq 0$, the lower bounds thus follow immediately from the previous theorem.

Now, for the upper bound, assume that $\mathcal{B}=\{S_1,\dots,S_\ell\}$ for some $\ell>0$, such that none of the $S_i$ satisfies $\{0\}\subsetneq S_i \subsetneq \mathbb{N}$. Equivalently, this means that for each $i\in[\ell]$, either $0\notin S_i$, $S_i=\{0\}$, or $S_i=\mathbb{N}$. For each $i\in [\ell]$, define a signature $s_i$ by setting $s_i(x)=1$ if $x\in S_i$ and $s_i(x)=0$ otherwise. Moreover, let $\mathcal{S}=\{s_1,\dots,s_\ell\}$ and note that $\textsc{Factor}(\mathcal{B}) \fptlinred \textsc{p-UnColHolant}(\mathcal{S})$. If $S_i=\{0\}$ or $S_i=\mathbb{N}$ then, clearly, $s_i(n)=s_i(1)^n$ for all $n>0$. The claim thus follows by Lemma~\ref{lem:linear_type_sigs_equivalence} and the previous theorem.
\end{proof}

\bibliographystyle{plain}
\bibliography{references}

\newpage

\appendix

\section{Generating Signature Sets for each Type}

\begin{lemma}
    There are infinitely many signature sets of each type $\mathbb{T}[\mathsf{Lin}]$, $\mathbb{T}[\omega]$, and $\mathbb{T}[\infty]$. This remains true even if computation is done modulo a prime $p$, that is, if the fingerprint $\chi$ is evaluated modulo $p$, and $s(0)\neq 0 \mod p$ for all signatures.
\end{lemma}
\begin{proof}
We start with a general observation. Let us  write $\top_d$ for the coarsest partition of $[d]$, that is, the partition containing only one block $B=[d]$. In particular, $|\top_d|=1$. Thus
\[ \chi(d,s) = \frac{s(d)}{s(0)} + \sum_{\sigma\neq \top_d} (-1)^{|\sigma|-1} (|\sigma|-1)! \cdot \prod_{B \in \sigma} \frac{s(|B|)}{s(0)}\,. \]
Since each block in $\sigma\neq \top_d$ is of size at most $d-1$, we can enforce $\chi(d,s)=0$ by setting
\[s(d):= -s(0)^{-1} \cdot \sum_{\sigma\neq \top_d} (-1)^{|\sigma|-1} (|\sigma|-1)! \cdot \prod_{B \in \sigma} \frac{s(|B|)}{s(0)}\,.\]
(Recall that we ensured $s(0)\neq 0 \mod p$ if computation is done modulo $p$.)
This easily enables us to generate signatures of each type as follows:
\begin{itemize}
    \item[(1)] For type $\mathbb{T}[\mathsf{Lin}]$, fix any constant $c\in \mathbb{Q}$ and define $s_c(0)=1$, and $s_c(1)=c$. Then, for $d\geq 2$ we recursively define
    \[s_c(d):= -\sum_{\sigma\neq \top_d} (-1)^{|\sigma|-1} (|\sigma|-1)! \cdot \prod_{B \in \sigma} s_c(|B|)\,.\]
    In that way, for each $c$, the signature set $\mathcal{S}=\{s_c\}$ is of type $\mathbb{T}[\mathsf{Lin}]$.
    \item[(2)] For type $\mathbb{T}[\mathsf{\omega}]$, fix any constant $c$ and define $s_c(0)=1$, and $s_c(1)=c$. Moreover, set
    \[s_c(2):= 1-\sum_{\sigma\neq \top_2} (-1)^{|\sigma|-1} (|\sigma|-1)! \cdot \prod_{B \in \sigma} s_c(|B|)\,,\]
    which guarantees $\chi(2,c)=1$, and, finally, define recursively
    \[s_c(d):= -\sum_{\sigma\neq \top_d} (-1)^{|\sigma|-1} (|\sigma|-1)! \cdot \prod_{B \in \sigma} s_c(|B|)\,,\]
    for all $d\geq 3$. Again, it is easy to see that for each $c\in \mathbb{Q}$, the signature set $\mathcal{S}=\{s_c\}$ is then of type $\mathbb{T}[\mathsf{\omega}]$.
    \item[(3)] For the last type $\mathbb{T}[\mathsf{\infty}]$, fix again any constant $c\in \mathbb{Q}$ and define $s_c(0)=1$, and $s_c(1)=s_c(2)=c$. Moreover, set
    \[s_c(3):= 1-\sum_{\sigma\neq \top_3} (-1)^{|\sigma|-1} (|\sigma|-1)! \cdot \prod_{B \in \sigma} s_c(|B|)\,,\]
    and $s_c(d)=0$ for all $d\geq 4$.
    Clearly $\mathcal{S}=\{s_c\}$ is of type $\mathbb{T}[\mathsf{\infty}]$.
\end{itemize}
We conclude the proof by pointing out that our choices for $s_c$ made sure that $s_c(2)=1$ in case (2), and $s_c(3)=1$ in case (3). Thus the lemma holds, as promised, even if the fingerprints are evaluated modulo~$p$.
\end{proof}

\section{Omitted proofs from \Cref{sec:sig0}}\label{sec:appendix_fastMM}
\subsection{Proof of \Cref{lem:listHomsMatrix}}\label{lem:AppendixListHomsMatrix}
\begin{lemma}
Let $\mathcal{H}$ be a class of graphs of treewidth at most 2.
Then $\#\listhomsprob(\mathcal{H})$ can be solved in time $f(|H|)\cdot \mathcal{O}(|V(G)|^{\omega})$ for some computable function $f$. Here, $\omega$ is the matrix multiplication exponent.
\end{lemma}
\begin{proof}
Let $H \in \mathcal{H}$ and $G$ be graphs and let $\mathcal{L} = (L_v)_{v\in V(H)}$ be a collection of sets $L_v \subseteq V(G)$. For what follows, for $X \subseteq V(H)$, let $\mathcal{L}|_{X} = (L_v)_{v \in X}$. 

As shown in \cite{CurticapeanDM17}, $H$ admits a tree decomposition $(T,\beta)$ satisfying what follows.
\begin{enumerate}
    \item The root bag $r$ has size 2 and every other bag has size exactly 3.
    \item The root bag has a unique child $r'$.
    \item For every $t \in V(T)\backslash\{r\}$, we have $|\sigma(t)| = 2$. 
\end{enumerate}

Following \cite{CurticapeanDM17} we consider a numbering of the vertex set $V(H)$ setting $V(H) = \{1,\ldots, n\}$, that respects the structure of the tree decomposition, in the following sense. For $u \in V(H)$ let $t \in T$ be the topmost bag of $u$, that is, $u$ is not contained in any ancestor of $t$. Then, for any vertex $u' \in V(H)\backslash\{u\}$ with topmost bag $t' \in T$ with $t$ being is an ancestor of $t'$, it holds $u < u'$.

By construction, each bag $t\in T\backslash\{r\}$ contains exactly three vertices, which we denote by $u_1(t) < u_2(t) < u_3(t)$. The vertices in $\sigma(t)$ have topmost bags above $t$, while the single vertex in $\beta(t)\setminus \sigma(t)$ has $t$ as its topmost bag. By the numbering introduced above, it follows that $\sigma(t) = \{u_1(t), u_2(t)\}$.

For $t\in T$ and $v_1, v_2 \in V(G)$, let $h_t(v_1,v_2)$ denote the number of homomorphisms in $\homs{H[\gamma(t)]}{G}[\mathcal{L}|_{\gamma(t)}]$ that map $u_1(t)$ to $v_1$ and $u_2(t)$ to $v_2$. For the root $r$, we have $|\beta(r)| = 2$, so summing up the values $h_r(v_1,v_2)$ over all $v_1,v_2 \in V(G)$ evaluates to $\#\homs{H}{G}[\mathcal{L}]$, which can be done in $\mathcal{O}(|V(G)|^2)$ time, assuming random access to the function table $h_r$.

To compute the function tables $h_t, t \in T$ we start from the leaves and then move upwards, as explained next. For every vertex $t \in T$, and for each $ij \in \{12,13,23\}$, let $C^{t}_{ij}$ be the set of children $t'$ of $t$ that satisfy $\sigma(t') = \{u_i(t), u_j(t)\}$. Note that $C^{t}_{12}\;\dot\cup\;C^{t}_{13}\;\dot\cup\;C^{t}_{23}$ is the set of all children of $t$. Further note that after fixing the values for the homomorphisms on $\beta(t)$, the extensions to the cones $\gamma(t')$ are independent for different children $t'$. So, $h_t(v_1, v_2)$ can be computed as follows.
\begin{equation}
    h_{t}(v_1,v_2) = \sum_{v_3\in V(G)}I_{t,v_1,v_2,v_3}\prod_{t'\in C^{t}_{12}}h_{t'}(v_1,v_2)\prod_{t'\in C^{t}_{13}}h_{t'}(v_1,v_3)\prod_{t'\in C^{t}_{23}}h_{t'}(v_2,v_3),
\end{equation}
where $I_{t,v_1,v_2,v_3} \in \{0,1\}$ indicates whether the mapping $\beta(t) \rightarrow \{v_1, v_2, v_3\}$ is a homomorphism in $\homs{H[\beta(t)]}{G}[\mathcal{L}|_{\beta(t)}]$. Note that if $t$ is a leaf, then $h_t(v_1, v_2) = \sum_{v_3 \in V(G)}I_{t, v_1, v_2, v_3}$.
Next, for every $t \in V(T)$, and $ij \in \{12,13,23\}$ we consider three functions $\alpha^t_{ij} : V(G) \times V(G) \rightarrow \mathbb{N}$, as defined below.
\begin{equation}
    \alpha^{t}_{ij}(v,v') = J^{ij}_{t,v,v'}\prod_{t'\in C^{t}_{ij}}h_{t'}(v,v'),
\end{equation}
where $J^{ij}_{t,v,v'} \in \{0,1\}$ indicates whether the mapping $\{u_i(t), u_j(t)\} \rightarrow \{v, v'\}$ is a homomorphism in $\homs{H[\{u_i(t), u_j(t)\}]}{G}[\mathcal{L}|_{\{u_i(t), u_j(t)\}}]$. If $t$ is a leaf, then $\alpha_{ij}^t(v,v') = J^{ij}_{t,v,v'}$. Note that, $I_{t,v_1,v_2,v_3} = J^{12}_{t,v_1,v_2}\cdot J^{13}_{t,v_1,v_3}\cdot J^{23}_{t,v_2,v_3}$. Thus, $h_t(v_1,v_2)$ can be rewritten as follows.

\begin{equation}
    h_{t}(v_1,v_2) = \alpha^{t}_{12}(v_1,v_2)\cdot\sum_{v_3 \in V(G)}(\alpha^{t}_{13}(v_1,v_3)\cdot\alpha^{t}_{23}(v_2,v_3))
\end{equation}

Let $A_{13}, A_{23}$ be two $|V(G)|\times |V(G)|$ matrices where the rows and columns are indexed by $1, \ldots, |V(G)|$ with $(v, v')$ entries equal to $\alpha^t_{13}(v, v')$ and $a_{23}^t(v, v')$ respectively, for all $v, v' \in V(G)$. Then, $h_t(v_1,v_2)$ is equivalently given by the product of $a_{12}^t(v_1,v_2)$ with the $(v_1, v_2)$-th entry of the matrix $A_{13}\cdot A_{23}^T$. The running time is dominated by the matrix multiplication, which yields a total
running time of $f(H)\cdot|V(G)|^
\omega$. Note that the entries of the matrices are nonnegative integers not greater than $|V(G)|^{|V(H)|}$, hence the arithmetic operations are on $\mathcal{O}(|V(H)|\log{|V(G)|})$ bit integers; arithmetic
operations on such integers are in time $f(|V(H)|)$ in the standard word-RAM model with $\log{|V(G)|}$-size
words.  
\end{proof}

\section{Omitted proofs from Section~\ref{sec:uncoloured}}\label{sec:app-uncol}
For the proofs of Lemma~\ref{lem:colEmb_to_colHom_S-coloured} and Lemma~\ref{lem:monotonicity_S-coloured}, it will be very convenient to consider $\mathcal{S}$-vertex-coloured graphs and signature grids over $\mathcal{S}$ as relational structures with unary predicates: For this section, fix a finite set of signatures $\mathcal{S}=\{s_1,\dots,s_\ell\}$. 

We define a \emph{vocabulary} $\tau=(E,P_1,\dots,P_\ell)$, where $E$ is a binary relation symbol, and each $P_i$ is a unary relation symbol. A $\tau$-\emph{structure} $\mathcal{H}$ consists of a finite set of vertices $V=V(\mathcal{H})$, a binary relation $E^\mathcal{H}$ over $V$, and unary relations $P_1^\mathcal{H},\dots,P_\ell^{\mathcal{H}}$ over $V$.
A \emph{homomorphism} from a $\tau$-structure $\mathcal{H}$ to a $\tau$-structure $\mathcal{G}$ is a mapping $\varphi: V(\mathcal{H}) \to V(\mathcal{G})$ such that
\begin{itemize}
    \item for all $(u,v)\in E^\mathcal{H}$ we have $(\varphi(u),\varphi(v))\in E^\mathcal{G}$, and
    \item for all $i\in[\ell]$ and for all $v\in P_i^\mathcal{H}$ we have $\varphi(v)\in P_i^\mathcal{G}$.
\end{itemize}

An \emph{embedding} from $\mathcal{H}$ to $\mathcal{G}$ is an injective homomorphism from~$\mathcal{H}$ to~$\mathcal{G}$, and an \emph{isomorphism} from~$\mathcal{H}$ to~$\mathcal{G}$ is a bijection $\iota:V(\mathcal{H})\to V(\mathcal{G})$ such that for all $u,v\in V(\mathcal{H})$ we have $(u,v)\in E^\mathcal{H} \Leftrightarrow (\iota(u),\iota(v))\in E^\mathcal{G}$ and, for each $i\in[\ell]$, $v\in P_i^\mathcal{H}\Leftrightarrow \iota(v)\in P_i^\mathcal{G}$. We write $\mathcal{H}\cong\mathcal{G}$ if an isomorphism exists. Finally, an \emph{automorphism} of $\mathcal{H}$ is an isomorphism from $\mathcal{H}$ to itself. We write $\homs{\mathcal{H}}{\mathcal{G}}$, $\embs{\mathcal{H}}{\mathcal{G}}$, and $\auts(\mathcal{H})$ for the sets of homomorphisms and embeddings from $\mathcal{H}$ to $\mathcal{G}$, and for the set of automorphisms of $\mathcal{H}$, respectively.

Now it is easy to see that both $\mathcal{S}$-vertex-coloured graph $(H,\xi)$ and signature grids $\Omega=(G,\{s_v\}_{v\in V(G)})$ over $\mathcal{S}$ correspond to $\tau$-structures where the binary relation symbol $E$ corresponds to the edges, and a vertex equipped with/coloured by signature $s_i$ is contained in the unary relation $P_i$. Moreover, it is also easy to see that the notions of homomorphisms, embeddings, isomorphisms, and automorphisms are identical when we them as $\tau$-structures.

We are now able to prove Lemma~\ref{lem:colEmb_to_colHom_S-coloured}, which we restate for convenience.
\begin{lemma}[Lemma~\ref{lem:colEmb_to_colHom_S-coloured}, restated]
      Let $\mathcal{S}$ be a finite set of signatures, let $(H,\nu)$ be an $\mathcal{S}$-vertex-coloured graph, and let $\Omega=(G,\{s_v\}_{v\in V(G)})$ be a signature grid over $\mathcal{S}$, we have
    \[ \#\embs{(H,\nu)}{\Omega} = \sum_{\rho \in \mathsf{colPart}(H)} \mu(\bot,\rho)\cdot \#\homs{(H,\nu)/\rho}{\Omega}\,,\]
    where $\mu(\bot,\rho)=\prod_{B\in \rho}(-1)^{|B|-1}(|B|-1)!$ is the (usual) M\"obius function of partitions.
\end{lemma}
\begin{proof}
    Let $\mathcal{H}$ and $\mathcal{G}$ be the $\tau$-structures corresponding to $(H,\nu)$ and $\Omega$, respectively.

    Let $\rho$ be a (not necessarily colour-consistent) partition of $V(H) (=V(\mathcal{H}))$. The \emph{quotient structure} $\mathcal{H}/\rho$ is defined as follows: The vertices are the blocks of $\rho$, we include $(B_1,B_2)$ in $E^{\mathcal{H}/\rho}$ if there are $u\in B_1$ and $v\in B_2$ such that $(u,v)\in E^\mathcal{H}$, and, for each $i\in[\ell]$ we include $B \in P^{\mathcal{H}/\rho}_i$ if there is a vertex $v\in B$ with $v\in P_i^\mathcal{H}$. Observe that, for colour-consistent $\rho$, the $\tau$-structure $\mathcal{H}/\rho$ corresponds precisely to $(H,\nu)/\rho$.

    Now, verbatim to the case of graphs (see~\cite[(5.18)]{Lovasz12}), we use M\"obius inversion over the lattice of (all) partitions of $V(\mathcal{H})$ and obtain:
    \[\#\embs{\mathcal{H}}{\mathcal{G}}=\sum_{\rho} \mu(\bot,\rho) \cdot \#\homs{\mathcal{H}/\rho}{\mathcal{G}} \,,\]
    where the sum is over all partitions of $V(\mathcal{H})$.

    Finally, observe that $\#\homs{\mathcal{H}/\rho}{\mathcal{G}}=0$ if $\rho$ is not colour-consistent: If there is a block $B\in \rho$ containing vertices $u,v$ with $u\in P_i$ and $v\in P_j$ (i.e., $\nu(u)=s_i$ and $\nu(v)=s_j$) for $i\neq j$, then $B$ is contained in $P_i^{\mathcal{H}/\rho}\cap P_j^{\mathcal{H}/\rho}$. However, no vertex of $\mathcal{G}$ is contained in $P_i^{\mathcal{G}}\cap P_j^{\mathcal{G}}$ since each vertex of $G$ is only assigned one signature. Thus no homomorphism can exist. This concludes the proof since, for colour-consistent $\rho$, we clearly have $\#\homs{\mathcal{H}/\rho}{\mathcal{G}}=\#\homs{(H,\nu)/\rho}{\Omega}$.
\end{proof}

We continue with proving Lemma~\ref{lem:monotonicity_S-coloured}, which we also restate for convenience.

\begin{lemma}[Lemma \ref{lem:monotonicity_S-coloured}, restated]
    Let $\mathcal{S}$ be a finite set of signatures and let $\mathcal{C}$ be the class of all graphs $H$ for which there is a positive integer $k$ and a colouring $\nu:V(H)\to \mathcal{S}$ such that $\zeta_{\mathcal{S},k}(H,\nu)\neq 0$.
    \begin{itemize}
        \item[(1)] If all graphs in $\mathcal{C}$ are acyclic, then  $\text{\sc{p-UnColHolant}}(\mathcal{S})$ can be solved in FPT-near-linear time.
        \item[(2)] If $\mathcal{C}$ has unbounded treewidth, then $\text{\sc{p-UnColHolant}}(\mathcal{S})$ is $\#\W[1]$-complete.\qed
    \end{itemize}
\end{lemma}
\begin{proof}
   First, recall that, by Lemma~\ref{lem:uncoloured_hombasis_general}, we have
   \begin{equation}\label{eq:appendix_recall_colhombasis}
       \mathsf{UnColHolant}(\Omega,k) = \prod_{i\in \ell}s_i(0)^{n_i} \cdot \sum_{(H,\nu)\in \mathcal{G}(\mathcal{S})} \zeta_{\mathcal{S},k}(H,\nu) \cdot \#\homs{(H,\nu)}{\Omega}\,.
   \end{equation}
    We now continue by proving both cases separately.
    \begin{itemize}
        \item[(1)] If all graphs in $\mathcal{C}$ are acyclic, then each term  $\#\homs{(H,\nu)}{\Omega}$ surviving in~(\ref{eq:appendix_recall_colhombasis}) satisfies that $H$ is acyclic, in which case the cardinality $\#\homs{(H,\nu)}{\Omega}$ can be computed in FPT-near-linear time in $|\Omega|$ (w.r.t.\ parameter $|(H,\nu)|$ which only depends on $k$). This follows from the fact that we can equivalently express $\#\homs{(H,\nu)}{\Omega}$ as $\#\homs{\mathcal{H}}{\mathcal{G}}$ for the $\tau$-structures $\mathcal{H}$ and $\mathcal{G}$ corresponding, respectively, to $(H,\nu)$ and $\Omega$. This is an instance of the problem of counting answers to acyclic conjunctive without quantified variables, which is well-known to be solvable in (FPT)-near-linear time (see e.g.\ \cite[Theorem 12]{BraultBaron13}). Finally, since we can compute each surviving term in FPT-near-linear time, and the number of terms only depends on $k$, we can evaluate the entire linear combination in FPT-near-linear time.
        \item[(2)] Note that the factor $\prod_{i\in \ell}s_i(0)^{n_i}$ in~\eqref{eq:appendix_recall_colhombasis} can trivially be computed in FPT time. Thus, $\text{\sc{p-UnColHolant}}(\mathcal{S})$ is interreducible to the parameterised problem that, on input $\Omega$ and $k$, outputs
        \[\sum_{(H,\nu)\in \mathcal{G}(\mathcal{S})} \zeta_{\mathcal{S},k}(H,\nu) \cdot \#\homs{(H,\nu)}{\Omega}\,.\]
        In other words, this problem is equivalent to computing linear combinations of homomorphism counts between vertex-coloured graphs. As shown by Curticapean, Dell, and Marx~\cite[Lemma 3.8 and Remark 3.9]{CurticapeanDM17}, this problem is $\#\W[1]$-hard whenever there is no constant upper bound on the treewidth of graphs $(H,\nu)$ with coefficient $\zeta_{\mathcal{S},k}(H,\nu)\neq 0$. Since $\mathcal{C}$ has unbounded treewidth, our proof is concluded.
    \end{itemize}
\end{proof}

\end{document}